\documentclass[12pt,reqno]{amsart}
\makeatletter
\def\l@subsection{\@tocline{2}{0pt}{2.5pc}{5pc}{}}
\def\l@subsubsection{\@tocline{2}{0pt}{5pc}{7.5pc}{}}
\makeatother
\usepackage{amsbsy}
\usepackage{amsmath}
\usepackage{amsthm}
\usepackage{amssymb}
\usepackage{mathrsfs}
\usepackage{amsfonts}
\usepackage{newlfont}
\usepackage[sans]{dsfont}
\usepackage{dcolumn}
\usepackage{bm}
\usepackage{stmaryrd}
\usepackage{bbm}
\usepackage{relsize}
\usepackage{euscript}
\usepackage{eufrak}
\usepackage[margin=1.0in]{geometry}
\numberwithin{equation}{section}
\newtheorem{thm}{Theorem}[section]

\newtheorem{cor}[thm]{Corollary}
\newtheorem{lem}[thm]{Lemma}
\newtheorem{prop}[thm]{Proposition}
\newtheorem{defn}[thm]{Definition}
\newtheorem{rem}[thm]{Remark}
\newtheorem{exam}[thm]{Example}
\begin{document}
\allowdisplaybreaks{
\title[]{HARMONIC FUNCTIONS AND LINEAR ELLIPTIC DIRICHLET PROBLEMS WITH RANDOM BOUNDARY VALUES:~STOCHASTIC EXTENSIONS OF SOME CLASSICAL THEOREMS AND ESTIMATES}
\author{Steven D. Miller}
\email{stevendm@ed-alumni.net}
\address{}
\maketitle
\begin{abstract}
Let $\psi:\rightarrow\bm{\mathrm{R}}$ be a harmonic function such that $\Delta \psi(x)=0$, for all $x\in\bm{\mathcal{D}}\subset\bm{\mathrm{R}}^{n}$. There are then many well-established classical results:the Dirichlet problem and Poisson formula, Harnack inequality, the Maximum Principle, the Mean Value Property, the Cacciopolli estimate etc. Here, a 'noisy' or random domain is one for which there also exists a classical scalar Gaussian random field (SGRF)$\mathscr{J}(x)$ defined for all $x\in{\mathcal{D}}\cup\partial\bm{\mathcal{D}}$ or $x\in\partial\bm{\mathcal{D}}$ only, with respect to a probability space $[\bm{\Omega},\bm{\mathcal{F}},{\mathrm{I\!P}}]$. The SGRF has vanishing mean value $\mathbf{E}\llbracket\mathscr{J}(x)\rrbracket=0$
and a regulated covariance $\mathbf{E}\big\llbracket\mathscr{J}(x)\otimes\mathscr{J}(y)
\big\rrbracket=\alpha J(x,y;\xi)$ for all $(x,y)\in\bm{\mathcal{D}}$, with correlation length $\xi$ and ${\mathbf{E}}\llbracket{\mathscr{J}(x)}\otimes\mathscr{J}(x)\rrbracket
=\alpha<\infty$. The gradient $\nabla_{i}\mathscr{J}$ and integrals
$\int_{\bm{\mathcal{D}}}\mathscr{J}(x)d\mu(x)$ also exist on $\bm{\mathcal{D}}$. Harmonic functions and potentials can become randomly perturbed SGRFs of the form $\overline{\psi(x)}=\psi(x)+\lambda\mathscr{J}(x)$. Physically, this scenario could arise from noisy sources or random fluctuations in mass or charge density; noisy boundary/surface data fluctuations; and introducing turbulence or randomness into smooth potential fluid flows, steady state diffusions or heat flow. This leads to stochastic modifications of classical theorems and estimates for randomly perturbed harmonic functions and Riesz and Newtonian potentials, and to stability estimates for the growth and decay of their volatility, covariances and higher-order moments.
\end{abstract}
\maketitle
\tableofcontents
\raggedbottom
\section{introduction}
This paper tentatively explores some ideas at the interface of probability theory and stochastic analysis with the theory of harmonic functions and elliptic PDEs with random field Dirichlet boundary values. There have been long-standing and efficacious interactions between probability theory, stochastic analysis, geometry and PDE, perhaps best exemplified by random motion and stochastic analysis on manifolds [1-7]. The well-known Feynman-Kac formula connects elliptic and parabolic PDE theory with Brownian motion, and these type of PDEs in turn play crucial roles in geometric analysis $\bm{[8-11]}$.

The Laplace and Poisson equations are the simplest forms of elliptic PDEs and the study of their solutions forms the basis of potential theory $\bm{[12-18]}$. Solutions of the Laplace equation are the harmonic functions and these have a well-established and important place within classical mathematical physics: in electromagnetic theory, Newtonian gravitation and astrophysics and fluid mechanics.[1-4]. Electric, gravitational and fluid potentials are described in certain circumstances by solutions of the Laplace equation and are therefore harmonic functions $\bm{[8-11]}$ Steady state fluid flow which is irrotational and incompressible is described by a Laplace equation as is steady state heat flow or chemical diffusion [1-4].

If a function $\psi:\bm{\mathcal{D}}\rightarrow \mathbf{R}^{n}$ is harmonic on some domain $\bm{\mathcal{D}}\subset\bm{\mathrm{R}}^{n}$ then it satisfies a number of well-established classical results, estimates, inequalities and theorems. These include the Harnack inequality, Poisson formula and Liouville Theorem, Dirichlet energy critical points, the Maximum Principle, the Mean Value Property, the Gradient Estimate, the Bochner formula etc. $\bm{[12-18]}$. These are briefly reviewed in Section 2. In Section 3, classical results are extended and modified to accommodate 'random domains' for which there exists a Gaussian random scalar field (GRSF) at all points within the domain, or for 'noisy data' on the boundary. Deterministic harmonic functions, fields or solutions are then randomly perturbed by this noise or random boundary data. Physically, this can correspond to random fluctuations due to noise or perturbations in mass or charge density sources, fluctuations in chemical concentrations or to 'turbulence' in the potential theory of fluids.

Classical random fields correspond naturally to structures, and properties of systems, that
are varying randomly in time and/or space. They have found many useful applications in
mathematics and applied science: in the statistical theory or turbulence, in geoscience,
medical science, engineering, imaging, computer graphics, statistical mechanics and
statistics, biology and cosmology $\bm{[19-33]}$. Coupling random fields or noise to ODEs or PDE also a useful methodology in studying turbulence, chaos, random systems, pattern formation etc. $\bm{[19-33,]}$. The study of stochastic partial differential equations (SPDEs),arise from the coupling of random fields/noises to PDEs is also a rapidly growing area $\bm{[34-39]}$. Such SPDES can model the propagation of heat, diffusions or waves in random medias or randomly fluctuating medias. Many dynamical systems are affected or influenced by colored noise (which is regulated.)

An important class of random fields are the Gaussian random scalar fields (GRSFS) which a
characterized only by the first and second moments. The GRSFS can also be isotropic,
homogenous and stationary. The details are made more precise in Appendix A, but the
advantages of GRVFs are:
\begin{enumerate}
\item GRVFS have convenient mathematical properties which generally simplify
    calculations;indeed, many results can only be evaluated using Gaussian fields.
\item A GRVF can be classified purely by its first and second moments, and all high-order
    moments and cumulants can be ignored. If $\mathscr{J}(x)$ is a time-independent GRSF
    existing for all
    $x\in\bm{\mathrm{R}}^{n}$ or $x\in{\mathcal{D}}\subset \bm{\mathrm{R}}^{n}$ then
    ${\mathbf{E}}\llbracket\mathscr{J}(x)\rrbracket=0$ and
    \begin{align}
    {\mathbf{E}}\llbracket\mathscr{J}(x)\otimes\mathscr{J}(y)\rrbracket=\zeta
    J(x,y;\ell)
    \end{align}
    with a correlation length $\ell$. It is regulated if
    ${\mathbf{E}}\llbracket\mathscr{J}(x)\otimes\mathscr{J}(x)\rrbracket=\zeta<\infty$
    For a white-in-space Gaussian noise or random field
    ${\mathbf{E}}\llbracket\mathscr{J}(x)\otimes\mathscr{J}(y)\rrbracket
    =\zeta\delta^{n}(x-y)$ and is unregulated. This paper utilises only regulated GRSFs.
\item Gaussian fields accurately describe many natural stochastic processes including
    Brownian motion.
\item A large superposition of non-Gaussian fields can approach a Gaussian field.
\end{enumerate}
The following generic stochastic problems can arise for linear, quasi-linear and nonlinear
PDEs subject to randomness or noise.
\begin{prop}
Let $\bm{\mathcal{D}}\subset\mathbf{R}^{n}$ be a domain or (compact) set with bound
$\partial\bm{\mathcal{D}}$ and let ${\mathbf{L}}$ be a linear differential operator
(usually elliptic)acting on a function or field $u(x,t)$ defined on $\bm{\mathcal{D}}\times
[0,T]$ or $\bm{\mathcal{D}}\times\mathbf{R}^{+}$. The generic Cauchy initial value
problem (CIVP) and boundary value problem (BVP) is then
\begin{align}
&{\mathbf{L}}u(x,t)=f(x,t),~(x,t)\in\bm{\mathcal{D}}\times[0,T],(x,t)\in\bm{\mathcal{D}}
u(x,0)=\phi(x),~x\in\bm{\mathcal{D}},t=0\\&
{\mathscr{B}}u(x,t)=\beta(x),~x\in\partial\bm{\mathcal{D}}
\end{align}
where $\bm{\mathfrak{B}}$ is a 'boundary operator', $f(x,t)$ is a source term, $\beta(x)$
some function and $\phi(x)$ is the initial Cauchy data with $\phi\in
C^{\infty}(\bm{\mathcal{D}})$ and at least $C^{2}(\bm{\mathcal{D}})$. Similarly, one can
have a quasi-linear operator $\mathbf{Q}$ and a nonlinear operator
$\pmb{\mathscr{N}}$. Introducing randomness, fluctuations or noise via random scalar field $\mathscr{J}(x)$ leads to the following possibilities for generic random problems:
\begin{enumerate}
\item The random or stochastic CIVP with initial data $\phi(x)$ having additive randomness
    random perturbations or fluctuations such that
    \begin{align}
    \widehat{\phi(x)}=\phi(x)+\mathscr{J}(x)\nonumber
    \end{align}
    where $\mathscr{J}(x)$ is a (Gaussian) random scalar field (GRSF).(Appendix A.) For
    multiplicative random perturbations
    \begin{align}
    \widehat{\phi(x)}=\phi(x)\otimes\mathscr{J}(x)\nonumber
    \end{align}
\item The random boundary value problem with $\beta(x)$ having randomness, random
    perturbations or fluctuations such that $\widehat{\beta(x)}=\beta(x)+\mathscr{J}(x)$.
\item The random source problem such that the source term $f(x,t)$ is subject to
    randomness, random perturbations or noise such that
    $\widehat{f(x,t)}=f(x,t)+\mathscr{J}(x,t)$, where $\mathscr{J}(x,t)$ is a random field
    or noise in space and time; for example, a white noise or Wiener process. This is the
    usual approach to studying stochastic PDEs.(refs.)
\item The random or stochastic geometry problem:the domain $\bm{\mathcal{D}}$ is random.
\item Some combination of these conditions.
\end{enumerate}
\end{prop}
In this paper, the cases (1), (2) and (3) will be tentatively developed  for elliptic PDEs, namely the Laplace and Poisson equations.
\section{Harmonic functions and Dirichlet Boundary Value problems}
In this section, properties of harmonic functions and Dirichlet boundary value problems and some classical results, are briefly reviewed. We begin with the following definitions [1-5].
\begin{defn}
An open connected set $\bm{\mathcal{D}}\subset{\mathbf{R}}^{n}$ is a domain. The closure of $\bm{\mathcal{D}}$ is denoted $\overline{\bm{\mathcal{D}}}$ and $\partial\bm{\mathcal{D}}$ is the boundary. A domain $\bm{\mathcal{D}}^{\prime}\subset {\mathbf{R}}^{n}$ is a strictly interior subdomain if $\bm{\mathcal{D}}^{\prime}\subset\bm{\mathcal{D}}\subset {\mathbf{R}}^{n}$. Then $\bm{\mathcal{D}}^{\prime}\subset\subset\bm{\mathcal{D}}$ if $\overline{\bm{\mathcal{D}}^{\prime}}\subset\bm{\mathcal{D}}$.
\end{defn}
\begin{defn}
The following notation is utilised. $x=(x_{1},x_{2},...,x_{n})\subset\mathbf{R}^{n}$
or $\mathbf{R}^{n}=\lbrace(x_{1}...x_{n}|x_{i}\in{\mathbf{R}}\rbrace$. The gradient of $\psi(x)$ is
\begin{equation}
\nabla_{i}\psi(x)=\frac{\partial\psi(x)}{\partial x_{i}}\equiv\partial_{i}\psi(x)
\end{equation}
or
\begin{equation}
\nabla\psi(x)\equiv \sum_{i=1}^{n}\nabla_{i}\psi(x)=\sum_{i=1}^{n}\frac{\partial\psi(x)}{\partial x_{i}}\equiv\sum_{i=1}^{n}\nabla_{i}\partial_{i}\psi(x)
\end{equation}
for any smooth function $\psi:\bm{\mathcal{D}}\rightarrow {\mathbf{R}}^{n}$. (Usually, the summation is dropped.)
\end{defn}
\begin{defn}
Given $\bm{\mathcal{D}}$ and $\psi\rightarrow{\mathbf{R}}^{n}$ the sets $L_{p}(\bm{\mathcal{D}}),\mathcal{L}_{p}(\bm{\mathcal{D}})$ and $\mathscr{L}_{p}(\bm{\mathcal{D}})$ are defined with respect to the following norms
\begin{enumerate}
\item $L_{p}(\bm{\mathcal{D}})$ is the set of all functions $\psi(x)$ in $\bm{\mathcal{D}}$ such that the norm
\begin{equation}
\|\psi(x)\|_{L_{p}}=\bigg(\int_{\bm{\mathcal{D}}}|\psi(x)|^{p}d^{n}x\bigg)^{1/p}<\infty
\end{equation}
is finite.
\item $\mathcal{L}_{p}(\bm{\mathcal{D}})$ is the set of all functions $\psi(x)$ in $\bm{\mathcal{D}}$ such that the norm
\begin{equation}
\|\psi(x)\|_{\mathcal{L}_{p}}=\bigg(\sum_{i=1}^{n}|\nabla_{i}\psi(x)|^{p}\bigg)^{1/p}
\equiv\left(|\nabla_{i}\psi(x)|^{p}\right)^{1/p}<\infty
\end{equation}
\item $\mathscr{L}_{p}(\bm{\mathcal{D}})$ is the set of all functions $\psi(x)$ in $\bm{\mathcal{D}}$ such that the norm
\begin{equation}
\|\psi(x)\|_{\mathscr{L}_{p}}=\left(\int_{\bm{\mathcal{D}}}
\sum_{i=1}^{n}|\nabla_{i}\psi(x)|^{p}d^{n}x\right)^{1/p}
\equiv<\infty
\end{equation}
$\mathscr{L}_{P}(\bm{\mathcal{D}}$ forms a Banach space.
\end{enumerate}
\end{defn}
\begin{defn}
Let $\bm{\mathcal{D}}\subset {\mathbf{R}}^{N}$ be a domain or open subset with boundary $\partial\bm{\mathcal{D}}$, where ${\mathbf{R}}^{n}=\lbrace(x_{1}...x_{n})|x_{i}\in{\mathbf{R}}\rbrace $ and $\psi:\bm{\mathcal{D}}\rightarrow{\mathbf{R}}^{n}$ is a smooth function. Then the Laplacian is
\begin{equation}
\Delta\psi(x)=\sum_{i=1}^{n}\frac{\partial^{2}\psi(x)}{\partial x_{i}^{2}}\equiv \sum_{i=1}^{n}\nabla_{i}\nabla^{i}\psi(x)\equiv \sum_{i=1}^{n}div(\nabla_{i}\psi(x))
\end{equation}
For $(f,g):\bm{\mathcal{D}}\rightarrow {\mathbf{R}}^{n}$, the Poisson equation with Dirichlet boundary conditions is
\begin{align}
&\Delta \psi(x)=f(x),~x\in\bm{\mathcal{D}}\\&
\psi(x)=g(x),~x\in\partial\bm{\mathcal{D}}
\end{align}
The Laplace equation is then
\begin{align}
&\Delta \psi(x)= 0,~x\in\bm{\mathcal{D}}\\&
\psi(x)=g(x),~x\in\partial\bm{\mathcal{D}}
\end{align}
or simply $\Delta\psi(x)= 0,~x\in\bm{\mathcal{D}}$. All solutions to $\Delta\psi(x)= 0$ are harmonic.
\end{defn}
\begin{defn}
Given a function $\psi(x)$ in $\bm{\mathcal{D}}\subset\mathbf{R}^{n}$, the Dirichlet energy integral is
\begin{align}
&\mathlarger{\mathcal{E}}[\psi]=\frac{1}{2}\int_{\bm{\mathcal{D}}}|\nabla_{i}\psi(x)|^{2}d^{n}x
\equiv\frac{1}{2}\int_{\bm{\mathcal{D}}}\sum_{i=1}^{n}|\nabla_{i}\psi(x)|^{2}d^{n}x\nonumber\\&
\equiv\frac{1}{2}\int...\int_{\bm{\mathcal{D}}}\left[\left(\frac{\partial\psi(x)}{\partial x_{1}}\right)^{2}+...+\left(\frac{\partial\psi(x)}{\partial x_{n}}\right)^{2}\right]dx_{1}...dx_{n}
\end{align}
The Laplace equation $\Delta \psi(x)=0 $ can then be interpreted as the Euler-Lagrange equation for the integral (2.11). Using (1.4), $\mathcal{E}[\psi]=(\|\psi(x)\|_{\mathscr{L}_{p}})^{p}$
\end{defn}
\begin{defn}
If the integral (2.11) is redefined with respect to the $p^{th}$ power then
\begin{align}
&\mathlarger{\mathcal{E}}[\psi]=\frac{1}{2}\int_{\bm{\mathcal{D}}}|\nabla_{i}\psi(x)|^{p}d^{n}x
\equiv\frac{1}{2}\int_{\bm{\mathcal{D}}}\sum_{i=1}^{n}|\nabla_{i}\psi(x)|^{p}d^{n}x\nonumber\\&
\equiv\int...\int_{\bm{\mathcal{D}}}\left[\left(\frac{\partial\psi(x)}{\partial x_{1}}\right)+...+\left(\frac{\partial\psi(x)}{\partial x_{n}}\right)\right]^{p/2}dx_{1}...dx_{n}
\end{align}
then the corresponding Euler-Lagrange equation is the p-Laplace equation
\begin{equation}
\Delta^{(p)}\psi(x)=div(|\nabla\psi(x)|^{p-2}\nabla\psi(x))\equiv
\sum_{i=1}^{n}\nabla^{i}(|\nabla_{i}\psi(x)|^{p-2}\nabla_{j}\psi(x))=0
\end{equation}
where $|\nabla\psi(x)|^{p-2}$ is defined as
\begin{equation}
(\|\psi(x)\|_{\mathcal{L}_{2}})^{p-2}\equiv|\nabla\psi(x)|^{p-2}=\left(\sum_{i=1}^{n}|\nabla_{i}\psi(x)|^{2}\right)^{(p-2)/2}
\end{equation}
The full p-Laplace equation is
\begin{align}
&\mathlarger{\Delta}^{(p)}\psi(x)=\left(\sum_{i=1}^{n}|\mathlarger{\nabla}_{i}\psi(x)|^{2}\right)^{(p-4)/2}\nonumber\\&\times\left\lbrace
\sum_{i=1}^{n}|\nabla_{i}\psi(x)|^{2})\Delta \psi(x))+(p-2)\sum_{ij}\nabla_{i}\psi(x)\nabla_{j}\psi(x)\nabla_{i}\nabla_{j}\psi(x)\right\rbrace=0
\end{align}
\end{defn}
The theory of harmonic functions can be extended to p-Laplace operators and general elliptic operators [6]. In this paper, only the standard Laplacian on Euclidean space $\mathbf{R}^{n}$ is considered.
\subsection{Brief historical background}
The origins of potential theory can be traced back to Newton's Principia $\bm{[40]}$. Newton's universal law of gravitation (1687) asserts that given a point mass $M$ and $x\in\mathbf{R}^{n}$ and a system of N point masses $m_{a}$ at $y_{a}\in\mathbf{R}^{3}$ for $a=1...N$, then the force exerted on M at x is
\begin{equation}
F(x)=\sum_{a=1}^{N}\frac{Cm_{i}M}{|x-y_{a}|^{2}}\frac{(x-y_{a})}{|x-y_{a}|}
\end{equation}
with $C<0$ since like masses attract. The same law of inverse squares was also experimentally established by Coulomb (1785) for a system of like point charges $(q_{1}....q_{N})$ interacting with a point charge Q so that
\begin{equation}
F=\sum_{a=1}^{N}\frac{Cq_{i}Q}{|x-y_{a}|^{2}} \frac{x-y_{a}}{|x-y_{a}|}
\end{equation}
but now with $C>0$ since like charges repel. Bernoulli (1748) introduced the function
\begin{equation}
\psi(x)=-\sum_{a=1}^{N}\frac{C m_{i}}{|x-y_{a}|}
\end{equation}
and it was subsequently observed by Lagrange (1773) that
\begin{equation}
F(x)=M\nabla \psi(x),~~x \ne y
\end{equation}
The function $\psi(x)$, named the 'potential function' by Green and the 'potential' by Gauss, completely describes a gravitational or electric field. For a distribution of masses and charges with density $\rho(x)$, vanishing outside some bounded domain $\bm{\mathcal{D}}\subset\mathbf{R}^{3}$, the potential becomes
\begin{equation}
\psi(x)=C\int_{\bm{\mathcal{D}}}\frac{\rho(y)d^{3}y}{|x-y|}
\end{equation}
Laplace in 1782 then observed that the potential satisfies the PDE
\begin{equation}
\Delta \psi(x)=0,~~outside~\bm{\mathcal{D}}
\end{equation}
Laplace's eponymous equation was discussed by him in several papers in the 1780s and made famous in his treatise $\mathit{Mecanique~Celeste}$ $\bm{[41]}$. Its solutions are the harmonic functions. Lagrange had also considered the same equation in 1760 in connection with fluid flow. Poisson (who was a student of Laplace) completed the result around 1813 when he demonstrated that
\begin{equation}
\Delta \psi(x)=-4\pi C\rho(x),~~within~\bm{\mathcal{D}}
\end{equation}
for smooth densities $\rho(x)$. The integral (2.20) can be split into two parts. Define a small ball $\bm{\mathcal{B}}_{\epsilon}(x)=(y\in\mathbf{R}^{3}:|x-y|<\epsilon)$ of radius $\epsilon$ with centre $x$ and $\bm{\mathcal{B}}_{\epsilon}(x)\subset
\bm{\mathcal{D}}$. Then the compliment is $\bm{\mathcal{D}}\symbol{92}\bm{\mathcal{B}}_{\epsilon}(x)$ so that
\begin{equation}
\bm{\mathcal{D}}={\bm{\mathcal{B}}}_{\epsilon}(x)\bigcup\bm{\mathcal{D}}
\symbol{92}{\bm{\mathcal{B}}}_{\epsilon}(x)\nonumber
\end{equation}
Then the integral (1.20) can be split as
\begin{equation}
\psi(x)=C\int_{\bm{\mathcal{B}}_{\epsilon}(x)}\frac{\rho(y)d^{n}y}{|x-y|}
+C\int_{\bm{\mathcal{D}}~\symbol{92}\bm{\mathcal{B}}_{\epsilon}(x)}\frac{\rho(y)d^{3}y}{|x-y|}
\end{equation}
By the Laplace result, the integral over $\bm{\mathcal{D}}\symbol{92}{\bm{\mathcal{B}}}_{\epsilon}(x)$ is harmonic and so vanishes under the Laplace operator so that
\begin{equation}
\Delta_{x}\psi(x)=C\Delta_{x}\int_{\bm{\mathcal{B}}_{\epsilon}(x)}\frac{\rho(y)d^{3}y}{|x-y|}
+C\Delta_{x}\int_{bm{\mathcal{D}}~\symbol{92}\bm{\mathcal{B}}_{\epsilon}(x)}\frac{\rho(y)d^{3}y}{|x-y|}=C\Delta_{x}
\int_{\bm{\mathcal{B}}_{\epsilon}(x)}\frac{\rho(y)d^{n}y}{|x-y|}=0
\end{equation}
The remaining integral can then be manipulated as
\begin{equation}
C\Delta_{x}\int_{\bm{\mathcal{B}}_{\epsilon}(x)}\frac{\rho(y)d^{n}y}{|x-y|}
=C\int_{\bm{\mathcal{B}}_{\epsilon}(x)}(\rho(y)-\rho(x)\Delta_{x}\frac{1}{|x-y|}d^{3}y
+C\rho(x)\int_{\bm{\mathcal{B}}_{\epsilon}(x)}\Delta_{x}\frac{1}{|x-y|}d^{3}y
\end{equation}
treating $1/|x-y|$ as a smooth function. Applying Gauss' divergence theorem to the second integral on the rhs
\begin{align}
&\int_{\bm{\mathcal{B}}_{\epsilon}(x)}\Delta_{x}\frac{1}{|x-y|}d^{3}y
=\int_{\partial\bm{\mathcal{B}}_{\epsilon}(x)}\nabla_{N(x)}\frac{1}{|x-y|}d^{2}y
\nonumber\\&=\int_{0}^{\epsilon}\frac{d}{dr}\frac{1}{r}dS|_{r=\epsilon|}
=+4\pi\epsilon^{2}\frac{d}{dr}\frac{1}{r}|_{r=\epsilon|}=-4\pi
\end{align}
where $\partial\bm{\mathcal{B}}_{\epsilon}(x)$ is the surface or boundary of the ball and $\nabla_{N(x)}$ is the normal derivative at y with respect to x. Then
\begin{equation}
\Delta_{x}\psi(x)=C\int_{\bm{\mathcal{B}}_{\epsilon}(x)}(\rho(y)-\rho(x))\Delta_{x}\frac{1}{|x-y|}d^{3}y=4\pi C\rho(x)
\end{equation}
Since $\Delta_{x}(1/|x-y|)=\delta^{3}(x-y)$ except at $x=y$ then $1/|x-y|$ is harmonic everywhere except at the blowup at $x=y$. Also if $(\rho(y)-\rho(x))\rightarrow 0$  sufficiently fast as $|x-y|\rightarrow 0$ then the integral should vanish. This is related to $\rho(x)$ having Holder continuity $|\rho(y)-\rho(x)|<C|y-x|$.  Also
\begin{equation}
C\int_{\bm{\mathcal{B}}_{\epsilon}(x)}(\rho(y)-\rho(x))\delta^{3}(x-y)d^{3}y=C(\rho(y)-\rho(x))=0
\end{equation}
For gravitation, Gauss showed that the flux of gravitational field through a spherical surface $\partial\bm{\mathcal{B}}$ is $-4\pi G M_{{\mathcal{B}}}$, where the surface is the boundary of $\bm{\mathcal{D}}\subset{\mathbf{R}}^{3}$ which encloses a mass density $\rho(x)$ for all $x\in\bm{\mathcal{D}}$. The total mass contained within $\bm{\mathcal{D}}$ is then
\begin{equation}
M_{\bm{\mathcal{D}}}=\int_{\bm{\mathcal{D}}}\rho(x)d^{3}x
\end{equation}
Then
\begin{equation}
\int_{\partial\bm{\mathcal{D}}}F(x)d^{2}x=-4\pi GM_{\bm{\mathcal{D}}}\equiv 4\pi G\int_{\mathcal{D}}\rho(x)d^{3}x
\end{equation}
Since $F(x)=-\nabla_{i}\psi(x)$ and using Gauss' Theorem it follows that
\begin{equation}
\int_{\partial\bm{\mathcal{D}}}F(x)d^{2}x=-\int_{\partial\bm{\mathcal{D}}}\nabla_{i}\psi(x)d^{2}x
\equiv\int_{\bm{\mathcal{D}}}\nabla_{i}\nabla^{i}\psi(x)d^{3}x=4\pi
G\int_{\bm{\mathcal{D}}}\rho(x)d^{3}x
\end{equation}
Hence, the Poisson equation for the Newtonian potential follows as
\begin{equation}
\Delta_{x}\psi(x)\equiv \nabla_{i}\nabla^{i}\psi(x)=4\pi G\rho(x)
\end{equation}
so that $C=-G$, where G is Newton's constant.
\begin{thm}
The gravitational or Coulomb potential for a ball $\bm{\mathcal{B}}_{R}(0)\subset\bm{\mathrm{R}}^{3} $ of radius R of uniform mass or charge density $\rho$ inside and outside the ball is
\begin{align}
&\psi(x)=\frac{C}{4\pi}\int_{\bm{\mathcal{D}}}\frac{\rho d^{3}y}{|x-y|}
=\rho\left(\frac{1}{2}R^{2}-\frac{1}{6}\|x|\|^{2}\right),~~\|x\|\le R\\&\psi(x)=\frac{C}{4\pi}\int_{\bm{\mathcal{D}}}\frac{\rho d^{3}y}{|x-y|}=\frac{C\rho R^{3}}{3\|x\|},~~\|x\|>R
\end{align}
\end{thm}
Outside the ball, the potential behaves as thought it is due to a point mass with all the mass of the ball concentrated at the point of origin. This property enables gravitational interactions of stars and planets within celestial mechanics to be reduced to the problem of interacting point masses. The result holds for all dimensions $n\ge 4$ and also for logarithmic potentials of discs within $\bm{\mathrm{R}}^{2}$. When the density is constant, the potential across the surface $\partial\bm{\mathcal{B}}(0)$ is constant. In his 1914 treatise $\bm{[42]}$, Herglotz studied the continuation of the potential inside the domain occupied by the masses or charges. That is, how far can $\psi(x)$ be continued analytically inside $\bm{\mathrm{R}}^{3}$ as a harmonic function?
\subsection{Common physical scenarios leading to the Laplace equation and harmonic functions}
We now consider some very common physical scenarios or examples where the PE, LE and harmonic functions naturally occur. These include classical gravitation and astrophysics, electrostatics, fluid mechanics and the steady state flows and diffusions of heat or chemical concentrations. Generally, for a domain $\bm{\mathcal{D}}$ with boundary $\partial\bm{\mathcal{D}}$, interior points describe fields, diffusions, fluids, temperatures, concentrations which eventually approach a steady state and become time-independent and can be described by either the PE or LE. Generally, we wish to consider the following generic (Dirichlet) scenario $\bm{[14,15,16]}$.
\begin{align}
&\Delta\psi(x)=f(x),~x\in\bm{\mathcal{D}}\\&
\psi(x)=g(x),~x\in\partial\bm{\mathcal{D}}
\end{align}
with the Laplace equation as the special case $f=0$. Note the difference between two solutions of the Laplace equation is harmonic. Since solutions of the PE are not unique it is necessary to impose boundary conditions at $\partial\bm{\mathcal{D}}$ of the form
\begin{equation}
a\psi(x)+b\nabla_{N(x)}\psi(x)=g(x),~x\in\partial\bm{\mathcal{D}}
\end{equation}
where $(a,b)=(1,0),(0,1),(1,>0), (1,<0)$ are the Dirichlet, Neumann, Robin and Stoklov boundary conditions respectively. In this paper, we consider only the Dirichlet boundary values problems.
\subsubsection{Newtonian astrophysics and gravitation}
The PE $\Delta_{x}\psi(x)=4\pi G\rho(x)$ plays a key role in Newtonian gravitation, astrophysics and stellar structure theory $\bm{[43,44,45]}$. Solutions of the PE and LE describe the interiors and exteriors of stars and planets. The fundamental model of a Newtonian star--which describes most of the stars in the night sky--is the Euler-Poisson system of a self-gravitating (compressible) fluid or gas in "gravi-hydrostatic" equilibrium. A fluid/gas of density $\rho:\bm{\mathrm{R}}^{3}\times\bm{\mathrm{R}}\rightarrow\bm{\mathrm{R}}$ has support in a compact domain $\bm{\mathcal{D}}\subset\bm{\mathrm{R}}^{3}$, essentially a ball $\bm{\mathcal{B}}_{R}(0)$ of radius R with $\rho=p=0$ on $\partial\bm{\mathcal{B}}_{R}(0)$. The unknowns are the fluid pressure $p:\bm{\mathrm{R}}^{3}\times\bm{\mathrm{R}}\rightarrow\bm{\mathrm{R}}$, the fluid velocity $u:\bm{\mathrm{R}}^{3}\times\bm{\mathrm{R}}\rightarrow\bm{\mathrm{R}}$ and the Newtonian gravitational potential $\psi:\bm{\mathrm{R}}^{3}\rightarrow\bm{\mathrm{R}}$ Then
\begin{align}
&\partial_{t}\rho+div(\rho)=0\\&
\rho(\partial_{t}u_{i})+u_{i}\nabla^{j}u_{j}+\nabla_{i}P(\rho))+\rho\nabla_{i}\psi(x)=0\\&
\Delta\psi(x)=4\pi \rho,~\lim_{|x|\uparrow\infty}=0
\end{align}
with $G=1$. For a star in equilibrium, $\partial_{t}\rho=0$ and $\partial_{t}u=0$. The system can be closed with a polytropic gas equation of state $p(\rho)=|\rho|^{\gamma}=\rho^{1+\tfrac{1}{n}}$, where $n$ is the polytropic index. Self-gravitating polytropic gas spheres are good approximations for most stars, including white dwarfs $\bm{[43,44]}$.

The Laplace operator is rotationally invariant in $\bm{\mathrm{R}}^{3}$ so that the PE (2.32) in spherical symmetry has the form
\begin{equation}
\frac{1}{r^{2}}\frac{d}{dr}\left(r^{2}\frac{d\psi(r)}{dr}\right)=4\pi\rho(r)
\end{equation}
where $\rho(r)$ is the density gradient in the star. For a self-gravitating sphere of fluid or gas in equilibrium
\begin{equation}
\frac{dp(r)}{dr}=GM(r),~~\frac{dM(r)}{dr}=4\pi r^{2}\rho(r)
\end{equation}
where $M(r)$ is the mass contained within a radius r or ball $\bm{\mathcal{B}}_{r}(0)$ or
$M(r)=\int_{0}^{r}\rho(r^{\prime}){r^{\prime}}^{2}dr^{\prime}$. This leads to a Poisson equation for the radial pressure $p(r)$
\begin{equation}
\Delta_{r}p(r)=\frac{1}{r^{2}}\frac{d}{dr}\left(r^{2}\frac{dp(r)}{dr}\right)=-4\pi \rho
\end{equation}
Defining dimensionless variables $\rho=\rho_{c}\theta^{n},p=p_{c}\theta^{n+1},r=\alpha\xi$, where $p_{c}$ and $\rho_{c}$ are the central pressure and density of the polytropic star leads to the Lane-Emden equation PE, a fundamental ODE for stellar structure theory $\bm{[43,44,45]}$.
\begin{equation}
\Delta_{\xi}\theta(\xi)=\frac{1}{\xi^{2}}\frac{d}{d\xi}\left(\xi^{2}\frac{d\theta(\xi)}{d\xi}\right)=-|\theta(\xi)|^{n}
\end{equation}
The equations can be solved analytically for $=0,1,5.$ $\bm{[43,45]}$.

In the vacuum regions exterior to a star or planet, $p=\rho=0$ and the PE for the Newtonian potential reduces to the LE so that $\Delta_{r}\psi(r)=0$ giving
\begin{equation}
\frac{1}{r^{2}}\frac{d}{dr}\left(r^{2}\frac{d\psi(r)}{dr}\right)= 0
\end{equation}
{\allowdisplaybreaks
Sometimes the full Laplace equation in spherical coordinates is utilised when the interior density of star or planet is not homogenous, so that the full LE is
\begin{equation}
\Delta\psi(r,\theta,\varphi)=\frac{1}{r^{2}}\frac{\partial}{\partial r}\left(r^{2}\frac{\partial\psi(r,\theta,\varphi)}{dr}\right)
+\frac{1}{{r^{2}\sin\theta}}\frac{\partial}{\partial \theta}\left(\sin\theta\frac{\partial\psi(r,\theta,\varphi)}{\partial \theta}\right)+\frac{1}{r^{2}\sin^{2}\theta}\frac{\partial^{2}}{\partial \varphi^{2}}
\psi(r,\theta,\varphi)= 0
\end{equation}
The equation is then solved by separation of variables. For the exterior of a sphere the general solution for the gravitational potential is
\begin{equation}
\psi(r,\theta,\varphi)=\sum_{\ell=0}^{\infty}\sum_{m=-\ell}^{\ell}
(\alpha_{\ell m}r^{\ell}+\beta_{\ell m}r^{-\ell+1})Y^{\ell}_{m}(\theta,\varphi)
\end{equation}
where the coefficients $\alpha_{\ell m},\beta_{\ell m}$ depend on the interior stellar density distribution, and $Y^{\ell}_{m}(\theta,\varphi)$ are the orthogonal spherical harmonics.
\subsubsection{Electromagnetic theory and electrostatics}
The PE for a charge density follows from the Coulomb law and potential solution using similar methods as discussed for gravitation. However, they also follow immediately (and with beautiful self consistency) from Maxwell's equations, discovered much later. For time-independent electric and magnetic fields the MEs are
$\bm{\nabla}.\bm{\mathrm{E}}=\nabla_{i}\mathrm{E}^{i}=\rho/\epsilon_{o},\epsilon_{ijk}\nabla^{j}E^{k}=curl~E$ and $\bm{\nabla}.\bm{\mathrm{B}}=\nabla_{i}B^{i}=0,\epsilon_{ijk}\nabla^{j}B^{k}=curl B=\mu_{o}J$. where the densities and currents have support in some domain $\bm{\mathcal{D}}$ with boundary $\partial\bm{\mathcal{D}}$. The equation $\nabla_{i}\mathrm{B}^{i}(x)=0$ establishes that there are no magnetic monopoles. Since the curl vanishes the electric field is irrotational so there exists a scalar potential $\psi(x)$, the electric potential, such that $E_{i}(x)=-\nabla_{i}\psi(x)$. Then
\begin{equation}
\nabla_{i}E^{i}(x)=-\nabla_{i}\nabla^{i}\psi(x)\equiv \Delta\psi(x)=\rho/\epsilon_{o}
\end{equation}
which reduces to a LE if the charge density is zero or in regions outside $\bm{\mathcal{D}}$ containing the charge density. Similarly, if the current is zero then $J=0$ and one can define a magnetic potential $\Psi(x)$ and a LE such that
\begin{equation}
\nabla_{i}B^{i}(x)=-\nabla_{i}\nabla^{i}\Psi(x)\equiv \Delta\Psi(x)=0
\end{equation}}
\subsubsection{Steady state fluid flow}
In fluid mechanics the velocity field  in a domain $\bm{\mathcal{D}}
\subset\bm{\mathrm{R}}^{3}$ is $u:\bm{\mathrm{R}}^{3}\times\bm{\mathrm{R}}\rightarrow\bm{\mathrm{R}}$ so that the components are $u_{i}(x,t)$. The dynamics is described by the Navier-Stokes equations. For a steady state fluid flow $\partial_{t}u_{i}(x,t)=0$ and for an irrotational flow
\begin{equation}
\epsilon_{ijk}\nabla^{j}u^{k}(x)=curl u(x)=0
\end{equation}
Hence, there exists a 'velocity potential' $\psi(x)$ such that $u_{i}(x)=-\nabla_{i}\psi(x)$. If the flow is also incompressible then $div u(x)=\nabla_{i}u^{i}(x)=0 $. Hence, $\Delta\psi(x)=0$ in analogy with electrostatics, and the velocity potential solutions are harmonic functions [1,2]. Examples are uniform laminar flow and smooth circular flow. The radial solution of the Laplace equation $\Delta_{r}\psi(r)=0$ for the velocity potential in polar coordinates is
\begin{align}
&\psi(r)=-\frac{C_{1}}{r}+C_{2},~~\bm{\mathrm{R}}^{3}\\&
\psi(r)=C_{1}\log|r|+C_{2},~~\bm{\mathrm{R}}^{2}
\end{align}
\begin{exam}
Some basic examples are:
\begin{enumerate}
\item For a 'line vortex' in $\bm{\mathrm{R}}^{2}$ one has $\psi(\theta)=k\theta$, where $k$ is the strength of the flow. ($k=\Gamma/2\pi$ if $\Gamma$ is the circulation.)The velocity components are then $u(r)=\frac{\partial\psi(\theta)}{\partial r}=0$ and $u(\theta)=\frac{1}{r}\tfrac{\partial\psi(\theta)}{\partial \theta}=\tfrac{k}{r}$ This describes a rotating fluid or 'bath-plug vortex' in the plane around $r=0$.
\item For a smooth laminar flow along the z-axis within a domain in $\bm{\mathrm{R}}^{3}$ then $u(x,y,z)=(0,0,u)$. The velocity potential is $\psi(z)=uz$ and so $\Delta\psi(z)=0$.
\item Finally, the smooth flow $u(r,z)$ in cylindrical coordinates of a fluid around a sphere of radius is given by the velocity potential
    \begin{equation}
     \psi(r,z)=uz\left(1+\frac{R^{3}}{2(r^{2}+z^{2})^{3/2}}\right)
    \end{equation}
where $u(r,z)=uz$ is the fluid velocity along the z-axis in the absence of the sphere. This function is harmonic since
\begin{equation}
\frac{1}{r^{2}}\frac{\partial}{\partial r}\left(r^{2}\frac{\psi(r,z)}{\partial r}
\right)+\frac{\partial^{2}\psi(r,z)}{\partial z^{2}}=0
\end{equation}
\end{enumerate}
\end{exam}
\subsubsection{Steady state heat flow and diffusions}
Let $\psi:\bm{\mathrm{R}}^{3}\times\bm{\mathrm{R}}\rightarrow\bm{\mathrm{R}}$ or $
\psi:\bm{\mathcal{D}}\times\bm{\mathrm{R}}\rightarrow\bm{\mathrm{R}}$ be a function of space and time. The parabolic PDE with initial data on an open domain $\bm{\mathcal{D}}$ is [3,4,5]
\begin{align}
&\partial_{t}\psi(x,t)=\frac{1}{2}\Delta\psi(x,t)\equiv\frac{1}{2}\sum_{i=1}^{3}\nabla_{i}\nabla^{i}\psi(x,t),~t>0,
x\in\bm{\mathcal{D}}\\&\psi(0,x)=f(x),~~x\in\bm{\mathcal{D}}\nonumber
\end{align}
and assuming the initial data $f:\bm{\mathcal{D}}\rightarrow\bm{\mathrm{R}}$ is bounded and continuous. This PDE describes the flow of heat or diffusions of chemicals within $\bm{\mathcal{D}}$. It can also describe diffusion of photons and radiative transfer in a highly scattering turbid medium $\bm{[46]}$. If $u(x,t)$ approaches a steady state within $\bm{\mathcal{D}}$ then $u(x,t)\rightarrow u(x)$ and the heat PDE reduces to the Dirichlet problem for the LE.
\subsubsection{Complex analysis and analytic functions}
Finally, it should be observed that there is a close connection between potential theory and complex analysis with complex analysis in turn being a useful tool in many potential theory applications. Let $z=x+iy$ then
\begin{equation}
f(z)=u(x,y)+iv(x,y)
\end{equation}
The functions u and v obey the Cauchy-Riemann equations so that
\begin{align}
&\partial_{x}u(x,y)=\partial_{y}v(x,y)\\&
\partial_{y}u(x,y)=-\partial_{x}v(x,y)
\end{align}
Then
\begin{align}
&\Delta_{x}u(x,y)+\Delta_{y}u(x,y)\equiv \partial_{x}\partial_{x}u(x,y)+\partial_{y}\partial_{y}u(x,y)\equiv u_{xx}(x,y)+u_{yy}(x,y)=0\\&
\Delta_{x}v(x,y)+\Delta_{y}v(x,y)\equiv \partial_{x}\partial_{x}v(x,y)+\partial_{y}\partial_{y}v(x,y)\equiv v_{xx}(x,y)+v_{yy}(x,y)=0
\end{align}
Hence, both the real and imaginary parts of $f(z)$ are harmonic. This provides a means of constructing a plethora of harmonic functions on $\bm{\mathrm{R}}^{2}$. For example, if $f(z)=z^{2}=(x+iy)^{2}=(x^{2}-y^{2})+i(2xy)$ then $\psi(x,y)=x^{2}-y^{2}$ and $\psi(x,y)=2xy$ are harmonic on $\bm{\mathrm{R}}^{2}$. If $f(z)=z^{n}=r^{n}\exp(i n\theta)=r^{n}\cos n\theta+r^{n}\sin n\theta$ in polar coordinates then
$\psi(r,\theta)=r^{n}\cos(n\theta)$ and $\psi(r,\theta)=r^{n}\sin(n\theta)$ are harmonic.

The Laplace equation on the complex plane has the form
\begin{equation}
\psi_{z\bar{z}}=\frac{\partial^{2}\psi(z,\bar{z})}{\partial z\partial \bar{z}}=\frac{1}{4}\Delta\psi=0
\end{equation}
where $\partial/\partial\bar{z}=\tfrac{1}{2}(\partial/\partial x+i\partial/\partial y)$ and
$\partial/\partial z=\tfrac{1}{2}(\partial/\partial x-i\partial/\partial y)$
\section{Some Classic theorems and estimates for the Laplace equation and harmonic functions}
\subsection{Dirichlet energy integral (DEI)}
This section reviews some classical theorems and estimates for harmonic functions $\bm{[12-18]}$. Given a domain $\bm{\mathcal{D}}\subset\bm{\mathrm{R}}^{n}$, the DEI was defined in (1.11)as $\mathcal{E}[\psi]=\tfrac{1}{2}\int_{\mathcal{D}}|\nabla_{i}
\psi(x)|^{2}d^{n}y $ and can be interpreted as an action integral with the Euler-Lagrange equations being the Laplace equation.
\begin{thm}
Let $\psi(x)$ be a harmonic function on $\bm{\mathcal{D}}\subset\bm{\mathrm{R}}^{n}$ such that $\Delta\psi(x)=0$. Let $\varphi(x)$ be a function $\varphi:\bm{\mathcal{D}}\rightarrow\bm{\mathrm{R}}$ that is not harmonic on
$\bm{\mathcal{D}}$ but $\psi(x)=\varphi(x)$ for all $x\in\partial\bm{\mathcal{D}}$ Then:
\begin{enumerate}
\item Harmonic functions are critical points of the DEI and if $\mathcal{E}[\psi]=0$ then $|\nabla_{i}\psi(x)|=0$ in $\bm{\mathcal{D}}$.
\item Harmonic functions have the minimum Dirichlet energy for their boundary values so that for all $\psi(x)=\varphi(x)$ with $\Delta\psi(x)=0$
    for all $x\in\partial\bm{\mathcal{D}}$
    \begin{equation}
    \int_{\bm{\mathcal{D}}}|\bm{\nabla}_{i}\psi(x)|^{2}d^{n}x\le \int_{\bm{\mathcal{D}}}|\bm{\nabla}_{i}\varphi(x)|^{2}d^{n}x
    \end{equation}
\end{enumerate}
\end{thm}
\begin{proof}
To prove (1), first define an infinitely differentiable real-valued function on $\bm{\mathcal{D}}$ so that $\phi\in C^{\infty}(\bm{\mathcal{D}})$ for any fixed $\psi$. Define the integral
\begin{equation}
\mathcal{E}(\psi,\phi,\xi)=\int_{\bm{\mathcal{D}}}|\nabla_{i}(\psi(x)+\xi\phi(x))|^{2}d^{n}x
\end{equation}
and the function $\phi(x)=0, \forall~x\in\partial\bm{\mathcal{D}}$ so that
\begin{equation}
\int_{\partial\bm{\mathcal{D}}}\phi(x)\nabla_{i}\psi(x)d^{n-1}x=0
\end{equation}
Applying the Gauss Theorem
\begin{equation}
\int_{\bm{\mathcal{D}}}\nabla^{i}(\phi(x)\nabla_{i}\psi(x))d^{n}x=\int_{\partial\bm{\mathcal{D}}}\phi(x)\nabla_{i}\psi(x)d^{n-1}x=0
\end{equation}
so that
\begin{equation}
\int_{\bm{\mathcal{D}}}\nabla^{i}\phi(x)\nabla_{i}\psi(x)d^{n}x=-\int_{\bm{\mathcal{D}}}\phi(x)\Delta\psi(x)d^{n}x
\end{equation}
Now the harmonic function $\psi(x)$ s a critical point of the DEI when
\begin{equation}
\frac{d}{d\xi}\mathcal{E}(\psi,\phi,\xi)|_{\xi=0}=0,~for~\Delta\psi(x)=0,\forall~\phi(x)
\end{equation}
The DEI (-) is
\begin{align}
&\mathlarger{\mathcal{E}}(\psi,\phi,\xi)=\frac{1}{2}\int_{\bm{\mathcal{D}}}|\nabla_{i}(\psi(x)+\xi\phi(x))|^{2}d^{n}x
=\frac{1}{2}\int_{\bm{\mathcal{D}}}|\nabla_{i}\psi(x)+\xi\nabla_{i}\phi(x)|^{2}d^{n}x\nonumber\\&
\le \frac{1}{2}\int_{\bm{\mathcal{D}}}|\nabla_{i}\psi(x)|^{2}d^{n}x
+\xi^{2}\int_{\bm{\mathcal{D}}}|\nabla_{i}\phi(x)|^{2}d^{n}x+2\xi\int_{\bm{\mathcal{D}}}z
{\nabla}_{i}\psi(x){\nabla}^{i}\phi(x)d^{n}x
\end{align}
so that
\begin{align}
\frac{d}{d\xi}\mathlarger{\mathcal{E}}(\psi,\phi,\xi)
-2\xi\int_{\bm{\mathcal{D}}}|\nabla_{i}\phi(x)|d^{n}x+2\int_{\bm{\mathcal{D}}}\nabla_{i}\psi(x)\nabla^{i}\phi(x)d^{n}x
\end{align}
At $\xi=0$
\begin{equation}
\frac{d}{d\xi}\mathlarger{\mathcal{E}}(\psi,\phi,\xi)|_{\xi=0}=2\int_{\bm{\mathcal{D}}}\nabla_{i}\psi(x)\nabla^{i}\phi(x)d^{n}x=0
\end{equation}
From (2.5) it follows that
\begin{equation}
\frac{d}{d\xi}\mathlarger{\mathcal{E}}(\psi,\phi,\xi)_{\xi=0}=-2\int_{\bm{\mathcal{D}}}\phi(x)
\Delta\psi(x)d^{n}=0
\end{equation}
Hence, the DEI is minimised for harmonic functions or $\Delta\psi(x)=0$. To prove (2) note that if $\psi(x)=\varphi(x),x\in\partial \bm{\mathcal{D}}$ then the following surface integrals vanish
\begin{align}
&\int_{\partial\bm{\mathcal{D}}}(\psi(x)-\varphi(x))\nabla(\psi(x)+\varphi(x))d^{n-1}x=0,~x\in\partial\bm{\mathcal{D}}\\&
\int_{\partial\bm{\mathcal{D}}}(\psi(x)-\varphi(x))\nabla(\psi(x)-\varphi(x))d^{n-1}x=0,~x\in\partial\bm{\mathcal{D}}
\end{align}
Using Gauss Thm
\begin{align}
&\int_{\bm{\mathcal{D}}}\nabla(\psi(x)-\varphi(x))\nabla(\psi(x)+\varphi(x))d^{n}x\nonumber\\&
=-\int_{\bm{\mathcal{D}}}(\varphi(x)-\psi(x)\Delta(\varphi(x)+\psi(x))d^{n}x\nonumber\\&
\int_{\bm{\mathcal{D}}}\nabla(\psi(x)-\varphi(x))\nabla(\psi(x)-\varphi(x))d^{n}x\nonumber\\&
=-\int_{\bm{\mathcal{D}}}(\varphi(x)-\psi(x)\Delta(\varphi(x)-\psi(x))d^{n}x
\end{align}
so that
\begin{align}
&\int_{\bm{\mathcal{D}}}\nabla(\psi(x)-\varphi(x))\nabla(\psi(x)+\varphi(x))d^{n}x\nonumber\\&
=\int_{\bm{\mathcal{D}}}|\nabla_{i}\varphi(x)|^{2}d^{n}x-\int_{\bm{\mathcal{D}}}|\nabla_{i}\psi(x)|^{2}d^{n}x\nonumber\\&
=-\int_{\bm{\mathcal{D}}}(\varphi(x)-\psi(x))\Delta(\varphi(x)+\psi(x))d^{n}x\nonumber\\&
=\int_{\bm{\mathcal{D}}}|\nabla(\varphi(x)-\psi(x)|d^{n}x\ge 0
\end{align}
and the result follows.
\end{proof}
\subsection{The mean value property and the strong maximum principle}
\begin{thm}($\bm{The~mean~value~property}$)
Let~$\psi(x)\in C^{2}(\bm{\mathcal{D}})$ be a harmonic function on a domain and let $\bm{\mathcal{B}}_{R}(x)\subset\subset\bm{\mathcal{D}}\subset\bm{\mathrm{R}}^{n} $ be ball of radius $R$ and boundary/surface $\partial\bm{\mathcal{B}}_{R}(x)$ centred at $x$. Then $\psi(x)$ satisfies the mean value property (MVP).
\begin{align}
&\psi(x)=\frac{n}{\omega_{n}R^{n}}\int_{\bm{\mathcal{B}}_{R}(x)}\psi(y)d^{n}y\equiv\frac{1}{\|\bm{\mathcal{B}}_{R}\|}
\int_{\bm{\mathcal{B}}_{R}(x)}\psi(y)d^{n}y
\\&\psi(x)=\frac{n}{\omega_{n}R^{n-1}}\int_{\partial\bm{\mathcal{B}}_{R}(x)}\psi(y)d^{n-1}\equiv\frac{1}{\|\partial\bm{\mathcal{B}}_{R}\|}\int_{\
\bm{\mathcal{B}}_{R}(x)}\psi(y)d^{n}y
\end{align}
where $\omega_{n}$ is the measure of $\partial\bm{\mathcal{B}}_{R}(x)$. The volume of an n-ball
is $V(\bm{\mathcal{B}}_{R}(x))=\int_{\bm{\mathcal{B}}_{R}(x)}d^{n}x=\omega_{n}R^{n}/n\equiv (\pi^{n/2}/\Gamma(n+1))R^{n}$. If $\psi(y)=\psi=const.$ for all $y\in\bm{\mathcal{B}}_{R}(x)$ then
\begin{equation}
\psi(x)=\frac{n}{\omega_{n}R^{n}}\int_{\bm{\mathcal{B}}_{R}(x)}\psi d^{n}y=\frac{n}{\omega_{n}R^{n}}\psi\frac{\omega_{n}R^{n}}{n}=\psi
\end{equation}
\end{thm}
The proof is found in most texts [3,4]
\begin{prop}
Given the harmonic function $\psi(x)$ and the MVP, define the real fields $\varphi(x)$ such that
\begin{align}
&\varphi(x)=\exp(\eta \psi(x))=\exp\left(\frac{\eta}{\|\bm{\mathcal{B}}_{R}\|}
\int_{\bm{\mathcal{B}}_{R}(x)}\psi(y)d^{n}y\right)
\\&\varphi(x)=\exp(\eta \psi(x))=\exp\left(\frac{\eta}{\|\partial\bm{\mathcal{B}}_{R}\|}
\int_{\partial\bm{\mathcal{B}}_{R}(x)}\psi(y)d^{n-1}y\right)
\end{align}
where $\alpha>0$. Then $\varphi(x)$ is harmonic only if $|\nabla_{i}\psi(x)|=0$ or $\psi=const.$ Complex fields $\chi(x)$ can be defined as
\begin{align}
&\mathrm{Z}(x)=\exp(i\eta \psi(x))=\exp\left(\frac{i\eta}{\|\bm{\mathcal{B}}_{R}\|}
\int_{\bm{\mathcal{B}}_{R}(x)}\psi(y)d^{n}y\right)
\\&\mathrm{Z}(x)=\exp(i\eta \psi(x))=\exp\left(\frac{i\eta}{\|\partial\bm{\mathcal{B}}_{R}\|}
\int_{\partial\bm{\mathcal{B}}_{R}(x)}\psi(y)d^{n-1}y\right)
\end{align}
The field $\mathrm{Z}(x)$ is harmonic if $|\nabla_{i}\psi(x)|=0$.
\end{prop}
\begin{proof}
The Laplacian of $\mathrm{Z}(x)$ is
\begin{align}
&\Delta\mathrm{Z}(x)=\nabla_{i}\nabla^{i}\varphi(x)=\nabla_{i}(\eta\nabla^{i}\psi(x))\exp(\eta \psi(x))\nonumber\\&
=\eta\nabla_{i}\psi(x)\nabla^{i}\psi(x)\exp(\eta\psi(x))+(\Delta\psi(x))\exp(\eta \psi(x))=\eta\nabla_{i}\psi(x)\nabla^{i}\psi(x)\exp(\eta\psi(x))
\end{align}
which is harmonic if $|\nabla_{i}\psi(x)|^{2}=0$, and which holds for $\psi(x)=const.$ or at the extremum or critical point. The same argument applies to the field $\mathrm{Z}(x)$.
\end{proof}
\begin{lem}($\bm{Harnack~Inequality}$)~
Let $\psi\in C^{2}(\bm{\mathcal{D}})$ be a non-negative function harmonic in $\bm{\mathcal{D}}\subset\bm{\mathrm{R}}^{n}$ so that $\Delta \psi(x)=0$. Let $\bm{\mathcal{B}}_{R}(y)\subset\bm{\mathcal{D}}$ be an n-ball with radius R and centre $y$ with $x\in\bm{\mathcal{B}}_{R}(y)$. Then
\begin{equation}
(R^{n}(R+|x-y|)^{-n}\psi(y)\le \psi(x)\le (R^{n}(R-|x-y|)^{-n}\psi(y)
\end{equation}
or
\begin{equation}
(1+Q)^{-n}\psi(y)\le \psi(x)\le (1-Q)^{-n}\psi(y)
\end{equation}
where $Q=|x-y|/R$.
\end{lem}
\begin{proof}
The proof follows quickly from the MVP and the positivity of $\psi$. let $r_{(\pm)}=R\pm|x-y|$ and let
$\bm{\mathcal{B}}_{r_{(\pm)}}(x)$ and $\bm{\mathcal{B}}_{R}(y)$ be balls with radii $r_{(\pm)}$ and $R$ and centres $x$ and $y$ so that $r_{\pm}<R$ and $\bm{\mathcal{B}}_{r_{(-)}}(x)\subset \bm{\mathcal{B}}_{R}(y)$. From (2.15)
\begin{align}
&\psi(x)=\frac{n}{\omega_{n}R^{n}}\int_{\bm{\mathcal{B}}_{r_{(-)}}(x)}\psi(y)d^{n}y\le
\frac{n}{\omega_{n}r_{(-)}^{n}}\int_{\bm{\mathcal{B}}_{R}(y)}\psi(z)d^{n}z\nonumber\\&
=\frac{R^{n}}{r_{(-)}^{n}}\psi(y)=R^{n}(R-|x-y|)^{n}\psi(y)
\end{align}
Similarly
\begin{align}
&\psi(x)=\frac{n}{\omega_{n}R^{n}}\int_{\bm{\mathcal{B}}_{r_{(+)}}(x)}\psi(y)d^{n}y\ge
\frac{n}{\omega_{n}r^{n}}  \int_{\bm{\mathcal{B}}_{R}(y)}\psi(z)d^{n}z\nonumber\\&
=\frac{R^{n}}{r_{(+)}^{n}}\psi(y)=R^{n}(R-|x-y|)^{-n}\psi(y)
\end{align}
\end{proof}
\begin{thm}(Liouville~Thm)
Suppose $\psi\in C^{2}(\bm{\mathrm{R}}^{n})$ is a harmonic function on $\bm{\mathrm{R}}^{n}$. Let $Q$ be a constant such that $\psi(x)\le |Q|$ or $\psi(x)\ge |Q|$ for all $x\in\bm{\mathrm{R}}^{n}$. Then $\psi(x)$ is necessarily a constant function.
\end{thm}
\begin{proof}
Suppose $\psi\in C^{2}(\bm{\mathrm{R}}^{n})$ is a harmonic function on $\bm{\mathrm{R}}^{n}$. Let $\Psi(x)=\psi(x)+|Q|$ then $\Psi(x)$ is harmonic so it satisfies the inequality
\begin{equation}
(R^{n}(R+|x-y|)^{-n}\Psi(y)\le \Psi(x)\le (R^{n}(R-|x-y|)^{-n}\Psi(y)
\end{equation}
Letting $R\rightarrow\infty$ then $\Psi(y)\le \Psi(x)\le \Psi(y)$
\end{proof}
A function satisfying the mean value property in a domain cannot attain its maximum or minimum at an interior point of the domain, unless it is constant. In case $\bm{\mathcal{D}}$ is bounded and $\psi(x)$ (non constant) is continuous up to the boundary of $\bm{\mathcal{D}}$, it follows that $\psi$ attains both its maximum and minimum only on $\partial\bm{\mathcal{D}}$. This result expresses a maximum principle that is stated precisely in the following theorem.
\begin{thm}($\bm{The~strong~maximum~principle}$)
Let $\psi\in C(\bm{\mathcal{D}})$ satisfy the MVP. If $\psi(x)$ attains a maximum within $\bm{\mathcal{D}}$ at some $x=p$ then $\psi(x)$ is constant in $\bm{\mathcal{D}}$. If $\bm{\mathcal{D}}$ is bounded with boundary $\partial\bm{\mathcal{D}}$ and $\psi\in C(\bm{\mathcal{D}})$ is not constant then $\forall x\in\bm{\mathcal{D}}$ it holds that
\begin{align}
&\psi(x)<\max \psi(y),~~y\in\partial \bm{\mathcal{D}}\\&
\psi(x)>\min \psi(y),~~y\in\partial \bm{\mathcal{D}}
\end{align}
It also holds that since $\psi(x)>0$
\begin{align}
&|\psi(x)|^{2}<\max |\psi(y)|^{2},~~y\in\partial \bm{\mathcal{D}}\\&
|\psi(x)|^{2}>\min |\psi(y)|^{2},~~y\in\partial \bm{\mathcal{D}}
\end{align}
and
\begin{align}
&|\psi(x)|^{p}<\max |\psi(p)|^{2},~~y\in\partial \bm{\mathcal{D}}
|\psi(x)|^{p}>\min |\psi(p)|^{2},~~y\in\partial \bm{\mathcal{D}}
\end{align}
for all $p\ge 2$.
\begin{proof}
Briefly, consider the minimum case with $n=3$. Let $\psi(p)$ be the minimum at a point $p\in\bm{\mathcal{D}}$ and $\psi(p)=\psi_{m}$. Let $\bm{\mathrm{R}}(p)$ be any ball of radius $R$ with centre p and let $z\in\bm{\mathcal{B}}_{R}(p)$ and
$\bm{\mathcal{B}}_{r}(z)\subset\bm{\mathcal{B}}_{R}(p)$ is a small ball of radius $r$ centred at $z$. By definition of a minimum
\begin{equation}
\psi(z)>\psi_{m}
\end{equation}
Since $\psi$ is harmonic then the MVP holds then using (2.15)
\begin{align}
&\psi_{m}\equiv \psi(p)=\frac{1}{\|\bm{\mathcal{B}}_{R}(p)\|}
\int_{\bm{\mathcal{B}}_{R}(p)}\psi(y)d^{3}y\nonumber\\&=\frac{1}{\|\bm{\mathcal{B}}_{R}(p)\|}
\bigg(\int_{\bm{\mathcal{B}}_{R}(p)}\psi(y)d^{3}y
+\int_{\bm{\mathcal{B}}_{R}(p)\symbol{92}\bm{\mathcal{B}}_{r}(z)}\psi(y)d^{3}y\bigg)
\nonumber\\&=\frac{1}{\|\bm{\mathcal{B}}_{R}(p)\|}\bigg(\|\bm{\mathcal{B}}_{r}(z)\|\psi(z)
+\int_{\bm{\mathcal{B}}_{R}(p)\symbol{92}\bm{\mathcal{B}}_{r}(z)}\psi(y)d^{3}y\bigg)\nonumber\\&\ge
\frac{1}{\|\bm{\mathcal{B}}_{R}(p)\|}\bigg(\|\bm{\mathcal{B}}_{r}(z)\|\psi(z)+\|\bm{\mathcal{B}}_{R}(p)\|
-\|\bm{\mathcal{B}}_{r}(z)\|\psi_{m})\nonumber\\&
\frac{\|\bm{\mathcal{B}}_{r}(z)\|}{\|\bm{\mathcal{B}}_{R}(p)\|}\psi(z)+\psi_{m}-\psi_{m}
\frac{\|\bm{\mathcal{B}}_{r}(z)\|}{\| \bm{\mathcal{B}}_{R}(p)\|}
\end{align}
Hence
\begin{equation}
\frac{\bm{\mathcal{B}}_{r}(z)}{\bm{\mathcal{B}}_{R}(p)}(\psi(z)-\psi_{m})\le 0
\end{equation}
or $\psi(z)\le\psi_{m}$. Comparing with (2.33) it follows that $\psi(x)=\psi_{m},\forall~x\in\bm{\mathcal{B}}_{R}(p)$. Hence, $\psi(x)$ is constant at any point where it is a minimum. Since $\bm{\mathcal{D}}$ is open and connected then $\psi(x)=\psi_{m},\forall x\in\bm{\mathcal{D}}$.
\end{proof}
\end{thm}
An immediate corollary allows a comparison between the size of two solutions to the Poisson equation if there is information about the size of the source terms and the values of the solutions on $\partial\bm{\mathcal{D}}$.
\begin{cor}
The Dirichlet problem $\Delta\psi(x)=0,~x\in\bm{\mathcal{D}}$ and $\psi(x)=f(x),x~\partial\bm{\mathcal{D}}$ has at most one solution $\psi_{f}\in C^{2}(\bm{\mathcal{D}})\cap C(\overline{\bm{\mathcal{D}}})$. If $(f,g)\in C(\bm{\mathcal{D}})$ are data on $\partial\bm{\mathcal{D}}$ then the following hold.
\begin{enumerate}
\item $\underline{Comparison Principle}$: if $f\ge g$ on $\partial\bm{\mathcal{D}}$ and $f\ne g$ then $\psi_{f}(x)>\psi_{g}(x)$ in $\bm{\mathcal{D}}$.
\item $\underline{Stability Estimate}$: for any $x\in\bm{\mathcal{D}}$
\begin{equation}
\|\psi_{f}(x)-\psi_{g}(x)\| \le \max_{y\in\partial\bm{\mathcal{D}}}\|f(y)-g(y)\|
\end{equation}
\begin{equation}
\|\psi_{f}(x)-\psi_{g}(x)\|^{p} \le \max_{y\in\partial\bm{\mathcal{D}}}\|f(y)-g(y)\|^{p}
\end{equation}
for all $p>1$.
\end{enumerate}
\end{cor}
\subsection{The Bochner-Weizentbock formula}
Let $\psi:\bm{\mathcal{D}}\rightarrow\bm{\mathrm{R}}^{n}$ be real-valued function on an open subset/domain $\bm{\mathcal{D}}$, and let $\|A\|^{2}=\sum_{i,j}^{n}|a_{ij}|^{2}$ be the norm of an $n\times n$ matrix with coefficients $a_{ij}$. The norm of the Hessian of $\psi(x)$ is
\begin{equation}
\|{\bm{\mathrm{H}}}\psi(x)\|^{2}=\sum_{i,j}^{n}|\nabla_{i}\nabla_{j}
\psi(x)|^{2}
\end{equation}
\begin{thm}
Given $\psi:\bm{\mathcal{D}}\rightarrow\bm{\mathrm{R}}^{n}$ for an open subset $\bm{\mathcal{D}}$ then
the Bochner-Weizentbock formula is
\begin{equation}
\frac{1}{2}\Delta|\nabla\psi(x)|^{2}=\bigg\langle \nabla_{j}(\Delta\psi(x)),\nabla_{j}\psi(x)\bigg\rangle
+\|{\bm{\mathrm{H}}}\psi(x)\|^{2}
\end{equation}
where $\langle X,Y\rangle=X_{i}Y^{i}$ is the inner product. If $\psi(x)$ is harmonic then $\Delta\psi(x)=0$ so that
\begin{equation}
\frac{1}{2}\Delta|\nabla\psi(x)|^{2}=\|{\bm{\mathrm{H}}}\psi(x)\|^{2}
\end{equation}
For a general Riemannian manifold $\mathcal{M}$ with metric $g$ where the derivatives do not commute the Ricci tensor arises as an extra term
\begin{equation}
\frac{1}{2}\Delta|\nabla\psi(x)|^{2}=\bigg\langle \nabla_{j}(\Delta\psi(x)),\nabla_{j}\psi(x)\bigg\rangle+\|{\bm{\mathrm{H}}}\psi(x)\|^{2}
+{\bm{Ric}}(\nabla\psi(x),\nabla\psi(x))
\end{equation}
\end{thm}
\begin{proof}
The proof follows from explicitly computing out the lhs.
\begin{align}
&\frac{1}{2}\Delta|\nabla\psi(x)|^{2}=\frac{1}{2}\sum_{ij}^{n}\nabla_{i}\bigg(\nabla^{i}
(\nabla_{j}\psi(x)\nabla_{j}\psi(x)\bigg)\nonumber\\&
=\sum_{ij}(\nabla^{i}\nabla_{j}\psi(x))(\nabla^{i}\nabla_{j}\psi(x))+\sum_{ij}(\nabla_{i}\nabla^{i}\nabla_{j}\psi(x))\nabla_{j}\psi(x)\nonumber\\&
=\|{\bm{\mathrm{H}}}\psi(x)\|^{2}+\sum_{ij}^{n}\nabla_{j}(\Delta\psi(x))\nabla_{j}\psi(x)\nonumber\\&
=\|{\bm{\mathrm{H}}}\psi(x)\|^{2}+\bigg\langle\nabla_{j}(\Delta\psi(x)),\nabla_{j}\psi(x)\bigg\rangle
\end{align}
\end{proof}
\subsection{The gradient estimate and the Cacciopolli estimates}
Harmonic functions within an n-ball obey a number of well-established inequality bounds and estimates.
\begin{thm}$(\bm{The~gradient~estimate})~$
For all harmonic functions $\psi(x)\in\bm{\mathcal{B}}_{2R}(x_{o})\subset\bm{\mathrm{R}}^{n}$ with $\bm{\mathcal{B}}_{R}(x_{o})\subset\bm{\mathcal{B}}_{2R}(x_{o})$, $\exists$ constants $C(n)$ such that
\begin{equation}
\sup_{\bm{\mathcal{B}}_{R}(x_{o})}|\nabla_{i}\psi(x)|\le \frac{C(n)}{R}\sup_{\bm{\mathcal{B}}_{2R}(x_{o})}|\psi(x)|
\end{equation}
\end{thm}
The proof is quite detailed and utilises the Bochner formula.
\begin{thm}
Let $\psi:\bm{\mathcal{B}}_{2R}\rightarrow\bm{\mathrm{R}}^{n}$ and let $\psi(x)\nabla_{i}\psi(x)=0$ then the following estimates holds
\begin{equation}
\int_{\bm{\mathcal{B}}_{R}}|\nabla_{i}\psi(x)|^{2}d^{n}x\le\frac{4}{R^{2}}\int_{\bm{\mathcal{B}}_{2R}\symbol{92}\bm{\mathcal{B}}_{R}}
|\psi(x)|^{2}d^{n}x
\end{equation}
\end{thm}
A consequence of the Cacciopolli estimate are the following estimates which bound the rate at which a harmonic function can decay.
\begin{thm}
If the Cacciopolli estimate holds then there are constants $C(n)$ and $D(n)$ such that the following estimates hold for the function and its gradient
\begin{align}
&\int_{\bm{\mathcal{B}}_{2R}}|\psi(x)|^{2}d^{n}x \ge (1+ D(n))\int_{\bm{\mathcal{B}}_{2R}}|\psi(x)|^{2}d^{n}x
\\&\int_{\bm{\mathcal{B}}_{2R}}|\nabla_{i}\psi(x)|^{2}d^{n}x \ge (1+
C(n))\int_{\bm{\mathcal{B}}_{2R}}|\nabla_{i}\psi(x)|^{2}d^{n}x
\end{align}
\end{thm}
\subsection{Riesz-Newtonian potentials}
The Newtonian potential and the Poisson integral formula solutions of the Laplace equation for the Dirichlet boundary value problem, can be considered as special cases of the Riesz potential, which is first defined $\bm{[46]}$.
\begin{defn}
Let $\bm{\mathcal{D}}\subset\bm{\mathrm{R}}^{n}$ and $f:\bm{\mathcal{D}}\rightarrow\bm{\mathrm{R}}^{n}$ then the volume and surface-integral Rietsz potentials are defined as
\begin{align}
&\mathlarger{\bm{\mathfrak{F}}}_{a}g(x)=\int_{\bm{\mathcal{D}}}\frac{g(y)d\mu(y}{|x-y|^{n-\alpha}}
=\int_{\bm{\mathcal{D}}}\frac{g(y)d^{n}y}{|x-y|^{n-\alpha}}
,~~(x,y)\in\bm{\mathcal{D}}\\&
\mathlarger{\bm{\mathfrak{F}}}_{a}g(x)=\int_{\partial\bm{\mathcal{D}}}
\frac{g(y)d\mu(y}{|x-y|^{n-\alpha}}=\int_{\partial\bm{\mathcal{D}}}
\frac{g(y)d^{n-1}y}{|x-y|^{n-\alpha}},~~y\in\partial\bm{\mathcal{D}}
\end{align}
with $L_{P}$ norms
\begin{align}
&\|\mathlarger{\bm{\mathfrak{F}}}_{a}g(x)\|=\left\|\int_{\bm{\mathcal{D}}}\frac{g(y)d^{n}y}{|x-y|^{n-\alpha}}\right\|
=\left(\int_{\bm{\mathcal{D}}}\left|\int_{\bm{\mathcal{D}}}\frac{g(y)d^{n}y}{|x-y|^{n-\alpha}}\right|^{p}d^{n}x\right)^{1/p}\\&
\|\mathlarger{\bm{\mathfrak{F}}}_{a}g(x)\|=\left\|\int_{\partial\bm{\mathcal{D}}}\frac{g(y)d^{n-1}y}{|x-y|^{n-\alpha}}\right\|=
\left(\int_{\bm{\mathcal{D}}}\left|\int_{\bm{\mathcal{D}}}\frac{g(y)d^{n-1}y}{|x-y|^{n-\alpha}}\right|^{p}d^{n-1}x\right)^{1/p}
\end{align}
This singular integral is well-defined provided $f(x)$ decays sufficiently rapidly at infinity. The RP is also related to embedding theory such as Sobolev embedding.
\end{defn}
\begin{prop}
Define the following '$\alpha$-energies' or capacities for the Reisz potential
\begin{align}
&{\mathscr{C}}_{\alpha}(n,\mu,g)=\mathlarger{\iint}_{\bm{\mathcal{D}}}
\frac{g(y)d\mu(x)d\mu(y)}{|x-y|^{n-\alpha}}\\&
{\mathscr{C}}_{\alpha}(n,\mu,g=1)=\mathlarger{\iint}_{\bm{\mathcal{D}}}
\frac{d\mu(x)d\mu(y)}{|x-y|^{n-\alpha}}\\&{\mathscr{C}}_{\alpha}(n,\mu,g)=\int_{\bm{\mathcal{D}}}\left|\int_{\bm{\mathcal{D}}}
\frac{g(y)d\mu(y)}{|x-y|^{n-\alpha}}\right|^{p} d\mu(x)\equiv
\left\|\int_{\bm{\mathcal{D}}}\frac{g(y)d\mu(y)}{|x-y|^{n-\alpha}}\right\|^{p}_{L_{p}}
\\&{\mathscr{C}}_{\alpha}(n,\mu,g)=\int_{\bm{\mathcal{D}}}
\left|\int_{\bm{\mathcal{D}}}\frac{d\mu(y)}{|x-y|^{n-\alpha}}\right|^{p} d\mu(x)\equiv
\left\|\int_{\bm{\mathcal{D}}}\frac{d\mu(y)}{|x-y|^{n-\alpha}}\right\|^{p}_{L_{p}}
\end{align}
\end{prop}
\begin{thm}
Let $\mathlarger{\bm{\mathfrak{F}}}_{a}f(x)$ be a RP with respect to a function $f(x)$. Then the following estimates hold:
\begin{enumerate}
\item For a function $\varphi\in C^{\infty}(\bm{\mathrm{R}}^{n})$ there is a bound
\begin{equation}
|\varphi(x)|\le C(n)\int_{\bm{\mathcal{D}}}\frac{|\nabla\varphi(x)|d^{3}y}{|x-y|^{n-a}}
=C(n)|\mathlarger{\bm{\mathfrak{F}}}_{a}\nabla\varphi(x)|
\end{equation}
\item There is a $C>0$ such that
\begin{equation}
\|\mathlarger{\bm{\mathfrak{F}}}_{a}f(x)\|_{L_{q}(\bm{\mathcal{D}})}\le C\|f\|_{L_{q}(\bm{\mathcal{D}})}
\end{equation}
or
\begin{equation}
\left\|\int_{\bm{\mathcal{D}}}\frac{f(y)d^{n}y}{|x-y|^{n-\alpha}}\right\|\le
C\|f\|_{L_{q}(\bm{\mathcal{D}})}
\end{equation}
\end{enumerate}
\end{thm}
\begin{proof}
Assuming that $f(x)\rightarrow 0$ for $|x|\rightarrow 0$ then $\mathlarger{\bm{\mathfrak{F}}}_{a}f(x)$ is well defined. For $\zeta>0$, define a re-scaled function as $f_{\zeta}(x)=f(\zeta x)$. Then
\begin{align}
&\bigg\|\int_{\bm{\mathcal{D}}}\frac{f(\zeta y)d^{n}y}{|x-y|^{n-\alpha}}\bigg\|_{L_{q}}=\bigg\|\int_{\bm{\mathcal{D}}}
\frac{f_{\zeta}(y)d^{n}y}{|x-y|^{n-\alpha}}\bigg\|_{L_{q}}\nonumber\\&
=\bigg(\int_{\bm{\mathcal{D}}}\bigg|\int_{\bm{\mathcal{D}}}
\frac{f_{\zeta}(y)d^{n}y}{|x-y|^{n-\alpha}}\bigg|^{q}d^{n}x\bigg)^{1/p}\nonumber\\&
\le C\bigg(\int_{\bm{\mathcal{D}}}|f_{\zeta}(x)|^{p}d^{n}x\bigg)^{1/p}\equiv C\bigg(\int_{\bm{\mathcal{D}}}|f_{\zeta}(\zeta x)|^{p}d^{n}x\bigg)^{1/p}
\bigg)
\end{align}
Now let $z=\zeta x,w=\zeta y$ so that $ x=(z/\zeta),y=(w/\zeta)$ and
\begin{equation}
|x-y|^{n-\alpha}=\tfrac{1}{\zeta^{n-\alpha}}|z-w|^{n-\alpha}\nonumber
\end{equation}
and $d^{3}x=d^{n}z/\zeta^{n}$ and $d^{n}y=d^{3}z/\zeta^{n}$.
\begin{align}
&\left\|\int_{\bm{\mathcal{D}}}\frac{\zeta^{n-\alpha}}{\zeta^{n}} \frac{f(w)}{|z-w|^{n-\alpha}}d^{n}w\right\|_{L_{q}}\equiv\left(\int_{\bm{\mathcal{D}}}
\left|\int_{\bm{\mathcal{D}}}\frac{\zeta^{n-\alpha}}{\zeta^{n}}
\frac{f(w)}{|z-w|^{n-\alpha}}d^{n}w\right|^{q}\right)^{1/q}
\nonumber\\&\le C\left(\frac{1}{r^{\zeta}}\int_{\bm{\mathcal{D}}}|f(z)|^{p}dz\right)^{1/p}
\end{align}
or
\begin{equation}
\left(\int|\mathlarger{\bm{\mathfrak{F}}}_{\alpha}f(z)d^{n}z
|\right)^{1/q}\le Cr^{\alpha-\frac{n}{p}+\frac{n}{p}}\left(\int|f(z)|^{p}d^{n}z\right)^{1/p}
\end{equation}
If $(\alpha-\frac{n}{p}+\frac{n}{p}\rightarrow 0$ then $f=0$ for $r\uparrow 0^{+}$. If
$(\alpha-\frac{n}{p}+\frac{n}{p}< 0$ then $f=0$ for $r\uparrow\infty$. The only possible scenario is $(n/p)=(n/p)-\alpha$. A positive $q$ requires $\alpha p<n$ so that $q=np/(n-\alpha p)$. Hence $r^{\alpha-\frac{n}{p}+\frac{n}{p}}=1$.
\end{proof}
\begin{cor}
The Newtonian potential $\psi(x)$ of (2.20) coincides with the Reisz potential for $\alpha=n-1$ and $g(x)=C\rho(x)$, the source density so that in $\bm{\mathrm{R}}^{n}$.
\begin{equation}
\psi(x)=C\int_{\bm{\mathcal{D}}}\frac{\rho(y)d^{n}y}{|x-y|}
=c\int_{\bm{\mathcal{D}}}\frac{f(y)}{|x-y|^{n-(n-1)}}d^{3}y
=I_{n-1}\rho(x)
\end{equation}
\end{cor}
\begin{thm}(Fraenkel,~2000,$\bm{[48]}$)\newline
Let $\bm{\mathcal{D}}$ be a bounded or open set/domain and let $\|\bm{\mathcal{B}}_{1}(0)\|\subset\mathrm{R}^{n}$ be the measure on a unit ball. Then
\begin{align}
&\psi(x)=\frac{1}{2\pi}\int_{\bm{\mathcal{D}}}\log\frac{1}{|x-y|}d\mu(y),~~n=2\nonumber\\&
\psi(x)=\frac{1}{(n-2)\|\partial\bm{\mathcal{B}}_{1}(0)\|}\int_{\bm{\mathcal{D}}}
\frac{d\mu}{|x-y|^{n-2}},~~n=3
\end{align}
If $\psi(x)=C $, a constant for all $x\in\partial\bm{\mathcal{D}}$ then $\bm{\mathcal{D}}$ is a finite ball $\bm{\mathcal{D}}=\bm{\mathcal{B}}_{R}(0)$ and $\psi(x)$ is harmonic in $\bm{\mathrm{R}}^{n}\symbol{92}\bm{\mathcal{D}}$.
\end{thm}
\subsection{The Dirichlet problem on a disc and the Poisson formula}
\begin{thm}
The Dirichlet BVP for the Laplace equation on a disc or circle $\bm{\mathcal{D}}\subset\bm{\mathrm{R}}^{n}$ of radius $R$ and area $A(\bm{\mathcal{D}})=\pi R^{2}$ is given by $\Delta\psi(x)=0,x\in\bm{\mathcal{D}}$ and $\psi(x)=g(x),~x\in\partial\bm{\mathcal{D}}$. In polar coordinates this is the Dirichlet BVP is
\begin{align}
&\Delta\psi(r,\theta)=\partial_{rr}\psi(r,\theta)+\frac{1}{r}\partial_{r}\psi(r,\theta)+\frac{1}{r^{2}}\partial_{\theta\theta}\psi(r,\theta)=0,~r\in[0,R)\\&
\psi(R,\theta)=g(\theta),~~\theta\in[0,2\pi)
\end{align}
The well-known solution is then given by the Fourier series
\begin{equation}
\psi(r,\theta)=\frac{1}{2}A_{o}+\sum_{m=1}^{\infty}\left(\frac{r}{R}\right)^{m}(A_{m}\cos(m\theta)+B_{m}\sin(m\theta)),~m=0,1,2,3...
\end{equation}
and on the boundary where $r=R$ and $g(\theta)=\psi(R,\theta)$
\begin{equation}
g(r)=\frac{1}{2}A_{o}+\sum_{m=1}^{\infty}(A_{m}\cos(m\theta)+B_{m}\sin(m\theta)),~m=0,1,2,3...
\end{equation}
with the Fourier coefficients
\begin{align}
&A_{m}=\frac{1}{\pi}\int_{0}^{2\pi}g(\beta)\cos(m\beta)d\beta\\&
B_{m}=\frac{1}{\pi}\int_{0}^{2\pi}g(\beta)\sin(m\beta)d\beta
\end{align}
The solution to (2.72) then has the Poisson integral representation
\begin{equation}
\psi(r,\theta)=\frac{1}{2\pi}\left|\int_{0}^{2\pi}\frac{R^{2}-r^{2}}{R^{2}-2rR\cos(\theta-\beta)+r^{2}}\right |g(\beta)d\beta\equiv \frac{1}{2\pi}\int_{0}^{2\pi}|\Pi(R,r,\theta,\beta)|g(\beta)d\beta
\end{equation}
\end{thm}
\begin{proof}
The well-known solution follows from separation of variables $\psi(r,\theta)
=R(r)\Theta(\theta)$. To derive the Poisson integral [3,4], the Fourier coefficients (1.79) and (1.80) are substituted back into (1.77).
\begin{align}
&\psi(r,\theta)=\frac{1}{2}A_{o}+\sum_{m=1}^{\infty}\left
(\frac{r}{R}\right)^{m}(A_{m}\cos(m\theta)+B_{m}\sin(m\theta))\nonumber\\&
=\frac{1}{\pi}\int_{0}^{2\pi}\left|\frac{1}{2}+\sum_{m=1}^{\infty}\left(\frac{r}{R}\right)^{m}\cos(m(\theta)\cos(m\beta))+\sin(m\theta)\sin(m\beta)
\right|g(\beta)d\beta\nonumber\\&
=\frac{1}{\pi}\int_{0}^{2\pi}\left|\frac{1}{2}+\sum_{m=1}^{\infty}\left(\frac{r}{R}\right)^{n}\cos(m(\theta-\beta))\right|g(\beta)d\beta
\nonumber\\&=\frac{1}{2\pi}\int_{0}^{\pi}\left|1+\sum_{m=1}^{\infty}\left(\frac{r}{R}\right)^{m}[\exp(i m(\theta-\beta)+\exp(-im(\theta-\beta))]\right|g(\beta)d\beta
\nonumber\\&=\frac{1}{2\pi}\int_{0}^{\pi}\bigg|1+\underbrace{\sum_{m=1}^{\infty}\bigg(\frac{r}{R}\exp(i(\theta-\beta))
\bigg)^{n}}_{geometric~series}+\underbrace{\sum_{m=1}^{\infty}\bigg(\frac{r}{R}
\exp(-i(\theta-\beta))\bigg)^{n}}_{geometric~series}
\bigg|g(\beta)d\beta
\nonumber\\&
\frac{1}{2\pi}\int_{0}^{2\pi}\left| 1+\frac{\frac{r}{R}\exp(i(\theta-\alpha))}   {1-(\frac{r}{R}\exp(i(\theta-\alpha))}+\frac{\frac{r}{R}\exp(-i(\theta-\alpha))}   {1-(\frac{r}{R}\exp(i(\theta-\alpha))}\right|g(\beta)d\beta\nonumber\\&
=\frac{1}{2\pi}\int_{0}^{2\pi}\left| 1+\frac{r\exp(i(\theta-\alpha))}   {R-r\exp(i(\theta-\alpha))}+\frac{r\exp(-i(\theta-\alpha))}   {R-r\exp(-i(\theta-\alpha))}\right|g(\beta)d\beta\nonumber\\&
=\frac{1}{2\pi}\int_{0}^{2\pi}\left| 1 +  \frac{r\exp(i(\theta-\beta)(R-r\exp(-i(\theta-\beta))
+r\exp(-i(\theta-\beta))(R-r\exp(i(\theta-\beta)}{R^{2}-2rR\cos(\theta-\beta)+r^{2}}\right|\nonumber\\&
=\frac{1}{2\pi}\int_{0}^{2\pi}\left| 1+    \frac{2rR\cos(\theta-\beta)-2r^{2}}{R^{2}-2rR\cos(\theta-\beta)+r^{2}}\right|g(\beta)d\beta
=\frac{1}{2\pi}\int_{0}^{2\pi}\left|\frac{R^{2}-r^{2}}{R^{2}-2rR\cos(\theta-\beta)+r^{2}}
\right|g(\beta)d\beta
\end{align}
\end{proof}
\begin{cor}
At the centre of the disc $(r,\theta)=(0,0)$ so that
\begin{equation}
\psi(0,0)=\frac{1}{2\pi}\int_{0}^{2\pi}g(\beta)d\beta
\end{equation}
so that the solution at the centre of the disc is equal to the mean value of g on the boundary.
\end{cor}
\subsection{The Dirichlet problem and the Poisson formula for an n-ball or hypersphere}
The Poisson integral representation can be given for all n-balls with $n>2$. First, the following standard results are given without proof. For details see [3,4,5] for example.
\begin{thm}
($\underline{Representation~formula~for~solutions~to~the~boundary\newline~value~Poisson~equation}$).
Let $\bm{\mathcal{D}}\subset\bm{\mathrm{R}}^{n}$ be a domain with a smooth boundary $\partial \bm{\mathcal{D}}$ with $\psi\in C^{2}(\bm{\mathcal{D}}),g\in C(\partial\bm{\mathcal{D}})$. The unique solution
$\psi\in C^{2}(\bm{\mathcal{D}})\cap C(\overline{\bm{\mathcal{D}}})$ to
\begin{align}
&\Delta\psi(x)=\nabla_{i}\nabla^{i}\psi(x)=f(x),~x\in\bm{\mathcal{D}}\\&
~~~~~~~~~~~~~~\psi(x)=g(x),~x\in\partial{\bm{\mathcal{D}}}
\end{align}
has the representation
\begin{equation}
\psi(x)=\int_{\bm{\mathcal{D}}}f(y)G(x,y)d^{n}y+\int_{\partial\bm{\mathcal{D}}}g(\sigma)
\nabla_{\hat{N}(\sigma)}G(x,\sigma)d\sigma
\end{equation}
for $(x,y)\in\bm{\mathcal{D}}$ and $\sigma\in\partial\bm{\mathcal{D}}$. The Green function $G(x,y)$ on $\bm{\mathcal{D}}$ satisfies $\Delta\psi(x)=\delta^{3}(x-y)$.
\end{thm}
For a proof see [5].
\begin{lem}$\bm{Green~function~for~an~n-ball}$~
Given a ball $\bm{\mathcal{B}}_{R}(p)\subset\bm{\mathrm{R}}^{3}$ of radius $R$ centred at $p$, the Green function has the form
\begin{align}
&G(x,y)=-\frac{1}{4\pi|x-y|}+\frac{1}{4\pi}\frac{R}{  \left |x-p\right|\left|\frac{R^{2}}{|x-p|^{2}}(x-p)-(y-p)\right|},~~x\neq p\\&
G(p,y)=-\frac{1}{4\pi|y-p|}+\frac{1}{4\pi R}
\end{align}
Also, if $x\in\bm{\mathcal{B}}_{R}(p)$ and $\sigma\in\partial b,{B}_{R}(p)$ then
\begin{equation}
\nabla_{\hat{N}(\sigma)}G(x,\sigma)=\frac{R^{2}-|x-p|^{2}}{4\pi R} \frac{1}{|x-\sigma|^{3}}
\end{equation}
For $\bm{\mathcal{B}}_{R}(p)\subset\bm{\mathrm{R}}^{n}$
\begin{equation}
\nabla_{\hat{N}(\sigma)}G(x,\sigma)=\frac{R^{2}-|x-p|^{2}}{\omega_{n} R} \frac{1}{|x-\sigma|^{3}}
\end{equation}
\end{lem}
For a proof see Speck notes.
\begin{thm}($\mathbf{Poisson~formula~for~harmonic~functions~in~an~Euclidean~ball}$).
Let $\bm{\mathcal{B}}_{R}(p)\subset\bm{\mathrm{R}}^{3}$ be a ball of radius $R$ centred at $p=(p_{1},p_{2},p_{3})$ and let $x\in\bm{\mathrm{R}}^{3}$. Let
$g\in C(\partial \bm{\mathcal{B}}_{R}(p)$. Then the unique solution
$\psi\in C^{2}(\bm{\mathcal{B}}_{R}(p))\cap C(\bm{\mathcal{B}}_{R}(p)$ of the PDE
\begin{align}
&\Delta \psi(x)=0,~x\in\bm{\mathcal{B}}_{R}(p)\\&
\psi(x)=g(x),~x\in\partial\bm{\mathcal{B}}_{R}(p)
\end{align}
is represented by the Poisson formula in $\bm{\mathrm{R}}^{3}$.
\begin{equation}
\psi(x)=\frac{R^{2}-|x-p|^{2}}{4\pi R}\int_{\partial\bm{\mathcal{B}}_{R}(p)}\frac{ g(\sigma)d\sigma}{|x-\sigma|^{3}}
\end{equation}
If $\bm{\mathcal{B}}_{R}(p)\subset\bm{\mathrm{R}}^{n}$ then
\begin{equation}
\psi(x)=\frac{R^{2}-|x-p|^{2}}{\|\partial\bm{\mathcal{B}}_{1}(p)\|R}
\int_{\partial\bm{\mathcal{B}}_{R}(p)}\frac{ g(\sigma)d\sigma}{|x-\sigma|^{n}}
\end{equation}
\end{thm}
\begin{proof}
The solution follows readily from Theorem (2.19) and Lemma (2.20). When $f(x)=0$ equation (2.74) becomes
\begin{equation}
\psi(x)=\int_{\partial\bm{\mathcal{D}}}g(\sigma)
\nabla_{\hat{N}(\sigma)}G(x,\sigma)d\sigma
\end{equation}
Substituting (2.77) then gives (2.81)
\end{proof}
\section{Stochastic extensions of classical theorems and estimates}
In this section, the classical theorems, estimates and BV problems of potential theory discussed in Section 2 are extended and modified to incorporate random or 'noisy' domains. These are tentatively defined as domains within which there also exists (regulated) Gaussian random scalar fields--defined with respect to a probability triplet--at all points within the domain and/or on its boundary. PDES and harmonic functions, and also boundary data, existing within the domain or on its boundary are then randomly perturbed by the GRSFs. Physically, this can correspond to random medias or the effects of noise or 'turbulence' on systems described by the LE or PE. GRSFs are also discussed in more detail in Appendix A.
\begin{prop}
Let $\bm{\mathcal{D}}\subset\mathbf{R}^{n}$ be a closed or open domain with boundary $\partial\bm{\mathcal{D}}$. Let $\mathscr{J}(x;\omega)\equiv {\mathscr{J}(x)}$ be a regulated GRSF defined with respect to a probability triplet $(\bm{\Omega},\mathcal{F},\bm{\mathsf{P}})$ and which exists for all $x\in\bm{\mathcal{D}}$ and /or $x\in\partial\bm{\mathcal{D}}$. the derivatives and integrals of ${\mathscr{J}(x)}$ exist. Then a 'noisy' or random domain is defined as the set
\begin{equation}
\bm{\mathfrak{D}}=\lbrace \bm{\mathcal{D}},{\mathscr{J}(x)},(\Omega,\mathscr{F},\bm{\mathsf{P}},
x\in\bm{\mathcal{D}}\bigcup\partial\bm{\mathcal{D}}~or~x\in\partial\bm{\mathcal{D}})\rbrace
\end{equation}
Similarly, a 'noisy' or random Euclidean ball $\bm{\mathfrak{B}}_{R}(0)$ is defined as
\begin{equation}
\bm{\mathfrak{B}}_{R}(0)=\lbrace \bm{\mathcal{B}}_{R}(0),\mathscr{J}(x),
(\Omega,\mathscr{F},\bm{\mathsf{P}}x\in\bm{\mathcal{B}}\bigcup\partial\mathcal{B}~or~x\in\partial\bm{\mathcal{B}})\rbrace
\end{equation}
\end{prop}
The Dirichlet BV problem for the Poisson equation in a random/noisy domain is defined as follows.
\begin{prop}
For $\psi\in C^{2}\bm{\mathcal{D}}$ and $(f,g):\bm{\mathcal{D}}\rightarrow{\mathbf{R}}$, the usual DBVP is $\Delta\psi(x)=f(x), x\in\bm{\mathcal{D}}$ and $\psi(x)=g(x)$. The possible DBVPs within the random domain $\bm{\mathfrak{D}}$ are then:
\begin{enumerate}
\item Noisy or randomly perturbed boundary data
\begin{align}
&\Delta\psi(x)=f(x),~x\in\bm{\mathcal{D}}\\&
\overline{\psi(x)}=g(x)+\lambda{\mathscr{J}(x)},~x\in\partial\bm{\mathcal{D}}
\end{align}
where $\lambda>0$;
\item Noisy/perturbed boundary data and noisy/perturbed source term so that
\begin{align}
&\Delta\overline{\psi(x)}=f(x)+\lambda {\mathscr{J}(x)},~x\in\bm{\mathcal{D}}\\&
\overline{\psi(x)}=g(x)+\lambda{\mathscr{J}(x)},~x\in\partial\bm{\mathcal{D}}
\end{align};
\item Noisy/perturbed source term only
\begin{align}
&\Delta\overline{\psi(x)}=f(x)+\eta\mathscr{J},~x\in\bm{\mathcal{D}}\\&
\psi(x)=g(x),~x\in\partial\bm{\mathcal{D}}
\end{align}
\end{enumerate}
\end{prop}
\begin{cor}
The DBVP for the Laplace equation with noisy boundary data follows from setting $f(x)=0$ so that
\begin{align}
&\Delta\overline{\psi(x)}=0 ,~x\in\bm{\mathcal{D}}\\&
\overline{\psi(x)}=g(x)+\lambda\mathscr{J}(x),~x\in\partial\bm{\mathcal{D}}
\end{align}
\end{cor}
\subsection{Solutions of PDEs from Brownian motion and the Feynman-Kac formula}
Before proceeding, the well-established and well-known connections between parabolic and elliptic PDEs and Brownian motion should be briefly reviewed. In 1944, Kakutani $\bm{[49]}$. showed that the Dirichlet problem for the Laplace equation in ${\mathbf{R}}^{2}$ can be solved using random walks. Given a point $x\in\bm{\mathcal{D}}$ in the interior of $\bm{\mathcal{D}}$, generate random walks that start at $x$ and which end when they reach the boundary $\partial\bm{\mathcal{D}}$. Then one computes the average of the values of the given function at these boundary points. This average value is then approximately equal to the value of the solution to the Dirichlet problem at the point. Feynman and Kac [37] later established that solutions of parabolic PDE can be expressed in terms of an underlying Brownian motion or Wiener process.
\begin{thm}($\bm{Feynman-Kac~formula}$)~\newline
Let $\bm{\mathcal{D}}\subset\mathbf{R}^{3}$ be a domain or open subset. For all $x\in\bm{\mathcal{D}}$ and $t\in[0,T]$, suppose the function $\psi(x,t)$ satisfies the PDE
\begin{equation}
\partial_{t}\psi(x,t)+\alpha(x,t)\nabla_{i}\psi(x,t)+\tfrac{1}{2}\beta^{2}(x,t)\Delta\psi(x,t)
-V(x,t)\psi(x,t)=f(x,t)=0
\end{equation}
with $(\alpha,\beta,\psi,V,f):\bm{\mathcal{D}}\rightarrow{\mathbf{R}}$ and the final condition  $\psi(x,T)=\varphi(x)$ some suitable boundary conditions. Let $X(t)$ be an Ito process satisfying the stochastic differential equation
\begin{equation}
dX(t)=\alpha(X,t)dt+\beta(X,t)d{\mathscr{W}}(t)
\end{equation}
where $\mathscr{W}(t)$ is a standard Brownian motion or Wiener process. Then the solution of the PDE is given by the Feynman-Kac formula
\begin{align}
&\psi(x,t)={{\mathbf{E}}}\bigg\llbracket\bigg[\int_{0}^{T}\exp\bigg(-\int_{t}^{s}
V(X(\tau),\tau)f(X(s),s)ds\nonumber\\&+\exp\bigg(-\int_{t}^{T}V(X(\tau),\tau)
d\tau\bigg)\varphi(X(T))
\bigg|X(t)=x\bigg]\bigg\rrbracket
\end{align}
If $\alpha=V=0$ then (-)reduces to a pure diffusion of the form $dX(t)=\alpha(x,t)dt+\beta(X,t)d\mathscr{W}(t)$ which has the generator $G=\tfrac{1}{2}\beta(X,t)\Delta$. The PDE (4.11) is then the heat equation
\begin{equation}
\partial_{t}\psi(x,t)+Q\Delta\psi(x,t)=-f(x,t)
\end{equation}
with FK solution
\begin{equation}
\psi(x,t)={\mathbf{E}}\bigg\llbracket\bigg[\int_{0}^{T}f(X(s),s)ds+
\psi(X(\tau)) X(t)=x\bigg]\bigg\rrbracket
\end{equation}
which is $\psi(x,t)={{\mathbb{E}}}\big\llbracket\psi(X(\tau))| X(t)=x\big]\big\rrbracket $
\end{thm}
The is proof is detailed but can be found in various texts. (refs).
\begin{thm}
The Dirichlet BV problem for the Poisson equation is given by (1.35) so that
$\Delta\psi(x)=-f(x),~x\in\bm{\mathcal{D}}$ and $\psi(x)=g(x)$ for $x\in\bm{\mathcal{D}}$ and $(f,g):\bm{\mathcal{D}}\rightarrow{\mathbf{R}}$ and $\psi\in C^{2}(\bm{\mathcal{D}}$. Let $\tau_{2\bm{\mathcal{D}}}$ be the 'first passage time' namely
\begin{equation}
\tau_{2\partial\bm{\mathcal{D}}}=\inf\lbrace t:\mathlarger{\mathscr{W}}(t)\in\partial\bm{\mathcal{D}}\rbrace
\end{equation}
with $\mathlarger{{\mathbf{E}}}\tau_{2\partial\bm{\mathcal{D}}}<\infty$. The 'first passage location' on the boundary $\partial\bm{\mathcal{D}}$ is $\mathscr{W}(\tau_{2\bm{\mathcal{D}}})$. The solution is
\begin{equation}
\psi(x)={{\mathbf{E}}}\bigg\llbracket
\bigg[\int_{0}^{\tau_{\partial\bm{\mathcal{D}}}}f(\mathscr{W}(t))dt
\bigg]\bigg\rrbracket+{\mathbf{E}}\bigg\llbracket f(\mathscr{W}(\tau_{2\bm{\mathcal{D}}}))
\bigg\rrbracket
\end{equation}
When $f=0$ this is the Dirichlet BV problem for the Laplace equation and the solution is
\begin{equation}
\psi(x)={\mathbf{E}}\bigg\llbracket f(\mathscr{W}(\tau_{2\bm{\mathcal{D}}}))
\bigg\rrbracket
\end{equation}
\end{thm}
\subsection{Stochastically averaged mean value properties}
The mean value property was discussed in Section 2. If $\bm{\mathcal{B}}_{R}(x)$ is ball with centre $x$ and radius $R$ and volume $\|\bm{\mathcal{B}}_{R}(x)\|$ then for any harmonic function $\psi(x)$
\begin{equation}
\psi(x)=\frac{1}{\bm{\mathcal{B}}_{R}(x)\|}\int_{\bm{\mathcal{B}}_{R}(x)}\psi(y)d^{n}y,~~(x,y)
\in\bm{\mathcal{B}}_{R}(x)
\end{equation}
This theorem can be extended to include randomly perturbed harmonic fields in a ball.
\begin{thm}
Let $\bm{\mathfrak{B}}_{R}(x)=\lbrace
\bm{\mathcal{B}}_{R}(x),{{\mathscr{J}}}(x),(\bm{\Omega},\mathcal{F},
\bm{\mathrm{I\!P}})\rbrace $ be a random or 'noisy' ball such that for all $x\in\bm{\mathcal{B}}_{R}(0),~\exists$ GRSF $\mathscr{J}(x)$. The derivatives $|\nabla {{\mathscr{J}}(x)}|$ and integrals $\int {{\mathscr{J}(x)}}d\mu(x)$ exist regulated covariances are $\mathbf{E}\llbracket\mathscr{J}(x)\rrbracket=0$, $\mathbf{E}\llbracket\mathscr{J}(x)\otimes \mathscr{J}(x)\rrbracket=\alpha$ and $\mathbf{E}\llbracket\mathscr{J}(x)\otimes \mathscr{J}(y)\rrbracket=\beta K(x,y;\epsilon)$. If the harmonic function is randomly perturbed then it becomes GRSF
\begin{equation}
\overline{\psi(x)}=\psi(x)+\lambda{\mathscr{J}(x)}
\end{equation}
(Note: all randomly perturbed quantities will be denoted by an overline.) Equation (4.20 ) becomes
\begin{equation}
\overline{\psi(x})=\frac{1}{\|\bm{{\mathcal{B}}}_{R}(x)\|}\int_{\bm{{\mathcal{B}}}_{R}}
\overline{\psi(y)}d^{n}y,~~(x,y)\in{\bm{\mathcal{B}}}_{R}(x)
\end{equation}
Then:
\begin{enumerate}
\item The stochastic average is $\mathbf{E}\llbracket\overline{\psi(x})\rrbracket=\psi(x)$
\item If $(x,x',y,y^{\prime})\in\bm{\mathcal{B}}_{R}(x)$ then the 2-point covariance is
\begin{align}
\mathbf{E}\bigg\llbracket\overline{\psi(x)}\otimes\overline{\psi(x^{\prime})}\bigg
\rrbracket=\psi(x)\psi(x^{\prime})+\frac{\lambda^{2}}{\|\bm{\mathcal{B}}_{R}(x)\|}\|
{\bm{\mathcal{B}}_{R}(x)\|}
\alpha\int_{\bm{\mathcal{B}}_{R}(x)}\int_{\bm{\mathcal{B}}_{R}(x)} K(y,y^{\prime};\epsilon)d^{n}yd^{n}y^{\prime}
\end{align}
\item The volatility of $\psi(x)$ is
\begin{align}
\mathbb{M}_{2}(x)\equiv
\mathbb{V}(x)=\mathbf{E}\llbracket |\overline{\psi(x)}|^{2}\rrbracket=|\psi(x)|^{2}+
\alpha\lambda\|\bm{\mathcal{B}}_{R}(x)\|^{2}<\infty
\end{align}
so that the volatility grows with the square of the volume of the ball.
\end{enumerate}
\end{thm}
\begin{proof}
To prove (1), the expectation of the GRSF is
\begin{align}
&\mathbf{E}\big\llbracket\overline{\psi(x)}\big\rrbracket
=\frac{1}{\|\bm{\mathcal{B}}_{R}(x)\|}\int_{\bm{\mathcal{B}}_{R}}
{\mathbf{E}}\bigg\llbracket
\overline{\psi(y)}\bigg\rrbracket d^{n}y\nonumber\\&=\frac{1}{\|\bm{\mathcal{B}}_{R}(x)\|}\int_{\bm{\mathcal{B}}_{R}(x)}
\psi(y)d^{n}y+\frac{\lambda}{\|\bm{\mathcal{B}}_{R}(x)\|}\int_{\bm{\mathcal{B}}_{R}(x)}
{{\mathbf{E}}}\bigg\llbracket\mathscr{J}(x)\bigg\rrbracket d^{n}y\nonumber\\&=\frac{1}{\|\bm{\mathcal{B}}_{R}(x)\|}\int_{\bm{\mathcal{B}}_{R}(x)}
\psi(y)d^{n}y=\psi(x)
\end{align}
To prove (2), the 2-point covariance, the randomly perturbed fields at points $(x,x^{\prime})\in\bm{\mathcal{B}}_{R}(x)$ are
\begin{align}
&\overline{\psi(x)}=\frac{1}{\|\bm{\mathcal{B}}(x)\|}\int_{{\mathcal{B}}_{R}(0)}(\psi(y)+\lambda {\mathscr{J}(y)})d^{n}y\\&
\overline{\psi(x^{\prime})}=\frac{1}{\|\bm{\mathcal{B}}(x)\|}\int_{{\mathcal{B}}_{R}(x)}
\big(\psi(y^{\prime})+\lambda {\mathscr{J}(y^{\prime})})d^{n}y^{\prime}
\end{align}
These are integrated over the same ball $\bm{{\mathcal{B}}}_{R}(x)$ although they could be integrated over balls $\bm{\mathcal{B}}_{R}(x)$ and $\bm{{\mathcal{B}}}_{R}(x^{\prime})$ with $\bm{\mathcal{B}}_{R}(x)\cap\bm{\mathcal{B}}_{R}(x^{\prime})=\varnothing $ or $\bm{\mathcal{B}}_{R}(x)\cap\bm{\mathcal{B}}_{R}(x^{\prime})\ne \varnothing $. The 2-point covariance is then
\begin{align}
{{\mathbf{E}}}\bigg\llbracket
\overline{\psi(x)}\otimes \overline{\psi(x^{\prime})}\bigg\rrbracket& =\frac{1}{\|\bm{\mathcal{B}}_{R}(x)\|^{2}}\mathlarger{\iint}_{\bm{\mathcal{B}}_{R}(0)}\big(\psi(y)+\lambda
{\mathscr{J}}(y))\big(\psi(y^{\prime})+\lambda {\mathscr{J}(y^{\prime})}\big)d^{n}yd^{n}y^{\prime}\nonumber\\&=
\frac{\lambda}{\|\bm{\mathcal{B}}_{R}(x)\|^{2}}{\mathlarger{\iint}}_{{\mathcal{B}}_{R}(0)}\psi(y)\psi(y^{\prime})d^{n}yd^{n}y^{\prime}\nonumber\\&
+\frac{\lambda}{\|\bm{\mathcal{B}}_{R}(x)\|^{2}}\iint_{\bm{\mathcal{B}}_{R}(0)}
\psi(y){{\mathbf{E}}}\bigg\llbracket{\mathscr{J}(y^{\prime})}\bigg\rrbracket d^{n}yd^{n}y^{\prime}\nonumber\\&+\frac{1}{\|\bm{\mathcal{B}}_{R}(x)\|^{2}}
\mathlarger{\iint}_{\bm{\mathcal{B}}_{R}(0)}\psi(y^{\prime})\mathbf{E}\big\llbracket {\mathscr{J}}(x)\bigg\rrbracket d^{n}yd^{n}y^{\prime}\nonumber\\&+\frac{\lambda^{2}}{\|\bm{\mathcal{B}}_{R}(x)\|^{2}}
\iint_{\bm{\mathcal{B}}_{R}(0)}\psi(y^{\prime}){\mathbf{E}}\bigg
\llbracket{\mathscr{J}(y)}\otimes{\mathscr{J}(y^{\prime})}\bigg\rrbracket d^{n}yd^{n}y^{\prime}\nonumber\\&=\frac{1}{\|\bm{\mathcal{B}}_{R}(x)\|^{2}}\iint_{\bm{\mathcal{B}}_{R}(0)}\psi(y)\psi(y^{\prime})d^{n}yd^{n}y^{\prime}\nonumber\\&
+\frac{\lambda^{2}}{\|\bm{\mathcal{B}}_{R}(x)\|^{2}}\iint_{\bm{\mathcal{B}}_{R}(0)}
\psi(y^{\prime}){\mathbf{E}}\bigg\llbracket{\mathscr{J}(y)}\otimes {\mathscr{J}(y^{\prime})}\bigg\rrbracket d^{n}yd^{n}y^{\prime}\nonumber\\&
=\psi(x)\psi(x^{\prime})+\frac{\lambda^{2}}{\|\bm{\mathcal{B}}_{R}(x)\|^{2}}\iint_{\bm{\mathcal{B}}_{R}(0)}
\psi(y^{\prime})K(y,y^{\prime})d^{n}yd^{n}y^{\prime}\nonumber\\&=\psi(x)\psi(x^{\prime})
+\frac{\alpha\lambda^{2}}{\|\bm{\mathcal{B}}_{R}(x)\|^{2}}\iint_{\bm{\mathcal{B}}_{R}(0)}
K(y,y^{\prime})d^{n}yd^{n}y^{\prime}
\end{align}
The volatility (3) then follows from taking the regulated limit $\lim_{(x,y)\rightarrow (x^{\prime},y^{\prime}}$ whereby $\alpha K(y,y^{\prime};\epsilon)=\alpha$.
\end{proof}
\begin{thm}(Binomial~Thm)
Let $(f(x),g(x))$ be two smooth functions for all $x\in\bm{\mathcal{D}}$. Then the binomial theorem is
\begin{equation}
(f(x)+g(x))^{P}\le \sum_{Q=1}^{P}\binom{P}{Q}|f(x)|^{P-Q}|g(x)|^{Q}
\end{equation}
\end{thm}
The binomial Theorem can be utilised to compute the moments.
\begin{thm}
The $P^{th}$-order moments or covariance is given by
\begin{align}
{{{\mathbf{M}}}}_{P}(x)\equiv
{{{\mathbf{E}}}}\bigg\llbracket\bigg\|\overline{\psi(x)}\bigg\|^{P}\bigg\rrbracket
&=\sum_{Q=1}^{P}\binom{P}{Q}\left(\frac{1}{\|\bm{\mathcal{B}}_{R}(x)\|}\int
_{\bm{\mathcal{B}}_{R}(x)}\psi(y)d^{n}y\right)^{P-Q}\lambda^{Q}\big[\tfrac{1}{2}(\alpha^{Q/2}+(-1)^{Q}\alpha^{Q/2}\big]
\nonumber\\&=\sum_{Q=1}^{P}\binom{P}{Q}|\psi(x)|^{P-Q}
\lambda^{Q}\big[\tfrac{1}{2}(\alpha^{Q/2}+(-1)^{Q}\alpha^{Q/2}\big]
\end{align}
so that all odd P moments vanish.
\end{thm}
\begin{proof}
The moments are estimated as
\begin{align}
{{\mathbf{E}}}\bigg\llbracket |\psi(x)|^{P}\bigg\rrbracket
&={\mathbf{E}}\bigg\llbracket(\frac{1}{\|\bm{\mathcal{B}}_{R}(x)\|}
\int_{\bm{\mathcal{B}}_{R}(x)}\psi(y)d^{n}y+\frac{\lambda}{\|\bm{\mathcal{B}}_{R}(x)\|}
\int_{\bm{\mathcal{B}}_{R}(x)}
\mathscr{J}(x)d^{n}y\bigg)^{P}\bigg\rrbracket\nonumber \\&=
{\mathbf{E}}\bigg\llbracket\sum_{Q=0}^{P}
\binom{P}{Q}\bigg(\frac{1}{\|\bm{\mathcal{B}}_{R}(x)\|}\int_{\bm{\mathcal{B}}_{R}(x)}
\psi(y)d^{n}y\bigg)^{P-Q}\bigg(\frac{\lambda}{\|\bm{\mathcal{B}}_{R}(x)\|}\int_{\bm{\mathcal{B}}_{R}(x)}
\mathscr{J}(x)d^{n}y\bigg)^{Q}\bigg\rrbracket\nonumber\\&=\sum_{Q=0}^{P}\binom{P}{Q}
\bigg(\frac{1}{\|\bm{\mathcal{B}}_{R}(x)\|}\int_{\bm{\mathcal{B}}_{R}(x)}
\psi(y)d^{n}y\bigg)^{P-Q}\mathscr{J}(x)\bigg\llbracket\bigg(\frac{\lambda}{\|\bm{\mathcal{B}}_{R}(x)\|}
\int_{\bm{\mathcal{B}}_{R}(x)}{\mathbf{E}}(x)d^{n}y\bigg)^{Q}\bigg\rrbracket\nonumber\\&\le C\sum_{Q=0}^{P}\binom{P}{Q}\bigg(\frac{1}{\|\bm{\mathcal{B}}_{R}(x)\|}\int_{\bm{\mathcal{B}}_{R}(x)}
\psi(y)d^{n}y\bigg)^{p-Q}\nonumber\\&\times\bigg(\frac{\lambda}{\|\bm{\mathcal{B}}_{R}(x)\|}\bigg)^{Q}\int_{\bm{\mathcal{B}}_{R}(x)}...\int_{\bm{\mathcal{B}}_{R}(x)}
{\mathbf{E}}\bigg\llbracket\underbrace{{\mathscr{J}}(x)}\otimes...
\otimes{\mathscr{J}(x)}_{Q~times}\bigg\rrbracket\underbrace{d^{n}y...d^{n}y}_{p~times} \nonumber\\&=C\sum_{Q=0}^{P}\binom{P}{Q}\bigg(\frac{1}{\|\bm{\mathcal{B}}_{R}(x)\|}\int_{\bm{\mathcal{B}}_{R}(x)}
\psi(y)d^{n}y\bigg)^{P-Q}\nonumber\\&\times\bigg(\frac{\lambda}{\|\bm{\mathcal{B}}_{R}(x)\|}
\bigg)^{Q}\int_{\bm{\mathcal{B}}_{R}(x)}...
\int_{\bm{\mathcal{B}}_{R}(x)}\bigg[\frac{1}{2}(\alpha^{Q/2}+(-1)^{Q}\alpha
^{Q/2}\bigg]\underbrace{d^{n}y...d^{n}y}_{P~times} \nonumber\\&
=C\sum_{Q=0}^{P}\binom{P}{Q}\bigg(\frac{1}{\|\bm{\mathcal{B}}_{R}(x)\|}\int_{\bm{\mathcal{B}}_{R}(x)}
\psi(y)d^{n}y\bigg)^{P-Q}\bigg(\frac{\lambda}{\|\bm{\mathcal{B}}_{R}(x)\|}\bigg)^{Q}
\nonumber\\&\bigg[\frac{1}{2}\bigg(\alpha^{Q/2}+(-1)^{Q}\alpha^{Q/2}\bigg)\bigg]
\bigg(\int_{\bm{\mathcal{B}}_{R}(0)}d^{n}y\bigg)^{Q}\nonumber\\&
=C\sum_{Q=0}^{P}\binom{P}{Q}\bigg(\frac{1}{\|\bm{\mathcal{B}}_{R}(x)\|}\int_{\bm{\mathcal{B}}_{R}(x)}
\psi(y)d^{n}y\bigg)^{P-Q}\bigg(\frac{\lambda}{\|\bm{\mathcal{B}}_{R}(x)\|}\bigg)^{Q}\nonumber\\&
\bigg[\frac{1}{2}\bigg(\alpha^{Q/2}+(-1)^{Q}\alpha^{Q/2}\bigg)\bigg]
\|\bm{\mathcal{B}}_{R}(x)\|^{Q}\nonumber\\&=C\sum_{Q=0}^{P}\binom{P}{Q}\bigg(\frac{1}{\|\bm{\mathcal{B}}_{R}(x)\|}
\int_{\bm{\mathcal{B}}_{R}(x)}\psi(y)d^{n}y\bigg)^{P-Q}\lambda^{Q}\bigg[\frac{1}{2}\bigg(\alpha^{Q/2}+(-1)^{Q}\alpha^{Q/2}\bigg)\bigg]\nonumber\\&
=C\sum_{Q=0}^{P}\binom{P}{Q}|\psi(x)|^{P-Q}\lambda^{Q}\bigg[\frac{1}{2}\bigg(\alpha^{Q/2}+(-1)^{Q}\alpha^{Q/2}\bigg)\bigg]<\infty
\end{align}
so that the (even) moments to all orders are always finite.
\begin{cor}
The volatility is for p=2 and agrees with (4.27) since
\begin{align}
&{\bm{\mathsf{M}}}_{2}(x)={\bm{\mathsf{V}}}(x)={\mathbf{E}}\bigg
\llbracket|\psi(x)|^{2}\rrbracket=\sum_{Q=0}^{2}|\binom{p}{2}\psi(x)|^{2-Q}\lambda^{Q}
\bigg[\frac{1}{2}\bigg(\alpha^{Q/2}+(-1)^{Q}\alpha^{Q/2}\bigg)\bigg]<\infty\nonumber\\&
=\binom{2}{0}|\psi(x)|^{2}[\tfrac{1}{2}+\tfrac{1}{2}]+\lambda\binom{2}{1}|\psi(x)|
\big[\alpha^{1/2}-\alpha^{1/2}\big]+\lambda^{2}\binom{2}{2}[\tfrac{1}{2}\alpha+\tfrac{1}{2}\alpha]\\&
=|\psi(x)|^{2}+\lambda^{2}\alpha
\end{align}
which agrees with (4.23).
\end{cor}
\end{proof}
\subsection{Averaged maximum principle for randomly perturbed harmonic functions}
The maximum principle extends to randomly perturbed or noisy harmonic functions.
\begin{thm}
Let $\bm{\mathfrak{D}}=(\bm{\mathcal{D}},\mathscr{J}(x),
(\Omega,\mathscr{F},\bm{\mathrm{I\!P}}))$ be a noisy or random domain and $\psi(x)$ is a harmonic function within $\bm{\mathcal{D}}$ that obeys the strong maximum principle. Let $\overline{\psi(x)}=\psi(x)+\lambda{\mathscr{J}(x)}$ be the randomly perturbed or noisy harmonic function in $\bm{\mathcal{D}}$. Then the stochastically averaged maximum principle holds such that
\begin{align}
&\mathbf{E}\big\llbracket\overline{\psi(x)}\big\rrbracket ~<~ \max{\mathbf{E}}\big\llbracket\overline{\psi(y)}\big\rrbracket,~y\in\partial\bm{\mathcal{D}}
\\&{\mathbf{E}}\big\llbracket\overline{\psi(x)}\rrbracket ~>~\min \mathbf{E}\big\llbracket\overline{\psi(y)}\big\rrbracket,~y\in\partial
\bm{\mathcal{D}}
\end{align}
or equivalently
\begin{align}
&\mathbf{E}\big\llbracket|\overline{\psi(x)}|^{2}\rrbracket ~<~ \max\mathbf{E}\big\llbracket|\overline{\psi(y)}|^{2}\big\rrbracket~y\in\partial
\bm{\mathcal{D}}\\&\mathbf{E}\big\llbracket|\overline{\psi(x)}|^{2}\big\rrbracket ~>~\min\mathbf{E}\big\llbracket|\overline{\psi(y)}|^{2}\big\rrbracket,~y\in\partial\bm{\mathcal{D}}
\\&\mathbf{E}\big\llbracket|\overline{\psi(x)}|^{p}\rrbracket ~<~ \max\mathbf{E}\big\llbracket|\overline{\psi(y)}|^{p}
\big\rrbracket~y\in\partial\bm{\mathcal{D}}\\& \mathbf{E}\big\llbracket|\overline{\psi(x)}|^{p}\big\rrbracket ~>~ \min{\mathbf{E}}\big\llbracket|\overline{\psi(y)}|^{p}\big\rrbracket,~y\in\partial
\bm{\mathcal{D}}
\end{align}
\end{thm}
\begin{proof}
If (3.36) and (3.37) hold then
\begin{align}
&{{\mathbf{E}}}\big\llbracket|\overline{\psi(x)}|^{P}\big\rrbracket=
\mathbf{E}\big\llbracket|\psi(x)+\lambda {\mathscr{J}(x)}|^{p}\big\rrbracket
=\mathbf{E}\bigg\llbracket\sum_{Q=0}^{P}\binom{P}{Q}
|\psi(x)|^{P-Q}|\lambda|^{Q}|{\mathscr{J}(x)}|^{Q}\bigg\rrbracket\nonumber\\&
=\sum_{Q=0}^{P}\binom{P}{Q}|\psi(x)|^{P-Q}|\xi|^{Q}\mathbf{E}
\big\llbracket|{\mathscr{J}(x)}|^{Q}
\big\rrbracket\nonumber\\&=\sum_{Q=0}^{p}\binom{P}{Q}
|\psi(x)|^{P-Q}|\lambda|^{Q}\big[\frac{1}{2}(\alpha^{Q/2}+(-1)^{Q}\alpha^{Q/2})\big]
\nonumber\\& \ge
{{\mathbf{E}}}\big\llbracket|\min\overline{\psi(y)}|^{P}\big\rrbracket\nonumber\\&=
{{\mathbf{E}}}\big\llbracket|\min\psi(y)+\lambda\min{\mathscr{J}(y)}|^{P}
\big\rrbracket=\mathbf{E}\bigg\llbracket\sum_{Q=0}^{P}\binom{P}{Q}
|\min\psi(y)|^{P-Q}|\lambda|^{Q}|\mathscr{J}(y)|^{Q}\bigg\rrbracket\nonumber
\\&=\sum_{Q=0}^{P}\binom{P}{Q}|\min\psi(y)|^{P-Q}|\lambda|^{Q}{\mathbf{E}}
\big\llbracket|\min {\mathscr{J}(y)}|^{Q}\big\rrbracket
\nonumber\\&=\sum_{Q=0}^{P}\binom{P}{Q}|\min\psi(y)|^{P-Q}|\xi|^{Q}
\big[\frac{1}{2}(\alpha^{Q/2}+(-1)^{Q}\alpha^{Q/2})\bigg]
\end{align}
which is equivalent to
\begin{align}
&\sum_{Q=0}^{P}\binom{P}{Q}
|\psi(x)|^{P-Q}|\lambda|^{Q}\bigg[\frac{1}{2}(\alpha^{Q/2}+(-1)^{Q}\alpha^{Q/2})\bigg]\nonumber\\&\ge
\sum_{Q=0}^{P}\binom{P}{Q}
|\min\psi(y)|^{P-Q}|\lambda|^{Q}\bigg[\frac{1}{2}(\alpha^{Q/2}+(-1)^{Q}\alpha^{Q/2})\bigg]
\end{align}
or
\begin{align}
&\sum_{Q=0}^{P}\binom{P}{Q}
|\psi(x)|^{P-Q}\ge
\sum_{Q=0}^{P}\binom{P}{Q}
|\min\psi(y)|^{P-Q}
\end{align}
which holds for any $P\ge 1$ if the original maximum principle holds, which is does for all harmonic functions $\psi(x)$.
\end{proof}
\begin{lem}($\underline{Stochastic~gradient~theorem}$)\newline
Let $\bm{\mathcal{C}}\subset\bm{\mathcal{D}}\subset\bm{\mathrm{R}}^{n}$ be a curve with endpoints $A$ and $B$. Let $\psi(x)$ be a harmonic function on $\bm{\mathcal{D}}$. If $\exists$ SGRF ${\mathscr{J}(x)}$ then the noisy or randomly perturbed field is $\overline{\psi(x)}=\psi(x)+\mathscr{J}$. The line integral of the gradient is
\begin{align}
&\mathlarger{\mathrm{L}}(\bm{\mathcal{C}})=\int_{\bm{\mathcal{C}}}\nabla_{i}
\overline{\psi(x)}dx^{i}=\int_{\bm{\mathcal{C}}}\nabla_{i}\psi(x)dx^{i}+\int_{\bm{\mathcal{C}}}\nabla_{i}
{\mathscr{J}(x)}dx^{i}\nonumber\\&
=|\psi(B)-\psi(A)|+\int_{\bm{\mathcal{C}}}\nabla_{i}{\mathscr{J}}(x)dx^{i}
\end{align}
The expectation is then
\begin{align}
\mathbf{E}\llbracket\overline{\mathlarger{\mathrm{L}}(\bm{\mathcal{C}})}\rrbracket
&=\bigg\llbracket\int_{\bm{\mathcal{C}}}\nabla_{i}\overline{\psi(x)}dx^{i}\bigg\rrbracket
=|\psi(B)-\psi(A)|+\mathbf{E}\bigg\llbracket\int_{\bm{\mathcal{C}}}\nabla_{i}
\mathscr{J}(x)dx^{i}\bigg\rrbracket\nonumber\\&
=\psi(B)-\psi(A)
\end{align}
For a closed (loop) integral
\begin{align}
\mathlarger{\bm{\mathsf{L}}}(\bm{\mathcal{C}})=\mathbf{E}
\llbracket\overline{L(\bm{\mathcal{C}})}\rrbracket&=\bigg\llbracket \oint_{\bm{\mathcal{C}}}\nabla_{i}\overline{\psi(x)}dx^{i}\bigg\rrbracket
=|\psi(B)-\psi(B)|+\mathbf{E}\bigg\llbracket
\oint_{\bm{\mathcal{C}}}\nabla_{i}\mathscr{J}(x)dx^{i}\bigg
\rrbracket\nonumber\\&=\psi(B)-\psi(B)=0
\end{align}
\end{lem}
The proof follows from $\mathbf{E}\llbracket\nabla_{i}\mathscr{J}(x)\rrbracket=0$.
There is a correlation between two line integrals
\begin{lem}
Let ${\mathcal{C}},{\mathcal{C}}^{\prime}\subset\bm{\mathcal{D}}$ be two curves with endpoints $(P,Q)$ and $(P^{\prime},Q^{\prime})$ If $x\in{\mathcal{C}}$ and
$x^{\prime}\in\bm{\mathcal{C}}^{\prime}$. If ${\mathbf{E}}\llbracket
\nabla_{i}{\mathscr{J}}(x)\otimes\mathscr{J}(y)\rrbracket=\delta_{ij}K(x,y;\epsilon)$
\begin{align}
&\overline{{\mathrm{L}}({\mathcal{C}})}=\int_{\bm{\mathcal{C}}}\nabla_{i}\psi(x)dx^{i}+\int_{\bm{\mathcal{C}}}\nabla_{i}
\mathscr{J}(x)dx^{i}=|\psi(B)-\psi(A)|+\int_{{\mathcal{C}}}\nabla_{i}
\mathscr{J}(x)dx^{i}\\&\overline{\mathlarger{\mathrm{L}}({\mathcal{C}^{\prime}})}
=\int_{\bm{\mathcal{C}^{\prime}}}\nabla_{i}\psi(x^{\prime})dx^{\prime~i }+\int_{\bm{\mathcal{C}}}\nabla_{i}\mathscr{J}(x^{\prime})dx^{\prime i}
=|\psi(B^{\prime})-\psi(A^{\prime})|+\int_{\bm{\mathcal{C}^{\prime}}}\nabla_{i}
{\mathscr{J}(x^{\prime})}dx^{\prime~i}
\end{align}
Then the correlation between the line integrals is
\begin{align}
\mathbf{E}\llbracket{\mathrm{L}}(\mathcal{C})\otimes {\mathrm{L}}(\mathcal{C}^{\prime})\rrbracket&=\mathlarger{\mathlarger{\iint}}_{\mathcal{C},\mathcal{C}^{\prime}}
\nabla_{i}\psi(x)\nabla_{j}\psi(x)+{\mathlarger{\iint}}_{\mathcal{C},\mathcal{C}^{\prime}}
\delta_{ij}K(x,x^{\prime};\epsilon)dx^{i}dx^{\prime~j}
\nonumber\\&=(\psi(B)-\psi(A))(\psi(B^{\prime})-\psi(A^{\prime}))+
{\mathlarger{\iint}}_{\mathcal{C},\mathcal{C}^{\prime}}K(x,x^{\prime};\epsilon)dx^{i}dx^{j}
\end{align}
For closed loop integrals
\begin{align}
\mathbf{E}\llbracket\mathlarger{\mathrm{L}}(\mathcal{\mathcal{C}})\otimes \mathlarger{\mathrm{L}}(\mathcal{C}^{\prime})\rrbracket&=\mathlarger{\mathlarger{\oint\oint}}_{\mathcal{C},\mathcal{C}^{\prime}}
\nabla_{i}\psi(x)\nabla_{j}\psi(x)+\mathlarger{\mathlarger{\oint\oint}}_{\mathcal{C},\mathcal{C}^{\prime}}
\delta_{ij}K(x,x^{\prime};\epsilon)dx^{i}dx^{\prime~j}\nonumber\\&
=(\psi(B)-\psi(B))(\psi(B^{\prime})-\psi(B^{\prime}))+
\mathlarger{\mathlarger{\oint\oint}}_{\mathcal{C},\mathcal{C}^{\prime}}K(x,x^{\prime};\epsilon)dx^{i}dx^{j}
\nonumber\\&=\mathlarger{\mathlarger{\oint\oint}}_{\mathcal{C},\mathcal{C}^{\prime}}K(x,x^{\prime};\epsilon)dx^{i}dx^{j}
\end{align}
\end{lem}
\begin{proof}
\begin{align}
\mathbf{E}\llbracket\mathlarger{\mathrm{L}}(\mathcal{C})\otimes \mathlarger{\mathrm{L}}(\mathcal{C}^{\prime})\rrbracket&=\mathbf{E}\bigg\llbracket
\mathlarger{\oint\oint}_{\mathcal{C},\mathcal{C}^{\prime}}(\nabla_{i}\psi(x)
+\nabla_{i}{\mathscr{J}(x))}\otimes(\nabla_{j}\psi(y)+\nabla_{j}{\mathscr{J}(y))}d^{3}x d^{3}y\bigg\rrbracket\nonumber\\&
=\mathlarger{\mathlarger{\oint\oint}}_{\mathcal{C},\mathcal{C}^{\prime}}\bigg[\nabla_{i}\psi(x)\nabla_{j}\psi(y)
+\nabla_{i}\psi(x)\mathbf{E}\llbracket\nabla_{j}{\mathscr{J}(y)}\rrbracket
\nonumber\\&+\nabla_{j}\psi(y)\mathbf{E}\llbracket\nabla_{i}{\mathscr{J}(x)}\rrbracket+\nabla_{i}
\mathbf{E}
\llbracket\mathscr{J}(x)\otimes\nabla_{j}\mathscr{J}\rrbracket \bigg]d^{3}xd^{3}y\nonumber\\&=\mathlarger{\mathlarger{\oint\oint}}_{\mathcal{C},\mathcal{C}^{\prime}}\bigg[\nabla_{i}\psi(x)\nabla_{j}\psi(y)
+\mathbf{E}\bigg\llbracket\nabla_{i}\mathscr{J}(x)\otimes\nabla_{j}
\mathscr{J}(y)\bigg\rrbracket\bigg]d^{3}xd^{3}y
\end{align}
and the result follows
\end{proof}
\subsection{Stochastically averaged Dirichlet energy integral}
The Dirichlet energy integral(DEI) was previously defined and is essentially the Lagrangian for the Laplace equation, with harmonic functions being critical points.(Thm 3.1).We now consider a stochastically averaged Dirichlet energy integral(SADEI).
\begin{thm}
Let $\mathlarger{\mathcal{E}}[\psi]$ be the DEI given by (1.11) on a domain $\bm{\mathcal{D}}$.Let$\bm{\mathcal{D}}=\lbrace\bm{\mathcal{D}},{\mathscr{J}(x)},
(\Omega,\mathscr{F},\bm{\mathsf{P}})\rbrace$ be a random domain for which there exists a GRF ${\mathscr{J}(x)}$ with derivative $\nabla_{i}{\mathscr{J}(x)}$ for all $x\in\bm{\mathcal{D}}$. The expectations are $\mathbf{E}=0$ and $\mathbf{E}\llbracket{\mathscr{J}(x)}\otimes
{\mathscr{J}(x)}\rrbracket=\alpha$. Let $\psi(x)$ be a harmonic function on $\bm{\mathcal{D}}$ that is randomly perturbed as $\overline{\psi(x)}=\psi(x)+\lambda{\mathscr{J}(x)}$. The SADEI is
\begin{equation}                                                                             \mathbf{E}\llbracket{\mathcal{E}}[\overline{\psi}]\rrbracket
=\mathbf{E}[\overline{\psi}]=\mathbf{E}
\bigg\llbracket\int_{\bm{\mathcal{D}}}|\nabla_{i}\overline{\psi(x)}|^{2}d^{n}x\bigg
\rrbracket\equiv\int_{\bm{\mathcal{D}}}\mathbf{E}\bigg\llbracket|\nabla_{i}
\overline{\psi(x)}|^{2}\bigg\rrbracket d^{n}x
\end{equation}
or formally
\begin{equation}
\mathbf{E}[{\mathcal{E}}[\overline{\psi}]
=\mathlarger{\bm{\mathcal{E}}}[\overline{\psi}]=\int_{\Omega}
\int_{\bm{\mathcal{D}}}|\nabla_{i}\overline{\psi(x)}|^{2}d^{n}x d\mathlarger{\mathsf{P}}(\omega)
\end{equation}
The following then hold:
\begin{enumerate}
\item The Dirichlet energy is shifted by a constant factor
\begin{equation}
\mathbf{E}[\mathlarger{\mathcal{E}}[\overline{\psi}]
={\mathlarger{\bm{\mathcal{E}}}}[\overline{\psi}]=\mathbf{E}
\left\llbracket\int_{\bm{\mathcal{D}}}|\nabla_{i}\overline{\psi(x)}|^{2}d^{n}x\right\rrbracket
=\int_{\bm{\mathcal{D}}}|\nabla_{i}\overline{\psi(x)}|^{2}d^{n}x+\lambda\alpha\|
\bm{\mathcal{D}}\|
\end{equation}
\item The deterministic harmonic functions $\psi(x)$ are still critical points of the SADEI $\mathbf{E}[\mathlarger{\mathcal{E}}[\overline{\psi}]$
\item If $\psi(x)$ is harmonic and $\varphi(x)$ is any other function on $\bm{\mathcal{D}}$ which is not harmonic then
\begin{equation}
\mathbf{E}
\left\llbracket\int_{\bm{\mathcal{D}}}|\nabla_{i}\overline{\psi(x)}|^{2}d^{n}x\right\rrbracket~~
<~~\mathbf{E}
\left\llbracket\int_{\bm{\mathcal{D}}}|\nabla_{i}\overline{\varphi(x)}|^{2}d^{n}x\right\rrbracket
\end{equation}
or $\bm{\mathcal{E}}[\overline{\psi}]<\mathlarger{\bm{\mathcal{E}}}[\overline{\varphi}]$
Hence, harmonic functions are still critical points of the SADEI.
\end{enumerate}
\end{thm}
\begin{proof}
To prove (1) simply expand and take the expectation
\begin{align}
&{{\mathbf{E}}}[\mathlarger{\mathcal{E}}[\overline{\psi}]
={\mathlarger{\bm{\mathcal{E}}}}[\overline{\psi}]=\mathbf{E}
\bigg\llbracket\int_{\bm{\mathcal{D}}}|\nabla_{i}\overline{\psi(x)}|^{2}d^{n}x\bigg
\rrbracket\equiv
\int_{\bm{\mathcal{D}}}\mathbf{E}\bigg\llbracket|\nabla_{i}
\overline{\psi(x)}|^{2}\bigg\rrbracket d^{n}x\nonumber\\&
={\mathbf{E}}\bigg\llbracket\int_{\bm{\mathcal{D}}}\bigg|\nabla_{i}\psi(x)+
\lambda\nabla_{i}{\mathscr{J}(x)}\bigg|^{2}d^{n}x\bigg\rrbracket\nonumber\\&
=\int_{\bm{\mathcal{D}}}\bigg(|\nabla_{i}\psi(x)|^{2}+2\lambda
\int \mathbf{E}
\big\llbracket\nabla_{i}\psi(x)\bigg\rrbracket d^{n}x+\lambda^{2}
\mathbf{E}\bigg\llbracket
 {\mathscr{J}(x)}\otimes {\mathscr{J}(x)}\bigg\rrbracket\bigg)d^{n}x\nonumber\\&
=\int_{\bm{\mathcal{D}}}\bigg(|\nabla_{i}\psi(x)|^{2}+\lambda^{2}
\mathbf{E}\bigg\llbracket{\mathscr{J}(x)}\otimes {\mathscr{J}(x)}\bigg\rrbracket\bigg)d^{n}x
\\&=\int_{\bm{\mathcal{D}}}\bigg(|\nabla_{i}\psi(x)|^{2}+\lambda^{2}\alpha\bigg)d^{n}x
=\int_{\bm{\mathcal{D}}}\bigg(|\nabla_{i}\psi(x)|^{2}\bigg)d^{n}x+\lambda^{2}
\alpha\|\bm{\mathcal{D}}\|
\end{align}
so that the DEI is shifted by a constant. To prove (2),let $\phi(x)$ is a scalar field that vanishes on $\partial\bm{\mathcal{D}}$ so that
\begin{equation}
\frac{d}{d\xi}{\mathlarger{\mathcal{E}}}[\psi,\phi,\xi]=
\le \int_{\bm{\mathcal{D}}}\varphi(x)\Delta\psi(x)d^{n}x
\end{equation}
If the fields are randomly perturbed as $\overline{\psi(x)}=\psi(x)+\lambda{\mathscr{J}(x)}$ and $\overline{\phi(x)}=\psi(x)+\lambda{\mathscr{J}(x)}$ then
\begin{align}
&\frac{d}{d\xi}\mathbf{E}
\bigg\llbracket \mathbf{E}[\overline{\psi},\overline{\phi},\xi]\bigg
\rrbracket=\le \mathbf{E}\bigg\llbracket
\int_{\bm{\mathcal{D}}}\overline{\varphi(x)}
\Delta\overline{\psi(x)}d^{n}x\bigg\rrbracket\nonumber\\&
=\mathbf{E}\bigg\llbracket\int_{\bm{\mathcal{D}}}\varphi(x)\Delta\psi(x)
+\lambda\varphi(x){\mathscr{J}(x)}+\lambda^{2}\mathscr{J}(x)
\otimes\mathscr{J}(x)\bigg]\nonumber\\&
=\int_{\bm{\mathcal{D}}}\varphi(x)\psi(x)+\lambda^{2}\mathbf{E}
\bigg\llbracket{\mathscr{J}(x)}\otimes\Delta {\mathscr{J}(x)}\bigg\rrbracket\nonumber\\&
=\int_{\bm{\mathcal{D}}}\varphi(x)\Delta\psi(x)+\lambda^{2}
\end{align}
which vanishes for $\Delta\psi(x)=0$ so that harmonic functions are still critical points. Finally, (3) follows readily from
\begin{align}
&\mathbf{E}
\bigg\llbracket\int_{\bm{\mathcal{D}}}|\nabla_{i}{\psi(x)}|^{2}d^{n}x\bigg\rrbracket
=\int_{\bm{\mathcal{D}}}|\nabla_{i}\psi(x)|^{2}d^{n}x+\lambda\alpha\|\bm{\mathcal{D}}\|
\nonumber\\&\le \mathbf{E}
\bigg\llbracket\int_{\bm{\mathcal{D}}}|\nabla_{i} {\phi(x)}|^{2}d^{n}x\bigg\rrbracket
=\int_{\bm{\mathcal{D}}}|\nabla_{i}\phi(x)|^{2}d^{n}x+\lambda\alpha\|\bm{\mathcal{D}}\|
\end{align}
so that again the Dirichlet energy is minimised.
\end{proof}
The volatility of the DEF at any point is finite and bounded.
\begin{lem}
If $\psi(x)$ is a harmonic function then $\mathcal{E}[\psi(x)]$ is minimised. Let $\bm{\mathsf{V}}[\epsilon(\psi(x))]$ be the volatility of $\mathcal{E}[\overline{\psi(x)}]$ so that
\begin{equation}
\mathbf{V}[\mathcal{E}(\overline{\psi(x)})]
=\mathbf{E}\bigg\llbracket\bigg|
\frac{1}{2}\int_{\bm{\mathcal{D}}}|\nabla_{i}\psi(x)|^{2}d^{3}x\bigg|^{2}\bigg\rrbracket
\end{equation}
where $\overline{\psi(x)}=\psi(x)+{\mathscr{J}(x)}$. The regulated covariance of the gradient of the SGRF is
\begin{equation}
\mathbf{E}\llbracket|\nabla{\mathscr{J}(x)}|^{2}\rrbracket=\lambda
\end{equation}
If $\bm{\mathcal{D}}=\bm{\mathcal{B}}_{R}(0)$ is a ball of radius R, then an estimate for the bound on the volatility of the DEF is then
\begin{align}
 (\|\psi(x)\|_{\mathscr{L}_{2}})^{2}=\mathbf{V}[\overline{\psi(x)}]
=\mathbf{E}\lbrace \mathcal{E}[\overline{\psi(x)}]\rbrace \le
\frac{1}{4}|\mathcal{E}(|\psi(x))|^{2})+\frac{1}{2}\lambda\pi R^{3}f(|\psi(x))
+\frac{4}{9}\lambda\pi^{2}R^{6}<\infty
\end{align}
so that the volatility grows with the radius of the ball but is always finite.
\end{lem}
\begin{proof}
\begin{align}
\mathbf{V}[f(\overline{\psi(x)})]&
=\mathbf{E}\bigg\llbracket\bigg|
\frac{1}{2}\int_{\bm{\mathcal{D}}}|\nabla_{i}\psi(x)|^{2}d^{3}x\bigg|\bigg\rrbracket\nonumber\\&
=\frac{1}{4}\mathbf{E}\bigg\llbracket\bigg(\int_{\bm{\mathcal{D}}}|
\nabla_{i}\overline{\psi(x)}|^{2}d^{3}x\bigg)\bigg(\int_{\bm{\mathcal{D}}}|
\nabla_{i}\overline{\psi(x)}|^{2}d^{3}x\bigg\rbrace\nonumber\\&
\le \mathbf{E}\bigg\llbracket\iint_{\bm{\mathcal{D}}}|
\nabla_{i}\overline{\psi(x)}|^{2}|\nabla_{i}\overline{\psi(x)}|^{2}d^{3}xd^{3}x\bigg\rrbracket
\nonumber\\&=\mathbf{E}\bigg\llbracket\iint_{\bm{\mathcal{D}}}|
\nabla_{i}\psi(x)+\nabla {\mathscr{J}(x)}|^{2}
|\nabla_{i}\psi(x)+\nabla{\mathscr{J}(x)}|^{2}d^{3}xd^{3}x\bigg\rrbracket
\nonumber\\&
\equiv\iint_{\bm{\mathcal{D}}}\mathbf{E}\bigg\llbracket|
\nabla_{i}\psi(x)+\nabla{\mathscr{J}(x)}|^{2}
|\nabla_{i}\psi(x)+\nabla{\mathscr{J}(x)}|^{2}\bigg\rrbracket d^{3}xd^{3}x\nonumber\\&=\frac{1}{4}\mathlarger{\mathlarger{\iint}}_{\bm{\mathcal{D}}}|\nabla_{i}\psi(x)|^{4}d^{3}x d^{3}x+ \iint_{\bm{\mathcal{D}}}|\nabla\psi_{i}(x)|^{3}\mathbf{E}\llbracket
{\mathscr{J}(x)}\rrbracket d^{3}d d^{3}x^{\prime}\nonumber\\&+
\frac{3}{2}\mathlarger{\mathlarger{\iint}}_{\bm{\mathcal{D}}}|\nabla_{i}\psi(x)|^{2}
\mathbf{E}\llbracket|{\mathscr{J}(x)|^{2}}d^{3}x d^{3}x+ \iint_{\bm{\mathcal{D}}}|\nabla\psi_{i}(x)|\mathbf{E}\llbracket
|\mathscr{J}(x)|^{3}\rrbracket d^{3}d d^{3}x\nonumber\\&
+\frac{1}{4}\iint_{\bm{\mathcal{D}}}\mathbf{E}\llbracket\nabla_{i}
\mathscr{J}(x)\nabla^{i}{\mathscr{J}(x)}\nabla_{i}
{\mathscr{J}(x)}\nabla^{i}{\mathscr{J}(x)(x)}
\rbrace d^{3}x d^{3}x\nonumber\\&=\frac{1}{4}\iint_{\bm{\mathcal{D}}}|\nabla_{i}\psi(x)|^{4}d^{3}x d^{3}x+
\frac{3}{2}\lambda\iint_{\bm{\mathcal{D}}}|\nabla_{i}\psi(x)|^{2}d^{x}d^{3}x
+\frac{1}{4}\lambda\iint_{\bm{\mathcal{D}}}d^{3}x d^{3}x\nonumber\\&
=\frac{1}{4}\iint_{\bm{\mathcal{D}}}|\nabla_{i}\psi(x)|^{4}d^{3}x d^{3}x+
\frac{3}{2}\lambda\|\bm{\mathcal{D}}\|\int_{\bm{\mathcal{D}}}|\nabla_{i}\psi(x)|^{2}d^{x}
+\frac{1}{4}\|\bm{\mathcal{D}}\|^{2}\nonumber\\&=\frac{1}{4}|
\mathcal{E}[\psi(x)]|^{2}+\frac{3}{2}\lambda\|\bm{\mathcal{D}}\|
\mathcal{E}[\psi(x)]+\frac{1}{4}\|\bm{\mathcal{D}}\|^{2}<\infty
\end{align}
so that the volatility of the DEI is finite and bounded. For a ball of radius $R$ in $\bm{\mathrm{R}}^{3}$ one has $\|\bm{\mathcal{D}}\|=\|\bm{\mathcal{B}}_{R}(0)\|
=\tfrac{4}{3}\pi R^{3}$.
\begin{align}
\mathsf{V}[\mathcal{E}(\overline{\psi(x)})]&\le \frac{1}{4}|\mathcal{E}[\psi(x)]|^{2}+\frac{3}{2}\lambda\|\bm{\mathcal{B}}_{R}(0)
\|\mathcal{E}[\psi(x)]+\frac{1}{4}\|\bm{\mathcal{B}}_{R}(0)\|^{2}\nonumber\\&
=\frac{1}{4}|\mathcal{E}[\psi(x)]|^{2}+\frac{1}{2}\lambda\pi R^{3}\mathcal{E}[\psi(x)]+\frac{4}{9}\lambda\pi^{2}R^{6}<\infty
\end{align}
\end{proof}
\subsection{Stochastically averaged Cacciopolli estimates}
\begin{lem}
Let $\bm{\mathcal{B}}_{R}(0)\subset\bm{\mathcal{B}}_{2R}(0)\subset\bm{\mathrm{R}}^{n}$ with compliment $\bm{\mathcal{B}}_{2R}(0)\symbol{92}\bm{\mathcal{B}}_{R}(0)$ and volumes $\|\bm{\mathcal{B}}_{2R}(0)\|,\|\bm{\mathcal{B}}_{R}(0)\|$,
$\|\bm{\mathcal{B}}_{2R}\symbol{92}\bm{\mathcal{B}}_{R}(0)\|$, and $\Delta\psi(x)=0$ so that $\psi(x)$ is a harmonic function. The Cacciopolli estimate is
\begin{equation}
\int_{\bm{\mathcal{B}}_{R}(x)}|\nabla_{i}\psi(x)d^{n}x\le \frac{4}{R^{2}}\int_{\bm{\mathcal{B}}_{2R}\symbol{92}\bm{\mathcal{B}}_{R}(x)}|\psi(x)|^{2}d^{n}x
\end{equation}
If ${\mathscr{J}(x,\omega)}\equiv{\mathscr{J}(x)}$ is a GRSF existing for all $x\in\bm{\mathcal{B}}_{2R}(0)$ with respect to a probability triple $(\bm{\Omega},\mathlarger{\mathscr{F}},\mathlarger{\mathsf{P}})$ then a random or 'noisy' ball/domain is $\bm{\mathcal{B}}_{2R}(0)
=\llbracket\bm{\mathcal{B}}_{2R}(0),\mathscr{J}(x), (\Omega,\mathlarger{\mathscr{F}},\bm{\mathsf{P}})\rrbracket$. The regulated covariances are ${\mathbf{E}}\llbracket \mathscr{J}(x)\otimes\ {\mathscr{J}(x)}\rrbracket=\alpha$ and ${\bm{\mathbf{E}}}(\nabla_{i}{\mathscr{J}(x)}\otimes\nabla_{j}
{\mathscr{J}(x)})
=\delta_{ij}\beta$ and summing, the factor of n can be absorbed into $\beta$.
The harmonic function is perturbed as $\overline{\psi(x)}=\psi(x)+\lambda{\mathscr{J}(x)}$ for $\lambda>0$ and the GRSF $\overline{\psi(x)}$ is stochastically harmonic in that $\mathbf{E}\llbracket\Delta\overline{\psi(x)}\rrbracket=0$. The stochastically averaged Cacciopolli estimate is then formally
\begin{equation}
\int_{\Omega}\int_{\bm{\mathcal{B}}_{R}(x)}|\nabla_{i}\overline{\psi(x,\omega)}d^{n}x d\omega\le \frac{4}{R^{2}}\int_{\Omega}\int_{\bm{\mathcal{B}}_{2R}\symbol{92}
\bm{\mathcal{B}}_{R}(x)}|\overline{\psi(x,\omega)}|^{2}d^{n}x d\omega
\end{equation}
or
\begin{equation}
\mathbf{E}\bigg\llbracket\int_{\bm{\mathcal{B}}_{R}(x)}|\nabla_{i}\overline{\psi(x)}d^{n}x
\bigg\rrbracket\le \frac{4}{R^{2}}\mathbf{E}\bigg\llbracket \int_{\bm{\mathcal{B}}_{2R}\symbol{92}\bm{\mathcal{B}}_{R}(x)}|\overline{\psi(x)}|^{2}d^{n}x
d\omega\bigg\rrbracket
\end{equation}
which holds if:
\begin{enumerate}
\item the inequality
\begin{equation}
\frac{R^{2}}{4}\left(\frac{\beta}{\alpha}\right)\le \frac{\|\bm{\mathcal{B}}_{2R}(0)\symbol{92}\bm{\mathcal{B}}_{R}(0)\|}{\|\bm{\mathcal{B}}_{R}(0)\|}
\end{equation}
holds.
\item Or equivalently if $\exists,~C_{1},C_{2}$ such that
\begin{align}
&C_{1}\sum_{Q=0}^{2}\binom{2}{Q}\left|\int_{\bm{\mathcal{B}}_{R}(0)}\nabla_{i}\psi(x)d^{n}x\right|^{2-\xi}[\frac{1}{2}\
\big\|\bm{\mathcal{B}}_{2R}\big\|(
\beta^{Q/2}+(-1)^{Q}\beta^{Q/2})]\nonumber\\&\le
\frac{4C_{2}}{R^{2}}\sum_{Q=0}^{2}\binom{2}{\xi}
\left|\int_{\bm{\mathcal{B}}_{2R}(0)\symbol{92}\bm{\mathrm{R}}(0)}
\nabla_{i}\psi(x)d^{n}x\right|^{2/3}[\frac{1}{2}\big\|\bm{\mathcal{B}}_{2R}\symbol{92}
\bm{\mathcal{B}}_{R}(0)\big\|(\alpha^{\xi/2}+(-1)^{Q}\alpha^{\xi/2})]
\end{align}
\end{enumerate}
\end{lem}
\begin{proof}
If (4.65) holds then the lhs is
\begin{align}
&\mathbf{E}\left\llbracket\int_{\bm{\mathcal{B}}_{R}(x)}|
\nabla_{i}\overline{\psi(x)}d^{n}x
\right\rrbracket\equiv\int_{\bm{\mathcal{B}}_{R}(x)}|\mathbf{E}
\bigg\llbracket\nabla_{i}\overline{\psi(x)}|^{2}\bigg\rrbracket d^{n}x\nonumber\\&
=\int_{\bm{\mathcal{B}}_{R}(x)}\big|\nabla_{i}\psi(x)\big|^{2}d^{n}x+\lambda^{2}
\int_{\bm{\mathcal{B}}_{R}(x)}\mathbf{E}\bigg\llbracket\bigg|\nabla_{i}
{\mathscr{J}(x)}\bigg|^{2}\bigg\rrbracket d^{n}x\nonumber\\&=\int_{\bm{\mathcal{B}}_{R}(x)}\big|\nabla_{i}\psi(x)\big|^{2}d^{n}x+\lambda^{2}
\beta\int_{\bm{\mathcal{B}}_{R}(x)}d^{n}x\\& =\int_{\bm{\mathcal{B}}_{R}(x)}\big|\nabla_{i}\psi(x)\big|^{2}d^{n}x+\lambda^{2}\beta\|\bm{\mathcal{B}}_{R}(x)\|
\nonumber\\&\le\frac{4}{R^{2}}\mathbf{E}\left\llbracket
\int_{\bm{\mathcal{B}}_{2R}(0)\symbol{92}\bm{\mathcal{B}}_{R}(0)}|\overline{\psi(x)}d^{n}x
\right\rrbracket\equiv\frac{4}{R^{2}}\int_{\bm{\mathcal{B}}_{2R}(0)\symbol{92}\bm{\mathcal{B}}_{R}(0)}|
\mathbf{E}
\bigg\llbracket\overline{\psi(x)}|^{2}\bigg\rrbracket d^{n}x\nonumber\\&
=\frac{4}{R^{2}}\int_{\bm{\mathcal{B}}_{2R}(0)\symbol{92}\bm{\mathcal{B}}_{R}(0)}\big|\psi(x)\big|^{2}d^{n}x+\frac{4}{R^{2}}\lambda^{2}
\int_{\bm{\mathcal{B}}_{2R}(0)\symbol{92}\bm{\mathcal{B}}_{R}(0)}\mathbf{E}\bigg\llbracket\bigg|
{\mathscr{J}(x)}\bigg|^{2}\bigg\rrbracket d^{n}x\nonumber\\&=\int_{\bm{\mathcal{B}}_{R}(x)}\big|\psi(x)\big|^{2}d^{n}x+\lambda^{2}
\alpha\int_{\bm{\mathcal{B}}_{R}(x)}d^{n}x\nonumber\\&=\frac{4}{R^{2}}\int_{\bm{\mathcal{B}}_{R}(x)}\big|\psi(x)\big|^{2}d^{n}x
 +\lambda^{2}\alpha\|\bm{\mathcal{B}}_{2R}(0)\symbol{92}\bm{\mathcal{B}}_{R}(0)
\|\end{align}
so that one must have
\begin{align}
&\int_{\bm{\mathcal{B}}_{R}(x)}\big|\nabla_{i}\psi(x)\big|^{2}d^{n}x+\lambda^{2}\beta\|\bm{\mathcal{B}}_{R}(x)\|
\nonumber\\&\le \frac{4}{R^{2}}\int_{\bm{\mathcal{B}}_{2R}\symbol{92}\bm{\mathcal{B}}_{R}(0)
(x)}\big|\psi(x)\big|^{2}d^{n}x+\frac{4}{R^{2}}\lambda^{2}\alpha\|\bm{\mathcal{B}}_{2R}(x)
\symbol{92}\bm{\mathcal{B}}_{R}(0)
\end{align}
which is
\begin{align}
\beta\|\bm{\mathcal{B}}_{R}(x)\| \le \frac{4}{R^{2}}\alpha\|\bm{\mathcal{B}}_{2R}(x)\symbol{92}\bm{\mathcal{B}}_{R}(0)
\end{align}
giving (4.67). To prove (4.68),the binomial theorem is applied so that
\begin{align}
&\mathbf{E}\left\llbracket\int_{\bm{\mathcal{B}}_{R}(0)}|\nabla_{i}\psi(x)+
\lambda{\mathscr{J}(x)}|^{2}d^{n}x\right\rrbracket
=\mathbf{E}\left\llbracket\sum_{Q=2}^{\infty}\binom{2}{Q}\left|\int_{\bm{\mathrm{R}}(0)}
\nabla_{i}\psi(x)d^{n}x\right|^{2-Q}\left|\int\nabla_{i}{\mathscr{J}(x)}\right|^{Q}d^{n}x
\right\rrbracket\nonumber\\& =\sum_{Q=2}^{\infty}\binom{2}{Q}\left|\int_{\bm{\mathrm{R}}(0)}
\nabla_{i}\psi(x)d^{n}x\right|^{2-Q}\mathbf{E}
\left\llbracket\left|\int\nabla_{i}{\mathscr{J}(x)}d^{n}x\right|^{Q}\right\rrbracket\nonumber\\&
=\sum_{Q=0}^{2}\binom{2}{Q}\left|\int_{\bm{\mathcal{B}}_{R}(0)}
\nabla_{i}\psi(x)d^{n}x\right|^{2-Q}\mathbf{E}
\left\llbracket\int\nabla_{i}{\mathscr{J}(x)}d^{n}x\underbrace{\times...\times}_{Q~times}
\int\nabla_{i}{\mathscr{J}(x)}d^{n}x\right\rrbracket\nonumber\\&
\le C_{1}\sum_{Q=0}^{2}\binom{2}{Q}\left|\int_{\bm{\mathcal{B}}_{R}(0)}
\nabla_{i}\psi(x)d^{n}x\right|^{2-Q}\int_{\bm{\mathcal{B}}_{R}(0)}...\int_{\bm{\mathcal{B}}_{R}(0)}
\mathbf{E}\bigg\llbracket|\nabla_{i}{\mathscr{J}(x)}|^{Q}\bigg\rrbracket
d^{n}x...d^{n}x\nonumber\\&=C_{1}\sum_{Q=0}^{2}\binom{2}{Q}\left|\int_{\bm{\mathcal{B}}_{R}(0)}
\nabla_{i}\psi(x)d^{n}x\right|^{2-Q}\int_{\bm{\mathcal{B}}_{R}(0)}
\mathbf{E}\bigg\llbracket|\nabla_{i}{\mathscr{J}(x)}|^{Q}\bigg\rrbracket d^{n}x\nonumber\\&=C_{1}\sum_{Q=0}^{2}\binom{2}{Q}\left|\int_{\bm{\mathcal{B}}_{R}(0)}
\nabla_{i}\psi(x)d^{n}x\right|^{2-Q}\int_{\bm{\mathcal{B}}_{R}(0)}
\big[\frac{1}{2}(\alpha^{Q/2}+(-1)^{Q/2}\alpha^{\alpha}\big]
d^{n}x\nonumber\\&=C_{1}\sum_{Q=0}^{2}\binom{2}{Q}\big|\int_{\bm{\mathcal{B}}_{R}(0)}
\nabla_{i}\psi(x)d^{n}x\big|^{2-Q}\left|\big[\frac{1}{2}(\alpha^{Q/2}+(-1)^{Q}
\alpha^{Q/2}\big]\right|^{Q}\|\bm{\mathcal{B}}_{R}(0)\|\nonumber\\&\le\frac{4}{R^{2}}
\mathbf{E}\left\llbracket\sum_{Q=2}^{\infty}\binom{2}{Q}\left|\int_{\bm{\mathrm{R}}(0)}
\psi(x)d^{n}x\right|^{2-Q}\left|\int{\mathscr{J}(x)}d^{n}x\right|^{Q}
\right\rrbracket\nonumber\\& =\frac{4}{R^{2}}\sum_{Q=2}^{\infty}\binom{2}{Q}\left|\int_{\bm{\mathrm{R}}(0)}
\psi(x)d^{n}x\right|^{2-Q}\mathbf{E}\left\llbracket\left|
\int{\mathscr{J}(x)}d^{n}x\right|^{Q}\right\rrbracket\nonumber
\\&=\frac{4}{R^{2}}\sum_{Q=0}^{2}\binom{2}{Q}\left|\int_{\bm{\mathsf{B}}_{R}(0)}
\psi(x)d^{n}x\right|^{2-Q}\mathbf{E}
\left\llbracket\int{\mathscr{J}(x)}d^{n}x\underbrace{\times...\times}_{Q~times}
\int \mathscr{J}(x)d^{n}x\right\rrbracket\nonumber\\&
\le\frac{C_{2}}{R^{2}}\sum_{Q=0}^{2}\binom{2}{Q}\left|\int_{\bm{\mathcal{B}}_{2R}(0)\symbol{92}\bm{\mathcal{B}}_{R}(0)}\psi(x)d^{n}x\right|^{2-Q}
\int_{\bm{\mathcal{B}}_{2R}(0)\symbol{92}\bm{\mathcal{B}}_{R}(0)}...\int_{\bm{\mathcal{B}}_{2R}(0)\symbol{92}\bm{\mathcal{B}}_{R}(0)}
\mathbf{E}
\bigg\llbracket|{\mathscr{J}(x)}|^{Q}\bigg\rrbracket
d\mu(x)...d\mu(x)x\nonumber\\&=\frac{4C_{2}}{R^{2}}\sum_{Q=0}^{2}\binom{2}{Q}\left|
\int_{\bm{\mathcal{B}}_{2R}(0)\symbol{92}\bm{\mathcal{B}}_{R}(0)}\psi(x)d\mu(x)\right|^{2-Q}
\int_{\bm{\mathcal{B}}_{2R}(0)\symbol{92}\bm{\mathcal{B}}_{R}(0)}\mathbf{E}
\big\llbracket|{\mathscr{J}(x)}|^{Q}\bigg\rrbracket d\mu(x)x\nonumber\\&=\frac{4C_{2}}{R^{2}}\sum_{Q=0}^{2}\binom{2}{Q}\left|\int_{\bm{\mathcal{B}}_{2R}(0)\symbol{92}\bm{\mathcal{B}}_{R}(0)}
\nabla_{i}\psi(x)d^{n}x\right|^{2-Q}\int_{\bm{\mathcal{B}}_{2R}(0)\symbol{92}\bm{\mathcal{B}}_{R}(0)}\big[\frac{1}{2}(\alpha^{Q/2}+(-1)^{Q}\alpha^{Q/2}\big]
d^{n}x\nonumber\\&=\frac{4C_{2}}{R^{2}}\sum_{Q=0}^{2}\binom{2}{Q}\big|
\int_{\bm{\mathcal{B}}_{2R}(0)\symbol{92}\bm{\mathcal{B}}_{R}(0)}\psi(x)d\mu(x)\big|^{2-Q}\left|\big[\frac{1}{2}(\alpha^{Q}+(-1)^{Q}\alpha^{Q}\big]
\right|^{Q}\|\bm{\mathcal{B}}_{2R}(0)\symbol{92}\bm{\mathcal{B}}_{R}(0)\|
\end{align}
giving (4.66)
\end{proof}
\subsection{Turbulence in potential flow theory}
The effects of 'turbulence' in potential flow can be considered by introducing the GRSF $\mathscr{J}(x)$ and its derivative $\nabla_{i}\mathscr{J}(x)$.
\begin{prop}
Let $u_{i}(x)=\nabla_{i}\psi(x)$ be a smooth potential flow in a domain $\bm{\mathcal{D}}$. Then $\psi(x)$ is harmonic in that $\Delta\psi(x)=0$ and $\nabla^{i}u_{i}(x)=0$ since the fluid is incompressible. If $\mathscr{J}(x)$ is a GRSF existing for all $x\in\bm{\mathcal{D}}$ then the randomly perturbed potential function is $\overline{\psi(x)}=\psi(x)+{\mathscr{J}(x)}$. The corresponding turbulent flow is
\begin{equation}
\overline{u_{i}(x)}=\nabla_{i}\psi(x)+\nabla_{i}{\mathscr{J}(x)}\equiv u_{i}(x)+\nabla_{i}{\mathscr{J}(x)}
\end{equation}
Then:
\begin{enumerate}
\item The stochastically averaged Laplace equation is
\begin{equation}
\mathbf{E}\big\llbracket\Delta\overline{\psi(x)}\big\rrbracket
=\Delta\psi(x)+\mathbf{E}\big\llbracket\Delta{\mathscr{J}(x)}\big\rrbracket=
\Delta\psi(x)=0
\end{equation}
\item The 2-point covariance for any pair $(x,y)\in\bm{\mathcal{D}}$ is
\begin{equation}
\mathbf{E}\big\llbracket\overline{u_{i}(x)}\otimes\overline{u_{j}(y)}
\big\rrbracket=u_{i}(x)u_{j}(y)+\lambda\delta_{ij}Q(x,y;\xi)
\end{equation}
and $Q(x,y;\xi)\rightarrow 0$ for $|x-y|\gg \xi$ so that the velocity correlations decay rapidly on all scales greater than the correlation length.
\item The volatility of the turbulent flow at any point $x\in\bm{\mathcal{D}}$ is
\begin{equation}
\mathbf{E}\big\llbracket|
\overline{u_{i}(x)}|^{2}\big\rrbracket=|u(x)|^{2}+\delta_{ij}\lambda<\infty
\end{equation}
which is always bounded.
\end{enumerate}
\end{prop}
\begin{proof}
Equation (4.75) follows from $\mathbf{E}\big\llbracket\Delta \mathscr{J}(x)\rrbracket\equiv\mathbf{E}\big\llbracket\nabla_{i}\nabla^{i}
\mathscr{J}\big\rrbracket=0$.
The volatility is
\begin{align}
&\mathbf{E}\big\llbracket|\overline{u_{i}(x)}|^{2}
\big\rrbracket=|u(x)|^{2}+u_{i}(x)\mathbf{E}
\big\llbracket\nabla_{i}\mathscr{J}(x)\big\rrbracket\nonumber\\&
+u_{j}(x)\mathbf{E}\big\llbracket\nabla_{i}\mathscr{J}(x)\big\rrbracket
+\mathbf{E}\big\llbracket\nabla_{i}\mathscr{J}(x)
\otimes \nabla_{j}\mathscr{J}(x)\big\rrbracket\nonumber\\&
=|u_{i}(x)|^{2}+\lambda\delta_{ij}<\infty
\end{align}
The 2-point function for any points $(x,y)$ is
\begin{align}
&\mathbf{E}\big\llbracket|{u_{i}(x)}\otimes\overline{u_{j}(y)}
\big\rrbracket=|u_{i}(x)|u_{j}(y)|+u_{i}(x)\mathbf{E}\big\llbracket\nabla_{j}
{\mathscr{J}(y)}\big\rrbracket\nonumber\\&+u_{j}(y)\mathbf{E}
\big\llbracket\nabla_{i}{\mathscr{J}(x)}\big\rrbracket
+\mathbf{E}\big\llbracket{\nabla_{i}\mathscr{J}(x)}\otimes \nabla_{j}{\mathscr{J}(y)}\big\rrbracket\nonumber\\&
=|u_{i}(x)||u_{j}(y)|+\lambda Q(x,y;\xi)\delta_{ij}<\infty
\end{align}
\end{proof}
\begin{lem}
Let $\overline{\psi(x)}=\psi(x)+{\mathscr{J}(x)}$, with $\nabla_{i}{\mathscr{J}(x)}\otimes\nabla_{j}{\mathscr{J}(x)}\rrbracket
=\delta_{ij}\lambda $. The stochastically averaged Kelvin energy (which is essentially the Dirichlet energy)is minimised for potential flows but the minimum is shifted by a constant factor such that
\begin{equation}
\mathbf{E}\bigg\llbracket~Kel[\psi]  \bigg\rrbracket=\frac{1}{2}\rho\mathbf{E}\bigg\llbracket\int_{\bm{\mathcal{D}}}|
\overline{u_{i}(x)}|^{2}d^{3}x\bigg\rrbracket=\frac{1}{2}\rho\int_{\bm{\mathcal{D}}}|u_{i}(x)|^{2}d^{3}x+\lambda\|\bm{\mathcal{D}}\|
\end{equation}
where $\|\bm{\mathcal{D}}\|$ is the volume of the domain of integration $\bm{\mathcal{D}}$
\end{lem}
\begin{proof}
\begin{align}
\mathbf{E}\bigg\llbracket~Kel[\psi]\bigg\rrbracket &=\frac{1}{2}
\rho\mathbf{E}\bigg\llbracket\int_{\bm{\mathcal{D}}}\bigg(|u(x)|^{2}+u_{i}(x)
\mathscr{J}^{i}(x)+u^{i}(x){\mathscr{J}_{i}(x)}+\nabla_{i}
{\mathscr{J}(x)}\otimes\nabla_{j}
\mathscr{J}(x)\bigg)d^{3}x\bigg\rrbracket\nonumber\\&=\frac{1}{2}
\rho\int_{\bm{\mathcal{D}}}\bigg(|u(x)|^{2}
+\mathbf{E}\bigg\llbracket\nabla_{i}{\mathscr{J}(x)}\otimes\nabla_{j}
{\mathscr{J}(x)}
\bigg\rrbracket\bigg)d^{3}x\nonumber\\&=\frac{1}{2}
\rho\int_{\bm{\mathcal{D}}}|u(x)|^{2}d^{3}x
+\lambda\rho\int_{\bm{\mathcal{D}}}d^{3}x\nonumber\\&=\frac{1}{2}
\rho\int_{\bm{\mathcal{D}}}|u(x)|^{2}d^{3}x
+\frac{1}{2}\lambda\rho\|\bm{\mathcal{D}}\|\equiv Kel[x]+C
\end{align}
\end{proof}
\subsection{Turbulent flow in a cylinder}
Let $\bm{\mathcal{C}}_{R}\subset\bm{\mathrm{R}}^{3}$ be a finite cylinder of radius $R$ and (essentially) infinite, and let $u(r,z)$ be an axisymmetric laminar fluid flow velocity, which is the gradient of a potential function $\psi(r,z)$. The flow is steady state, irrotational and incompressible. Then $\Delta(r,z)\psi(r,z)=0$ and $\psi(R,z)=g(z)$ at the boundary. In full
\begin{equation}
\Delta_{(r,z)}\psi(r,z)=\frac{1}{r}\partial_{r}(\tfrac{1}{2}\partial_{r}(\psi(r,z))+\tfrac{1}{r^{2}}
\partial_{\varphi}^{2}\psi(r,z)+\partial^{2}_{z}\psi(r,z)=0
\end{equation}
If a SGRF ${\mathscr{J}(r,z)}$ exists for all $(r,z)\in\bm{C}_{R}$ then the
'turbulent' potential functions at $(r,z)$ and $(r^{\prime},z^{\prime})$ are then
\begin{align}
&\overline{\psi(r,z)}=\psi(r,z)+{\mathscr{J}(r,z)}\\&
\overline{\psi(r^{\prime},z^{\prime})}=\psi(r^{\prime},z^{\prime})
+{\mathscr{J}(r^{\prime},z^{\prime})}
\end{align}                                                                                  The GRSF within $\bm{C}_{R,L}$ is ${\mathscr{J}(r,z,\varphi)}$. With rotational symmetry about the cylinder axis, the random field is ${\mathscr{J}(r,z)}$. The expectation and 2-point functions are
\begin{align}
&\mathbf{E}\big\llbracket{\mathscr{J}(r,z)}\big\rrbracket=0\\&
\mathbf{E}\big\llbracket{\mathscr{J}(r,z)}\otimes
{\mathscr{J}(r^{\prime},z^{\prime})}\big\rrbracket=\lambda K(r,r^{\prime};\epsilon)Z(z,z^{\prime};\xi)
\end{align}
and the 2-point function is regulated so that the volatility of the field at any point in the cylinder is finite and bounded
\begin{equation}
\lim_{r\rightarrow r^{\prime},z\rightarrow z^{\prime}}\mathbf{E}\big\llbracket \mathscr{J}(r,z)(r,z)\otimes
\mathscr{J}(r,z)(r^{\prime},z^{\prime})\big\rrbracket=\lambda
\end{equation}
The gradient is $\nabla_{i}=(\partial_{r},\partial_{z})$ so that the possible turbulent
velocity components at $(r,z,r^{\prime},z^{\prime})\in\bm{C}_{R,L}$ are
\begin{align}
&\overline{u_{r}(r,z)}=\partial_{r}\psi(r,z)+\partial_{r}{\mathscr{J}(r,z)}\\&
\overline{u_{r}(r,z)}=\partial_{z}\psi(r,z)+\partial_{z}{\mathscr{J}(r,z)}\\&
\overline{u_{r}(r,z)}=\partial_{r}\psi(r,z)+\partial_{r^{\prime}}{\mathscr{J}(r,z)}\\&
\overline{u_{r}(r,z)}=\partial_{z}\psi(r,z)+\partial_{z^{\prime}}{\mathscr{J}(r,z)}
\end{align}
with $\mathbf{E}\llbracket\overline{u_{r}(r,z)}\rbrace=u_{r}(r,z)$ and so on. The possible velocity correlations at any two points within the cylinder are then (Appendix C)
\begin{align}
&\mathbf{E}\llbracket\overline{u_{r}(r,z)}\otimes\overline{u_{r^{\prime}}(r^{\prime},z^{\prime})}\rbrace
=\mathbf{E}\llbracket\nabla_{r}\overline{\psi(r,z)}\otimes \nabla_{r^{\prime}}
\overline{\psi(r^{\prime},z^{\prime})}\rbrace\nonumber\\&= u_{r}(r,z) u_{r^{\prime}}(r^{\prime},z^{\prime})+\lambda \partial_{r}\partial_{r^{\prime}} K(r,r;\epsilon)Z(z,z^{\prime})\\&\mathbf{E}\llbracket\overline{u_{z}(r,z)}\otimes\overline{u_{z^{\prime}}(r^{\prime},z^{\prime})}\rrbracket
=\mathbf{E}\llbracket\nabla_{z}\overline{\psi(r,z)}\otimes \nabla_{z^{\prime}}\overline{\psi(r^{\prime},z^{\prime})}\rrbracket\nonumber\\&
=u_{r}(r,z) u_{r^{\prime}}(r^{\prime},z^{\prime})+\lambda K(r,r;\epsilon)\partial_{z}\partial_{z^{\prime}}Z(z,z^{\prime}\\&
\mathbf{E}\llbracket\overline{u_{r}(r,z)}\otimes\overline{u_{z^{\prime}}(r^{\prime},z^{\prime})}\rrbracket
=\mathbf{E}\llbracket\nabla_{r}\overline{\psi(r,z)}\otimes \nabla_{z^{\prime}}\overline{\psi(r^{\prime},z^{\prime})}\rrbracket\nonumber\\&=u_{r}(r,z) u_{r^{\prime}}(r^{\prime},z^{\prime})+\lambda\partial_{r}K(r,r;\epsilon)\partial_{z^{\prime}}Z(z,z^{\prime}\\&
\mathbf{E}\llbracket\overline{u_{z}(r,z)}\otimes\overline{u_{z^{\prime}}(r^{\prime},z^{\prime})}\rrbracket
=\mathbf{E}\llbracket\nabla_{z}\overline{\psi(r,z)}\otimes \nabla_{z^{\prime}}\overline{\psi(r^{\prime},z^{\prime})}\rrbracket\nonumber\\&=u_{r}(r,z) u_{r^{\prime}}(r^{\prime},z^{\prime})+\lambda\partial_{r^{\prime}}K(r,r;\epsilon)
\partial_{z}Z(z,z^{\prime}
\end{align}
The volatilities at any point are finite and bounded since
\begin{align}
&\lim_{r\rightarrow r^{\prime},z\rightarrow z^{\prime}}
\mathbf{E}\llbracket\overline{u_{r}(r,z)}\otimes
\overline{u_{r^{\prime}}(r^{\prime},z^{\prime})}\rrbracket
=\lim_{r\rightarrow^{\prime},z\rightarrow z^{\prime}}\mathbf{E}
\llbracket \nabla_{r}\overline{\psi(r,z)}\otimes \nabla_{r^{\prime}}\overline{\psi(r^{\prime},z^{\prime})}\rrbracket=\lambda
\\& \lim_{r\rightarrow r^{\prime},z\rightarrow z^{\prime}}
\mathbf{E}\llbracket\overline{u_{z}(r,z)}\otimes
\overline{u_{z^{\prime}}(r^{\prime},z^{\prime})}\rrbracket=
\lim_{r\rightarrow r^{\prime},z\rightarrow z^{\prime}}\mathbf{E}
\llbracket\nabla_{z}\overline{\psi(r,z)}\otimes \nabla_{z^{\prime}}\overline{\psi(r^{\prime},z^{\prime})}\rrbracket=\lambda
\\&\lim_{r\rightarrow^{\prime},z\rightarrow z^{\prime}}\mathbf{E}\lbrace\overline{u_{r}(r,z)}\otimes
\overline{u_{z^{\prime}}(r^{\prime},z^{\prime})}\rrbracket
\lim_{r\rightarrow r^{\prime},z\rightarrow z^{\prime}}\mathbf{E}
\llbracket\nabla_{r}\overline{\psi(r,z)}\otimes \nabla_{z^{\prime}}\overline{\psi(r^{\prime},z^{\prime})}\rrbracket=\lambda\\&
\lim_{r\rightarrow r^{\prime},z\rightarrow z^{\prime}}\mathbf{E}\llbracket\overline{u_{z}(r,z)}\otimes
\overline{u_{z^{\prime}}(r^{\prime},z^{\prime})}\rrbracket\lim_{r\rightarrow r^{\prime},z
\rightarrow z^{\prime}}\mathbf{E}\llbracket\nabla_{z}\overline{\psi(r,z)}\otimes \nabla_{r^{\prime}}\overline{\psi(r^{\prime},z^{\prime})}\rrbracket=\lambda
\end{align}
At the cylinder boundary $r=r^{\prime}=R$.
\begin{align}
\mathbf{E}\llbracket\nabla_{z}\overline{\psi(R,z)}\otimes \nabla_{z^{\prime}}\overline{\psi(R,z^{\prime})}\rrbracket=\lambda K(R,R;\epsilon)\partial_{z}\partial_{z^{\prime}}Z(z,z^{\prime})
=\partial_{z}\partial_{z^{\prime}}Z(z,z^{\prime})
\end{align}
All  velocity correlations decay rapidly for $|r-r^{\prime}|\gg\epsilon$ and $|z-z^{\prime}|\gg \xi$
\subsection{Stochastic Riesz potentials: expectations and moments}
The Riesz potential was defined in subsection 2.5 and the RP for GRSFs was considered in Appendix A. Here, we consider the Riesz potential for random fields of the form $\overline{g(x)}=g(x)+\lambda\mathscr{J}$.
\begin{prop}
Let ${\mathscr{J}(x)}$ be a GRSF and $g:\bm{\mathcal{D}}\rightarrow\mathbf{R}$ a smooth function. The randomly perturbed function is $\overline{g(x)}=g(x)+\lambda{\mathscr{J}(x)}$. This gives a randomly perturbed Riesz potential of the form
\begin{align}
&\mathlarger{\bm{\mathfrak{F}}}_{a}\overline{g(x)}=\gamma(a)\int_{\bm{\mathcal{D}}}\frac{\overline{g(y)}d\mu(y)}{|x-y|^{n-a}}
=\gamma(a)\int_{\bm{\mathcal{D}}}\frac{g(y)d\mu(y)}{|x-y|^{n-a}}
+\lambda\gamma(a)\int_{\bm{\mathcal{D}}}
\frac{\overline{\mathscr{J}(y)}d\mu(y)}{|x-y|^{n-a}}\nonumber\\&
=\gamma(a)\int_{\bm{\mathcal{D}}}g(y)\mathcal{O}(x,y)d\mu(y)
+\lambda\gamma(a)\int_{\bm{\mathcal{D}}}\mathcal{O}(x,y)\mathscr{J}(y)d\mu(y)
\end{align}
Then if $\mathbf{E}\llbracket{\mathscr{J}(x)}\rrbracket=0$ and $\mathbf{E}\llbracket{\mathscr{J}(x)}
\otimes{\mathscr{J}}(x^{\prime})
=\alpha K(x,x^{\prime})$
\begin{enumerate}
\item The expectation or mean value is
\begin{equation}
\mathbf{E}\llbracket \mathlarger{\bm{\mathfrak{F}}}_{a}\overline{g(x)}\rrbracket
=\mathlarger{\bm{\mathfrak{F}}}_{a}g(x)
\end{equation}
\item The 2-point covariance between points $(x,x^{\prime})\in\bm{\mathcal{D}}$ is
\begin{align}
&\mathbf{E}\llbracket \mathlarger{\bm{\mathfrak{F}}}_{a}\overline{g(x)}\otimes \mathlarger{\bm{\mathfrak{F}}}_{a}\overline{g(x^{\prime})}\rrbracket
=|\mathlarger{\bm{\mathfrak{F}}}_{a}\overline{g(x)}\mathlarger{\bm{\mathfrak{F}}}_{a}\overline{g(x^{\prime)}}|+
|\gamma(a)|^{2}\iint_{\bm{\mathcal{D}}}\frac{K(y,y^{\prime};\epsilon)d\mu(y)d\mu(y^{\prime})}{(|x^{\prime}-y^{\prime}|^{n-\alpha}|x-y|^{n-\alpha}}
\nonumber\\&
=|\mathlarger{\bm{\mathfrak{F}}}_{a}\overline{g(x)}\mathlarger{\bm{\mathfrak{F}}}_{a}\overline{g(x^{\prime)}}|
+\alpha|\gamma(a)|^{2}\iint_{\bm{\mathcal{D}}}K(y,y^{\prime};\epsilon)
\mathcal{O}(x,y)\mathcal{O}(x^{\prime},y^{\prime})d\mu(y)d\mu(y^{\prime})\nonumber\\&
=|\mathlarger{\bm{\mathfrak{F}}}_{a}\overline{g(x)}\mathlarger{\bm{\mathfrak{F}}}_{a}\overline{g(x^{\prime)}}|+\mathcal{H}(x,x^{\prime})
\end{align}
\item The $p^{th}$-order moments are finite for $x\ne y$ such that
\begin{align}
&\mathbf{E}\llbracket
|\mathlarger{\bm{\mathfrak{F}}}_{a}\overline{g(x)}|^{p}\rrbracket\le C\sum_{Q=0}^{P}\binom{P}{Q}\gamma(a)
\left|\int_{\bm{\mathcal{D}}}\frac{g(y)d\mu(y)}{|x-y|^{n-a}}\right|^{P-Q}|\gamma(a)|^{Q}
\nonumber\\&\times\bigg[\frac{1}{2}\bigg(\alpha^{Q/2}+(-1)^{Q}\alpha^{Q/2}\bigg)\bigg]\left|
\int_{\bm{\mathcal{D}}}\frac{d\mu(y)}{|x-y|^{n-a}}\right|^{Q}<\infty,~x\ne y
\end{align}
\end{enumerate}
\end{prop}
\begin{proof}
To prove(1)
\begin{align}
&\mathbf{E}\bigg\llbracket \mathlarger{\bm{\mathfrak{F}}}_{a}\overline{g(x)}\bigg\rrbracket =\mathbf{E}\left\llbracket\gamma(a)\int_{\bm{\mathcal{D}}}\frac{\overline{g(y)}d\mu(y)}{|x-y|^{n-a}}
\right\rrbracket\nonumber\\&=\gamma(a)\int_{\bm{\mathcal{D}}}\frac{g(y)d\mu(y)}{|x-y|^{n-a}}+\lambda\gamma(a)
\mathbf{E}\left\llbracket\int_{\bm{\mathcal{D}}}
\frac{\mathscr{J}(y)d\mu(y)}{|x-y|^{n-a}}\right\rrbracket\nonumber \\&
=\gamma(a)\int_{\bm{\mathcal{D}}}\frac{g(y)d\mu(y)}{|x-y|^{n-a}}+\lambda\gamma(a)\int_{\bm{\mathcal{D}}}
\frac{\mathbf{E}\bigg\llbracket{\mathscr{J}(y)}\bigg\rrbracket d\mu(y)}{|x-y|^{n-a}}=\gamma(a)\int_{\bm{\mathcal{D}}}\frac{g(y)d\mu(y)}{|x-y|^{n-a}}
\end{align}
To prove (2), the stochastic RPs at points $(x,x^{\prime})$ are
\begin{align}
&\mathlarger{\bm{\mathfrak{F}}}_{a}\overline{g(x)}=\gamma(a)\int_{\bm{\mathcal{D}}}\frac{\overline{g(y)}d\mu(y)}{|x-y|^{n-a}}=\gamma(a)\int_{\bm{\mathcal{D}}}\frac{g(y)d\mu(y)}{|x-y|^{n-a}}
+\lambda\gamma(a)\int_{\bm{\mathcal{D}}}\frac{{\mathscr{J}(y)}d\mu(y)}{|x-y|^{n-a}}\\&
\mathlarger{\bm{\mathfrak{F}}}_{a}\overline{g(x^{\prime})}=\gamma(a)\int_{\bm{\mathcal{D}}}
\frac{\overline{g(y)}d\mu(y)}{|x-y|^{n-a}}
=\gamma(a)\int_{\bm{\mathcal{D}}}\frac{g(y^{\prime})d\mu(y^{\prime})}{|x^{\prime}-y^{\prime}|^{n-a}}
+\lambda\gamma(a)\int_{\bm{\mathcal{D}}}\frac{{\mathscr{J}(y^{\prime})}d\mu(y)}{|x^{\prime}-y^{\prime}|^{n-a}}
\end{align}
Then
\begin{align}
&\mathbf{E}\bigg\llbracket \mathlarger{\bm{\mathfrak{F}}}_{a}\overline{g(x)}\mathlarger{\bm{\mathfrak{F}}}_{a}\overline{g(x^{
\prime})}\bigg\rrbracket=\mathbf{E}\bigg\llbracket\bigg(\gamma(a)\int_{\bm{\mathcal{D}}}\frac{g(y)d\mu(y)}{|x-y|^{n-a}}
+\lambda\gamma(a)\int_{\bm{\mathcal{D}}}\frac{{h(y)}d\mu(y)}{|x-y|^{n-a}}\bigg)\nonumber\\&\times\bigg
(\gamma(a)\int_{\bm{\mathcal{D}}}\frac{g(y^{\prime})d\mu(y^{\prime})}{|x^{\prime}-y^{\prime}|^{n-a}}
+\lambda\gamma(a)\int_{\bm{\mathcal{D}}}
\frac{\mathscr{J}(y^{\prime})d\mu(y)}{|x^{\prime}-y^{\prime}|^{n-a}}\bigg)
\bigg\rrbracket\nonumber\\&=|\gamma(a)|^{2}\iint_{\bm{\mathcal{D}}
}\frac{g(y)g(y^{\prime}d\mu(y)d\mu(y^{\prime})}{(|x-y|^{n-a}|x^{\prime}-y^{\prime}|^{n-a}}
+\lambda|\gamma(a)|^{2}\mathbf{E}\bigg\llbracket\iint_{\bm{\mathcal{D}}
}\frac{g(y)\overline{\mathscr{J}(y^{\prime}}d\mu(y)d\mu(y^{\prime})}{(|x-y|^{n-a}|x^{\prime}-y^{\prime}|^{n-a}}\bigg\rrbracket
\nonumber\\&+\lambda|\gamma(a)|^{2}\mathbf{E}\bigg\llbracket\iint_{\bm{\mathcal{D}}
}\frac{\overline{\mathscr{J}(y)}g(y^{\prime})d\mu(y)d\mu(y^{\prime})}{(|x-y|^{n-a}|x^{\prime}-y^{\prime}|^{n-a}}\bigg\rrbracket
+\lambda^{2}|\gamma(a)|^{2}\mathbf{E}\bigg\llbracket\int_{\bm{\mathcal{D}}}
\frac{\mathscr{J}(y)\otimes \mathscr{J}(y^{\prime})d\mu(y)d\mu(y^{\prime})}{(|x-y|^{n-a}|x^{\prime}-y^{\prime}|^{n-a}}
\bigg\rrbracket\nonumber\\&=|\gamma(a)|^{2}\iint_{\bm{\mathcal{D}}
}\frac{g(y)g(y^{\prime}d\mu(y)d\mu(y^{\prime})}{(|x-y|^{n-a}|x^{\prime}-y^{\prime}|^{n-a}}
+\lambda|\gamma(a)|^{2}\bigg\llbracket\iint_{\bm{\mathcal{D}}
}\frac{g(y)\mathbf{E}\bigg\llbracket {\mathscr{J}(y^{\prime}}\bigg\rrbracket d\mu(y)d\mu(y^{\prime})}{(|x-y|^{n-a}|x^{\prime}-y^{\prime}|^{n-a}}\bigg\rrbracket
\nonumber\\&+\lambda|\gamma(a)|^{2}\mathlarger{\mathlarger{\iint_{\bm{\mathcal{D}}}}
}\frac{\mathbf{E}\bigg\llbracket{\mathscr{J}(y)}\bigg\rrbracket g(y^{\prime})d\mu(y)d\mu(y^{\prime})}{(|x-y|^{n-a}|x^{\prime}-y^{\prime}|^{n-a}}\bigg\rrbracket
+\lambda^{2}|\gamma(a)|^{2}\iint_{\bm{\mathcal{D}}}\frac{\mathbf{E}\bigg
\llbracket{\mathscr{J}(y)}\otimes \mathscr{J}(y^{\prime})\bigg\rrbracket d\mu(y)d\mu(y^{\prime})}{(|x-y|^{n-a}|x^{\prime}-y^{\prime}|^{n-a}}\bigg
\rrbracket\nonumber\\&
=|\gamma(a)|^{2}\iint_{\bm{\mathcal{D}}}\frac{g(y)g(y^{\prime}d\mu(y)d\mu(y^{\prime})}{(|x-y|^{n-a}|x^{\prime}-y^{\prime}|^{n-a}}
+\lambda^{2}|\gamma(a)|^{2}\iint_{\bm{\mathcal{D}}}\frac{\mathbf{E}\bigg\llbracket{\mathscr{J}(y)}\otimes \mathscr{J}(y^{\prime})\bigg\rrbracket d\mu(y)d\mu(y^{\prime})}{(|x-y|^{n-a}|x^{\prime}-y^{\prime}|^{n-a}}\bigg\rrbracket
\nonumber\\&
=|\gamma(a)|^{2}\iint_{\bm{\mathcal{D}}}\frac{g(y)g(y^{\prime}d\mu(y)d\mu(y^{\prime})}{(|x-y|^{n-a}|x^{\prime}-y^{\prime}|^{n-a}}
+\lambda^{2}|\gamma(a)|^{2}\iint_{\bm{\mathcal{D}}
}\frac{\alpha\mathrm{K}(y,y^{\prime}) d\mu(y)d\mu(y^{\prime})}{(|x-y|^{n-a}|x^{\prime}-y^{\prime}|^{n-a}}
\nonumber\\& =|\mathlarger{\bm{\mathfrak{F}}}_{a}g(x)||\mathlarger{\bm{\mathfrak{F}}}_{a}g(x^{\prime})|+\lambda^{2}|\gamma(a)|^{2}\iint_{\bm{\mathcal{D}}
}\frac{\alpha K(y,y^{\prime}) d\mu(y)d\mu(y^{\prime})}{(|x-y|^{n-a}|x^{\prime}-y^{\prime}|^{n-a}}\nonumber\\&
=|\mathlarger{\bm{\mathfrak{F}}}_{a}g(x)||\mathlarger{\bm{\mathfrak{F}}}_{a}g(x^{\prime})|
+\mathrm{K}(x,x^{\prime})
\end{align}
Finally, the $p^{th}$-order moments or covariance is computed as
\begin{align}
&\mathbf{E}\bigg\llbracket|\mathlarger{\bm{\mathfrak{F}}}_{a}\overline{g(x)}|^{P}\bigg
\rrbracket=\mathbf{E}
\bigg\llbracket\bigg|\gamma(a)\int_{\bm{\mathcal{D}}}\frac{g(y)d\mu(y)}{(|x-y|^{n-a}}+\lambda\gamma(a)
\int_{\bm{\mathcal{D}}}\bigg|\frac{{\mathscr{J}(y)}d\mu(y)}{(|x-y|^{n-a}}\bigg|^{P}\bigg\rrbracket
\nonumber\\&
=\mathbf{E}\bigg\llbracket\sum_{Q=0}^{P}\binom{P}{Q}\bigg|\gamma(a)
\int_{\bm{\mathcal{D}}}\frac{g(y)d\mu(y)}
{(|x-y|^{n-a}}\bigg|^{P-Q}\bigg|\llbracket\gamma(a)\int_{\bm{\mathcal{D}}}
\bigg|\frac{{\mathscr{J}(y)}d\mu(y)}{(|x-y|^{n-a}}\bigg|^{Q}
\bigg\rbrace\nonumber\\&=\sum_{Q=0}^{P}\binom{P}{Q}\bigg|\gamma(a)\int_{\bm{\mathcal{D}}}\frac{g(y)d\mu(y)}{(|x-y|^{n-a}}
\bigg|^{P-Q}\mathbf{E}
\bigg\llbracket\bigg|\lambda\gamma(a)\int_{\bm{\mathcal{D}}}\bigg|
\frac{{\mathscr{J}(y)}d\mu(y)}{(|x-y|^{n-a}}\bigg|^{Q}
\bigg\rrbracket\nonumber\\&=
C\sum_{Q=0}^{P}\binom{P}{Q}\bigg|\gamma(a)\int_{\bm{\mathcal{D}}}\frac{g(y)d\mu(y)}{(|x-y|^{n-a}}
\bigg|^{p-Q}\bigg||\lambda\gamma(a)|^{Q}
\int_{\bm{\mathcal{D}}}...\int_{\bm{\mathcal{D}}}\bigg|\frac{{\mathbf{E}}
\bigg\llbracket|\mathscr{J}|^{Q}\bigg\rrbracket
d\mu(y)...d\mu(y)}{(|x-y|^{n-a}\times...|x-y|^{n-a}}\bigg|\nonumber\\&=
C\sum_{Q=0}^{P}\binom{P}{Q}\bigg|\gamma(a)\int_{\bm{\mathcal{D}}}\frac{g(y)d\mu(y)}{(|x-y|^{n-a}}
\bigg|^{P-Q}\bigg||\lambda\gamma(a)|^{Q}\nonumber\\&\times
\int_{\bm{\mathcal{D}}}...\int_{\bm{\mathcal{D}}}\bigg[\frac{1}{2}\bigg(\alpha^{Q/2}+(-1)^{Q}
\alpha^{Q/2}\bigg)\bigg]
\frac{d\mu(y)...d\mu(y)}{(|x-y|^{n-a}\times...|x-y|^{n-a}}\bigg|\nonumber\\&=
C\sum_{Q=0}^{P}\binom{P}{Q}\bigg|\gamma(a)\int_{\bm{\mathcal{D}}}\frac{g(y)d\mu(y)}{(|x-y|^{n-a}}
\bigg|^{P-Q}\bigg||\lambda\gamma(a)|^{Q}\nonumber\\&\times
\bigg[\frac{1}{2}\bigg(\alpha^{Q/2}+(-1)^{Q}\alpha^{Q/2}\bigg)\bigg]
\bigg|\int_{\bm{\mathcal{D}}}\frac{d\mu(y)}{(|x-y|^{n-a}}\bigg|^{Q}
\end{align}
\end{proof}
\begin{lem}
The integral over a ball can be evaluated and is essentially the Newtonian or Coulomb potential of a ball of constant mass or charge density of Theorem (1.7). The volume integral over $\bm{\mathrm{R}}^{3}$ is $\bm{\mathcal{D}}=\bm{\mathcal{B}}_{R}(0)$.
\begin{equation}
\int_{\bm{\mathcal{B}}_{R}(0)}\frac{d\mu(y)}{|x-y|^{n-a}}\equiv\int_{\bm{\mathcal{B}}_{R}(0)}\frac{d^{3}y}{(|x-y|^{n-a}}
\end{equation}
This will agree with Theorem (1.7) when $C=\rho=1$ and $n-\alpha=|3-a|=1$ so $a=2$. Letting $x=(0,0,a)$ and $\|x\|=a$, the integral is
\begin{align}
&\int_{\bm{\mathcal{B}}_{R}(0)}\frac{d^{3}y}{|x-y|}=\frac{4}{3}\pi\frac{R^{3}}{a}
=\frac{4}{3}\pi\frac{R^{3}}{\|x\|},~~a>R\\& \int_{\bm{\mathcal{B}}_{R}(0)}\frac{d^{3}y}{|x-y|}=2\pi(R^{2}-\tfrac{1}{3}a^{2}),~~0\le a\le R
\end{align}
\end{lem}
\begin{proof}
If $x=((0,0,a)$ and $\|x\|=a$ then $\hat{a}=z\hat{a}$ is a vector along the z-axis. The integral over a ball $\bm{\mathcal{B}}_{R}(0)$ is
\begin{align}
\int_{\bm{\mathcal{B}}_{R}(0)}\frac{d^{3}y}{|x-y|}&=\iiint_{R>r}\frac{d^{3}r}{|r-a|}=
2\pi\int_{0}^{R}\bigg|\frac{\sin\theta d\theta}{\sqrt{r^{2}-2ra\cos\theta+a^{2}}}
\bigg|r^{2}dr\nonumber\\&=2\pi\int_{0}^{R}\bigg|\int_{-1}^{1}
\frac{d\xi}{\sqrt{r^{2}-2ra\cos\theta+a^{2}}}
\bigg| r^{2} dr\nonumber\\&
=2\pi\int_{0}^{R}\bigg|\bigg[\frac{\sqrt{r^{2}-2ra\xi+a^{2}}}{-ra}\bigg]_{\xi=-1}^{\xi=1}
\bigg|r^{2} dr\nonumber\\&=\frac{2\pi}{a}\int_{0}^{R}\bigg[\bigg|\sqrt{r^{2}+2ra+a^{2}}-\sqrt{r^{2}-2ra+a^{2}}
\bigg|\bigg]r^{2}dr\nonumber\\&=\frac{2\pi}{a}\int_{0}^{R}\bigg[\bigg|\sqrt{(r+a)(r+a)}
-\sqrt{(r-a)(r-a)}\bigg|\bigg]r^{2}dr\nonumber\\&=\frac{2\pi}{a}\int_{0}^{R}\bigg[\bigg|{|r+a|}
-{|r-a|}\bigg|\bigg]r^{2}dr
\end{align}
Evaluating the integral from $a\le R$ or $a>R$ then gives (4.111) or (4.112)                as required.
\end{proof}
\subsection{Newtonian potential for ball subject to random density fluctuations}
These results are now used to make estimates of the covariances, high-order moments and volatility of a randomly perturbed Newtonian potential, which is essentially the Riesz potential for $n-\alpha=1$
\begin{thm}
Let $\bm{\mathcal{B}}_{R}(0)\subset\bm{\mathrm{R}}^{3}$ and let $\rho(x)=\rho$ be a constant density of mass or charge contained within the ball $\bm{\mathcal{B}_{R}}(0)$. The Newtonian or Coulomb potential outside the ball is then given by (-) so that
\begin{equation}
\psi(x)=\frac{C\rho}{4\pi}\int_{\bm{\mathcal{B}}_{R}(0)}\frac{d\mu(y)}{|x-y|}
=\frac{C\rho}{4\pi}\int_{\bm{\mathcal{B}}_{R}(0)}\frac{d^{3}y}{|x-y|}
\equiv\gamma\int_{\bm{\mathcal{B}}_{R}(0)}\frac{d\mu(y)}{|x-y|^{3-\alpha}}
\end{equation}
Since the mass or charge behaves as if concentrated at a point, $\psi(x)$ is harmonic for $0\symbol{92}R^{3}$. Let $\mathscr{J}(x)$ be GRSF which also exists through out $\bm{\mathcal{B}}_{R}(0)$ but which vanishes for all $x\in\bm{\mathcal{B}}_{R}(0)\symbol{92}\bm{\mathrm{R}}^{3}$. The GRSF has the usual properties $\mathbf{E}\llbracket{\mathscr{J}}(x)\rrbracket=0$ and $\mathbf{E}\llbracket\mathscr{J}(x)\otimes
{\mathscr{J}(x)}\rrbracket=0$. If the density becomes 'noisy' or is randomly perturbed as
\begin{equation}
\overline{\rho(x)}=\rho+\lambda{\mathscr{J}(x)}
\end{equation}
so that $\overline{\rho(x)}$ also a GRSF withe same statistical properties.
\begin{align}
&\mathbf{E}\llbracket\overline{\rho(x)}\rrbracket=\rho\\&
\mathbf{E}\llbracket\overline{\rho(x)}\otimes \overline{\rho(x)}\rrbracket=\rho^{2}+\lambda^{2}\alpha
\end{align}
This scenario can represent a 'star' of radius $R$ (contained within $\bm{\mathcal{B}}_{R}(0)$) with a uniform density $\rho$ that randomly fluctuates with mean $\rho$. The exterior Newtonian potential of the star is randomly perturbed as
\begin{equation}
\overline{\psi(x)}=\frac{C}{4\pi}\int_{\bm{\mathcal{B}}_{R}(0)}\frac{\overline{\rho(x)}d^{3}y}{|x-y|}
=\frac{C}{4\pi}\int_{\bm{\mathcal{B}}_{R}(0)}\frac{\rho d^{3}y}{|x-y|}+\frac{C\lambda}{4\pi}\int_{\bm{\mathcal{B}}_{R}(0)}
\frac{{\mathscr{J}(x)}d^{3}y}{|x-y|}
\end{equation}
Then:
\begin{enumerate}
\item The averaged potential outside the ball is
\begin{equation}
{\mathbf{E}}\bigg\llbracket\overline{\psi(x)}\bigg\rrbracket={\mathbf{E}}\bigg\llbracket\frac{C}{4\pi}
\int_{\bm{\mathcal{B}}_{R}(0)}
\frac{\overline{\rho(x)}d^{3}y}{|x-y|}\bigg\rrbracket
=\frac{C}{4\pi}\int_{\bm{\mathcal{B}}_{R}(0)}\frac{\rho d^{3}y}{|x-y|}
\end{equation}
which is the static Newtonian potential.
\item If $x=(0,0,a)$ with $\|x\|=a$ the covariance or $P^{th}$-order moments are finite and bounded and decay to zero for $\|x\|\rightarrow \infty$.
    \begin{align}
     &\mathsf{M}_{P}(x)=\mathbf{E}\bigg\llbracket|\overline{\psi(a)}|^{P}\bigg\rrbracket   \le C\sum_{Q=0}^{\infty}\binom{P}{Q}
     \bigg|\frac{4\pi}{Q}\bigg|^{P}\bigg|\frac{C\rho}{4\pi}\bigg|^{P-Q}
     \bigg[\frac{1}{2}\bigg(\alpha^{Q/2}+(-1)^{Q}\alpha^{Q/2}\bigg)\bigg]
     \frac{R^{3}}{a^{3}}\nonumber\\&
     \equiv C\sum_{Q=0}^{\infty}\binom{P}{Q}
     \bigg|\frac{4\pi}{Q}\bigg|^{P}\bigg|\frac{C\rho}{4\pi}\bigg|^{P-Q}
     \bigg[\frac{1}{2}\bigg(\alpha^{Q/2}+(-1)^{Q}\alpha^{Q/2}\bigg)\bigg]
     \frac{R^{3}}{\|x\|^{3}}
    \end{align}
\item The volatility is for $p=2$ so that
\begin{align}
     &\bm{\mathsf{V}}(x)=\mathbf{E}\bigg\llbracket|\overline{\psi(a)}|^{2}\bigg\rrbracket=
     \sum_{Q=0}^{\infty}\binom{2}{Q}\bigg|\frac{4\pi}{Q}\bigg|^{2}\bigg|\frac{C\rho}{4\pi}\bigg|^{2-Q}
     \bigg[\frac{1}{2}\bigg(\alpha^{Q/2}+(-1)^{Q}\alpha^{Q/2}\bigg)\bigg]
     \frac{R^{3}}{a^{3}}\nonumber\\&
     \equiv\sum_{Q=0}^{\infty}\binom{2}{Q}
     \bigg|\frac{4\pi}{Q}\bigg|^{2}\bigg|\frac{C\rho}{4\pi}\bigg|^{2-Q}
     \bigg[\frac{1}{2}\bigg(\alpha^{Q/2}+(-1)^{Q}\alpha^{Q/2}\bigg)\bigg]
     \frac{R^{3}}{\|x\|^{3}}
    \end{align}
\end{enumerate}
All moments of any order $p$ decay at the same rate of $\sim R^{3}/\|x^{3}\|$.
\end{thm}
\begin{proof}
The proof follows from the previous theorem.
\end{proof}
\begin{lem}
Let $\psi(x)$ be Newtonian potential due a constant density $\rho$ of mass or charge contained within a ball $\bm{\mathcal{B}}_{R}(0)\subset\bm{\mathrm{R}}^{n}$. Let $(x,x^{\prime})\in\bm{\mathcal{B}}_{R}(0)$ and $(y,y^{\prime})\in\bm{\mathcal{B}}_{R}(0)\symbol{92}\bm{\mathrm{R}}^{n}$. The PE is $\Delta\psi(x)=C\rho$. Let ${\mathscr{J}(x)}$ be a regulated GRSF with
$\mathbf{E}\llbracket \mathscr{J}(x)
\otimes\mathscr{J}(x^{\prime})\rrbracket=\alpha$ and
\begin{equation}
\mathbf{E}\llbracket|\mathscr{J}(x)|^{p}\rrbracket =
\frac{1}{2}(\alpha^{p/2}+(-1)^{p}\alpha^{p/2}]
\end{equation}
so that odd moments vanish. The noisy or randomly fluctuating density within the ball is
\begin{equation}
\overline{\rho(x)}=\rho+\lambda\mathscr{J}(x)
\end{equation}
and $\mathbf{E}\llbracket \overline{\rho(x)}\rrbracket=0$, for some $\xi>0$. Then the bounded moments of the Laplacian of the randomly perturbed Newtonian potential are estimated as
\begin{enumerate}
\item
\begin{equation}
\mathbf{E}\bigg\llbracket\bigg|\Delta_{x}\overline{\psi(x)}\bigg|^{P}\bigg\rrbracket=\sum_{Q=0}^{P}
\binom{P}{Q}\Omega(P,Q)\rho^{P-Q}[\tfrac{1}{2}(\alpha^{Q/2}+(-1)^{Q}
\alpha^{Q/2}]<\infty
\end{equation}
\item The volatility for $p=2$ is then
\begin{equation}
\mathbf{E}\bigg\llbracket\bigg|\Delta_{x}\overline{\psi(x)}\bigg|^{2}
\bigg\rrbracket=\sum_{Q=0}^{p}\binom{2}{Q}\Omega(2,Q)\rho^{2-Q}[\tfrac{1}{2}
(\alpha^{Q/2}+(-1)^{Q}\alpha^{Q/2}]= \left|{C\rho}\right|^{2}
+\left|\frac{\xi C}{4\pi}\right|^{2}\Lambda
\end{equation}
and for $p=1$ the LE and harmonic properties are recovered on average so that $\mathbf{E}\llbracket\Delta\overline{\psi(x)}\rrbracket=0$.
\end{enumerate}
\end{lem}
\begin{proof}
The Laplacians of the fields $\overline{\psi(x)}$ and $\overline{\psi(x^{\prime})}$ are
\begin{align}
&\Delta_{x}\overline{\psi(x)}=\frac{C\rho}{4\pi}\nabla_{x}
\left(\int_{\bm{\mathcal{B}})_{R}(0)}\frac{d^{3}y}{|x-y|}\right)
+\frac{\lambda C}{4\pi}\nabla_{x}\left(\int_{\bm{\mathcal{B}}_{R}(0)}
\frac{{\mathscr{J}(y)}d^{3}y}{|x-y|}\right)\\&
\Delta_{x^{\prime}}\overline{\psi(x^{\prime})}=\frac{C\rho}{4\pi}\nabla_{x}
\left(\int_{\bm{\mathcal{B}})_{R}(0)}\frac{d^{3}y^{\prime}}{|x^{\prime}-y^{\prime}|}\right)
+\frac{\lambda C}{4\pi}\nabla_{x^{\prime}}\left(\int_{\bm{\mathcal{B}}_{R}(0)}
\frac{{\mathscr{J}(y^{\prime})}d^{3}y^{\prime}}{|x^{\prime}-y^{\prime}|}\right)
\end{align}
\begin{align}
\mathbf{E}\bigg\llbracket|\Delta_{x}\overline{\psi(x)}|^{p}\bigg\rrbracket
&=\mathbf{E}\bigg\llbracket\bigg|\frac{C\rho}{4\pi}\nabla_{x}
\bigg(\int_{\bm{\mathcal{B}})_{R}(0)}\frac{d^{3}y}{|x-y|}\bigg)
+\frac{\lambda C}{4\pi}\nabla_{x}\bigg(\int_{\bm{\mathcal{B}}_{R}(0)}
\frac{\mathscr{J}(y)d^{3}y}{|x-y|}\bigg)\bigg|^{p}\bigg\rrbracket
\nonumber\\&=
\mathbf{E}\bigg\llbracket\bigg|\frac{C\rho}{4\pi}
\bigg(\underbrace{\int_{\bm{\mathcal{B}}_{R}(0)}\nabla_{x}\bigg(\frac{1}{|x-y|}\bigg)d^{3}y}_{4\pi}
\bigg)+\frac{\lambda C}{4\pi}\bigg(\int_{\bm{\mathcal{B}}_{R}(0)}\nabla_{x}
\bigg(\frac{1}{|x-y|}\bigg){\mathscr{J}(y)}d^{3}y\bigg)
\bigg|^{P}\bigg\rrbracket\nonumber\\&
=\mathbf{E}\bigg\llbracket\bigg|C\rho+\frac{\lambda C}{4\pi}\bigg(\int_{\bm{\mathcal{B}}_{R}(0)}\nabla_{x}\bigg(\frac{1}{|x-y|}\bigg)
{\mathscr{J}(x)(y)}d^{3}y\bigg)\bigg|^{P}\bigg\rrbracket\nonumber\\&
=\sum_{Q=0}^{P}\binom{P}{Q}|C\rho|^{P-Q}
\mathbf{E}\bigg\llbracket\bigg|\frac{\lambda C}{4\pi}\bigg(\int_{\bm{\mathcal{B}}_{R}(0)}\nabla_{x}
\bigg(\frac{1}{|x-y|}\bigg){\mathscr{J}(y)}d^{3}y\bigg)
\bigg|^{Q}\bigg\rrbracket\nonumber\\&=\sum_{Q=0}^{P}\binom{P}{Q}|C\rho|^{P-Q}
\mathbf{E}\bigg\llbracket\bigg|\frac{\lambda C}{4\pi}\bigg(\int_{\bm{\mathcal{B}}_{R}(0)}\delta^{n}(x-y){\mathscr{J}(y)}d^{3}y\bigg)
\bigg|^{Q}\bigg\rrbracket\nonumber\\&=\sum_{Q=0}^{p}\binom{P}{Q}|C\rho|^{P-Q}\bigg|\frac{\lambda C}{4\pi}\bigg|\mathbf{E}\bigg\llbracket|{\mathscr{J}(x)}|^{Q}\bigg
\rrbracket\nonumber\\&=\sum_{Q=0}^{P}\binom{P}{Q}|C\rho|^{P-Q}\bigg|\frac{\lambda C}{4\pi}\bigg|[\tfrac{1}{2}(\alpha^{Q/2}+(-1)^{Q}\alpha^{Q/2}]
\nonumber\\&\equiv\sum_{Q=0}^{p}\binom{P}{Q}|\Omega(P,Q)\rho^{P-Q}[\tfrac{1}{2}
(\alpha^{Q/2}+(-1)^{Q}\alpha^{Q/2}]
\end{align}
\end{proof}
\begin{thm}
Let $\bm{\mathcal{B}}_{R}(0)\subset\bm{\mathrm{R}}^{n}$ and let $\mathscr{J}(x)$ be a GRSF with the usual properties. As before, the ball contains matter of uniform density $\rho$ and the potential $\psi(x)$ outside the ball or 'star' is given by (1.34). Let $m$ and $m^{\prime}$ be point masses at $x$ and $x^{\prime}$ with $(x,x^{\prime})\in\bm{\mathcal{B}}_{R}(0)\symbol{92}\bm{\mathrm{R}}^{n}$ then the forces on the point masses due to the ball are
\begin{align}
&F_{i}(x)=-m\nabla_{i}\psi(x)=-m\nabla_{i}\bigg(\frac{C\rho}{4\pi}\int_{\bm{\mathcal{B}}_{R}(0)} \frac{d^{n}y}{|x-y|}\bigg)
\\&
F_{j}(x^{\prime})=-m^{\prime}\nabla_{j}\psi(x^{\prime})=
-m\nabla_{j}\bigg(\frac{C\rho}{4\pi}\int_{\bm{\mathcal{B}}_{R}(0)}\frac{d^{n}y^{\prime}}{|x^{\prime}-y^{\prime}|}\bigg)
\end{align}
If the density is randomly perturbed as
\begin{equation}
\overline{\rho(x)}=\rho+\lambda{\mathscr{J}(x)}
\end{equation}
then the forces on the masses become randomly perturbed as
\begin{align}
&\overline{F_{i}(x)}=F_{i}(x)+\overline{\mathcal{F}_{i}(x)}\\&=-m\nabla_{i}\psi(x)
=-m\nabla_{i}\left(\frac{C\rho}{4\pi}\int_{\bm{\mathcal{B}}_{R}(0)} \frac{d^{n}y}{|x-y|}\right)-m\nabla_{i}\left(\frac{C\lambda}{4\pi}\int_{\bm{\mathcal{B}}_{R}(0)} \frac{\mathscr{J}(y)d^{n}y}{|x-y|}\right)
\\&
\overline{F_{i}(x^{\prime})}=F_{i}(x^{\prime})+\overline{\mathcal{F}_{i}(x^{\prime})}\nonumber\\&=
-m^{\prime}\nabla_{j}\left(\frac{C\rho}{4\pi}\int_{\bm{\mathcal{B}}_{R}(0)}
\frac{d^{n}y^{\prime}}{|x^{\prime}-y^{\prime}|}\right)
-m^{\prime}\nabla_{i}\left(\frac{C\lambda}{4\pi}\int_{\bm{\mathcal{B}}_{R}(0)} \frac{{\mathscr{J}}(y^{\prime})d^{n}y^{\prime}}{|x^{\prime}-y^{\prime}|}\right)
\end{align}
It follows that:
\begin{enumerate}
\item If $x=(0,0,a)$ with $\|x\|=a$ with $a>R$ then the moments are
\begin{align}
&\mathbf{E}\bigg\llbracket|\overline{F_{i}(x)}|^{P}\bigg
\rrbracket\le \Lambda\sum_{Q=0}^{P}\binom{P}{Q}\Omega(m,P,C)\rho^{P}[\tfrac{1}{2}(
\alpha^{Q/2}+(-1)^{Q}\alpha^{Q/2})]
\bigg(-\frac{4}{3}R^{3}\frac{\hat{a}}{|a|^{3}}\bigg)^{P}\nonumber\\&
\mathbf{E}\bigg\llbracket|\overline{F_{i}(x^{\prime})}|^{P}\bigg\rrbracket=
\le \Lambda\sum_{Q=0}^{P}\binom{P}{Q}\Omega(m,P,C)\rho^{P}[\tfrac{1}{2}(
\alpha^{Q/2}+(-1)^{Q}\alpha^{Q/2})]
\bigg(-\frac{4}{3}R^{3}\frac{\hat{a}}{|a^{\prime}|^{3}}\bigg)^{P}
\end{align}
which vanish for odd $P$ with $\|\hat{a}\|=1$ and where $\Omega(m,P,C)=$. The moments are then always finite and bounded and decay to zero for $|a|\gg R$.
\item The 2-point covariance is
\begin{align}
\mathbf{E}\bigg\llbracket\overline{F_{i}(x)}\otimes\overline{F_{i}(x^{\prime})}\bigg\rrbracket
&=mm^{\prime}\bigg(\frac{C\rho}{4\pi}\bigg)^{2}\iint_{\bm{\mathcal{B}}_{R}(0)}\frac{(x-y)_{i}(x^{\prime}-y^{\prime})^{i}d^{3}y d^{3}y^{\prime}}{|x-y|^{i}|x^{\prime}-y^{\prime}|^{3}}\nonumber\\&
+mm^{\prime}\bigg(\frac{C\rho}{4\pi}\bigg)^{2}\iint_{\bm{\mathcal{B}}_{R}(0)}
\frac{K(x,x^{\prime};\epsilon)d^{3}y d^{3}y^{\prime}}{|x-y|^{3}|x^{\prime}-y^{\prime}|^{3}}
\end{align}
\item The volatility then follows from (4.135) or (4.136) as
\begin{align}
\mathbf{E}\llbracket|\overline{F_{i}(x)}|^{2}\rrbracket
&=mm^{\prime}\bigg(\frac{C\rho}{4\pi}\bigg)^{2}\int_{\bm{\mathcal{B}}_{R}(0)}\bigg(\frac{|(x-y)_{i}|d^{3}y}{|x-y|^{3}}\bigg)
\nonumber\\&+mm^{\prime}\bigg(\frac{C\rho}{4\pi}\bigg)^{2}\bigg(\int_{\bm{\mathcal{B}}_{R}(0)}
\frac{\alpha d^{3}y}{|x-y|^{3}}\bigg)^{2}
\end{align}
which can also be evaluated from (4.135) for $P=2$.
\end{enumerate}
\end{thm}
\begin{proof}
The moments are
\begin{align}
&\mathbf{E}\llbracket|\overline{F_{i}(x)}|^{P}\rrbracket=\mathbf{E}
\bigg\llbracket\bigg|-m\nabla_{i}\bigg(\frac{C\rho}{4\pi}\int_{\bm{\mathcal{B}}_{R}(0)} \frac{d^{n}y}{|x-y|}\bigg)-m\nabla_{i}\left(\frac{C\lambda}{4\pi}\int_{\bm{\mathcal{B}}_{R}(0)} \frac{{\mathscr{J}(y)}d^{n}y}{|x-y|}\right)\bigg|^{P}\bigg\rrbracket\nonumber\\&
=\mathbf{E}\bigg\llbracket\sum_{Q=0}^{P}\binom{P}{Q}\bigg(-\frac{mC\rho}{4\pi}
\int_{\bm{\mathcal{B}}_{R}(0)}\frac{(x-y)^{i}}{|x-y|^{3}}d^{3}y\bigg)^{P-Q}\bigg(-\frac{mC\lambda}{4\pi}\int_{\bm{\mathcal{B}}_{R}(0)}
\frac{\mathscr{J}(y)(x-y)^{i}d^{3}y}{|x-y|^{3}}\bigg)^{Q}\bigg\rrbracket\nonumber\\&
=\sum_{Q=0}^{P}\binom{P}{Q}\bigg(-\frac{mC\rho}{4\pi}\int_{\bm{\mathcal{B}}_{R}(0)}\frac{ (x-y)^{i}}{|x-y|^{3}}d^{3}y\bigg)^{P-Q}\mathbf{E}\bigg\llbracket
\bigg(-\frac{m C\xi}{4\pi}\int_{\bm{\mathcal{B}}_{R}(0)}
\frac{{\mathscr{J}}(y)(x-y)^{i}d^{3}y}{|x-y|^{3}}\bigg)^{Q}\bigg\rrbracket
\nonumber\\&\le \lambda\sum_{Q=0}^{P}\binom{P}{Q}\bigg(-\frac{mC\rho}{4\pi}\int_{\bm{\mathcal{B}}_{R}(0)}\frac{ (x-y)^{i}}{|x-y|^{3}d^{3}y}\bigg)^{P-Q}\nonumber\\&\bigg[\bigg(\frac{1}{2}
\bigg(\alpha^{Q/2}+(-1)^{Q}\alpha^{Q/2}\bigg)\bigg]\bigg(\frac{mC\lambda}{4\pi}\bigg)^{Q}\bigg|
\int_{\bm{\mathcal{B}}_{R}(0)}\frac{(x-y)^{i}d^{3}y}{|x-y|^{3}}
\bigg|^{Q}\nonumber\\&
\equiv\lambda \sum_{Q=0}^{P}\binom{P}{Q}\bigg(\int_{\bm{\mathcal{B}}_{R}(0)}\frac{ (x-y)^{i}d^{3}y}{|x-y|^{3}}\bigg)^{P-Q}\Omega(m,P,Q,C,\rho)\nonumber\\& \bigg[\bigg(\frac{1}{2}\bigg(\alpha^{Q/2}+(-1)^{Q}\alpha^{Q/2}
\bigg)\bigg] \bigg|\int_{\bm{\mathcal{B}}_{R}(0)}\frac{(x-y)^{i}d^{3}y}{|x-y|^{3}}
\bigg|^{Q}\nonumber\\&
\equiv\lambda \sum_{Q=0}^{P}\binom{P}{Q}\Omega(m,P,C,\rho)\bigg[\bigg(\frac{1}{2}
\bigg(\alpha^{Q/2}+(-1)^{Q}\alpha^{Q/2}
\bigg)\bigg] \bigg|\int_{\bm{\mathcal{B}}_{R}(0)}\frac{(x-y)^{i}d^{3}y}{|x-y|^{3}}
\bigg|^{p}
\end{align}
The integral outside the ball can be evaluated with $a > R$ since
\begin{equation}
\int_{\bm{\mathcal{B}}_{R}(0)}\frac{(x-y)^{i}d^{3}y}{|x-y|^{3}}=\nabla_{i} \int_{\bm{\mathcal{B}}_{R}(0)}\frac{d^{3}y}{|x-y|}
\end{equation}
and the integral (4.11) has already been evaluated in (4.111)
\begin{equation}
\nabla_{a}\iiint_{r<R} \frac{d^{3}R}{|R-a|}=-\frac{4\pi}{3} R^{3}\frac{\hat{a}}{|a|^{3}}
\end{equation}
for $a>R$. Hence (4.106) follows.

The 2-point covariance follows readily from (4.135) and (4.136) so that
\begin{align}
&\mathbf{E}\llbracket\overline{F(x)}\otimes\overline{F(x^{\prime})}\rrbracket
=mm^{\prime}\frac{C\rho}{4\pi}\iint_{\bm{\mathcal{B}}_{R}(0)}\frac{(x-y)^{i}(x-y)_{i}
d^{3}yd^{3}y^{\prime}}{|x-y|^{3}|x^{\prime}-y^{\prime}|^{3}}\nonumber\\&+mm^{\prime}\frac{C\rho}{4\pi}
\iint_{\bm{\mathcal{B}}_{R}(0)}\frac{{\mathbf{E}}\bigg
\llbracket{\mathscr{J}(y)}
\otimes{\mathscr{J}(y^{\prime})}\bigg\rrbracket
d^{3}yd^{3}y^{\prime}}{|x-y|^{3}|x^{\prime}-y^{\prime}|^{3}}\nonumber\\&
=mm^{\prime}\frac{C\rho}{4\pi}\iint_{\bm{\mathcal{B}}_{R}(0)}\frac{(x-y)^{i}(x-y)_{i}
d^{3}yd^{3}y^{\prime}}{|x-y|^{3}|x^{\prime}-y^{\prime}|^{3}}\nonumber\\&
+mm^{\prime}\frac{C\rho}{4\pi}\iint_{\bm{\mathcal{B}}_{R}(0)}\frac{\alpha(y,y^{\prime};\epsilon)
d^{3}yd^{3}y^{\prime}}{|x-y|^{3}|x^{\prime}-y^{\prime}|^{3}}
\end{align}
which gives the volatility when $x=x^{\prime}$ and $y=y^{\prime}$.
\end{proof}
\subsection{Stochastically averaged Harnack inequality}
The Harnack inequality is given in (-) and also follows from the Poisson integral. A harmonic function $\psi(x)$ in a ball of radius $R$ and centre zero satisfies the Harnack inequality
\begin{equation}
\frac{R^{n-1}(R-|x|)}{(R+|x|)^{n-1}}\psi(0)\le \psi(x)\le \frac{R^{n-1}(R+|x|)}{(R-|x|)^{n-1}}
\end{equation}
Equivalently, for all $p\ge 1$
\begin{equation}
\bigg|\frac{R^{n-1}(R-|x|)}{(R+|x|)^{n-1}}\psi(0)\bigg|^{p}\le |\psi(x)|^{p}\le \bigg|\frac{R^{n-1}(R+|x|)}{(R-|x|)^{n-1}}\psi(0)|^{p}
\end{equation}
Setting $\bm{\mathfrak{X}}(R,|x|,n)=\frac{R^{n-1}(R-|x|)}{R+|x|)^{n-1}}
$ and $\bm{\mathfrak{Y}}(R,|x|,n)=\frac{R^{n-1}(R+|x|)}{R-|x|)^{n-1}}$, this is
\begin{equation}
|\bm{\mathfrak{X}}(R,|x|,n)|^{p}|\psi(0)|^{p}\le |\psi(x)|^{p}\le |\bm{\mathfrak{Y}}(R,|x|,n)|^{p}|\psi(0)|^{p}
\end{equation}
\begin{lem}
If a noisy/random ball containing a regulated GRSF ${\mathscr{J}(x)}$ with the usual properties, then the harmonic function is randomly perturbed within the ball as
\begin{equation}
\overline{\psi(x)}=\psi(x)+\lambda{\mathscr{J}(x)}
\end{equation}
The averaged Harnack inequality for $p\ge 1$ is then
\begin{equation}
\mathbf{E}\bigg\llbracket |\bm{\mathfrak{X}}(R,|x|,n)|^{p}|\overline{\psi(0)}|^{p}\bigg\rrbracket\le \mathbf{E}\bigg\llbracket|\overline{\psi(x)}|^{p}\bigg\rrbracket\le \mathbf{E}\bigg\llbracket|\bm{\mathfrak{Y}}(R,|x|,n)|^{p}|\overline{\psi(0)}|^{p}
\bigg\rrbracket
\end{equation}
\end{lem}
\begin{proof}
\begin{align}
&\mathbf{E}\bigg\llbracket|\bm{\mathfrak{X}}(R,|x|,n)
\overline{\psi(0)}|^{p}\bigg\rrbracket
=\mathbf{E}\bigg\llbracket|\bm{\mathfrak{X}}(R,|x|,n)\psi(0)+\lambda \bm{\mathfrak{X}}(R,|x|,n){\mathscr{J}(x)}|^{p}\bigg\rrbracket\nonumber\\&
=\mathbf{E}\bigg\llbracket\sum_{Q=0}^{P}\binom{P}{Q}|\bm{\mathfrak{X}}(R,|x|,n)\psi(0)|
^{P-Q}|\lambda\bm{\mathfrak{X}}(R,|x|,n){\mathscr{J}(x)}|^{Q}\bigg\rrbracket\nonumber\\&
=\sum_{Q=0}^{P}\binom{P}{Q}|\bm{\mathfrak{X}}(R,|x|,n)\psi(0)|^{P-Q}|\lambda \bm{\mathfrak{X}}(R,|x|,n)|^{Q}\mathbf{E}
\bigg\llbracket{\mathscr{J}(x)}|^{Q}\bigg\rrbracket\nonumber\\&
=\sum_{Q=0}^{P}\binom{P}{Q}|\bm{\mathfrak{X}}(R,|x|,n)\psi(0)|^{P-Q}|\lambda \bm{\mathfrak{X}}(R,|x|,n)|^{Q}
\big[\tfrac{1}{2}(\alpha^{Q/2}+(-1)^{Q}\alpha^{Q/2})\big]
\end{align}
and
\begin{align}
&\mathbf{E}\bigg\llbracket |\bm{\mathfrak{X}}(R,|x|,n)\overline{\psi(0)}|^{p}\bigg\rrbracket
=\mathbf{E}\bigg\llbracket |\bm{\mathfrak{X}}(R,|x|,n)\psi(0)+\lambda \bm{\mathfrak{X}}(R,|x|,n){\mathscr{J}(x)}|^{P}\bigg\rrbracket\nonumber\\&
=\mathbf{E}\bigg\llbracket\sum_{Q=0}^{P}\binom{P}{Q}|\bm{\mathfrak{X}}(R,|x|,n)\psi(0)|^{P-Q}|\xi \bm{\mathfrak{X}}(R,|x|,n){\mathscr{J}(x)}|^{Q}\bigg\rrbracket\nonumber\\&
=\sum_{Q=0}^{P}\binom{P}{Q}|\bm{\mathfrak{X}}(R,|x|,n)\psi(0)|^{P-Q}|\lambda \bm{\mathfrak{X}}(R,|x|,n)|^{Q}\mathbf{E}\bigg
\llbracket{\mathscr{J}(x)}|^{Q}\bigg\rrbracket\nonumber\\&
=\sum_{Q=0}^{P}\binom{P}{Q}|\bm{\mathfrak{X}}(R,|x|,n)\psi(0)|^{P-Q}|\lambda \mathscr{J}(R,|x|,n)|^{Q}
\big[\tfrac{1}{2}(\alpha^{Q/2}+(-1)^{Q}\alpha^{Q/2})\big]
\end{align}
Then
\begin{equation}
\mathbf{E}\bigg\llbracket\overline{\psi(x)|}^{P}\bigg\rrbracket=
\sum_{Q=0}^{P}\binom{P}{Q}|\psi(x)|^{P-Q}
\big[\tfrac{1}{2}(\alpha^{Q/2}+(-1)^{Q}\alpha^{Q/2}]
\end{equation}
so that the inequality becomes
\begin{align}
&\sum_{Q=0}^{P}\binom{P}{Q}|\bm{\mathfrak{X}}(R,|x|,n)\psi(0)|^{P-Q}|\lambda\bm{\mathfrak{X}}(R,|x|,n)|^{Q}
\big[\tfrac{1}{2}(\alpha^{Q/2}+(-1)^{Q}\alpha^{Q/2})\big]\nonumber\\&
\le \sum_{Q=0}^{P}\binom{P}{Q}|\psi(x)|^{P-Q}
\big[\tfrac{1}{2}(\alpha^{Q/2}+(-1)^{Q}\alpha^{Q/2}]\nonumber\\&
\le\sum_{Q=0}^{P}\binom{P}{Q}|\bm{\mathfrak{Y}}(R,|x|,n)\psi(0)|^{P-Q}|\lambda\bm{\mathfrak{Y}}(R,|x|,n)|^{Q}
\big[\tfrac{1}{2}(\alpha^{Q/2}+(-1)^{Q}\alpha^{Q/2})\big]
\end{align}
\end{proof}
\subsection{Stochastically averaged Bochner formula}
The BW formula for a function $\psi(x)$ on $\mathlarger{\bm{\mathrm{R}}}^{n}$ was given in (2.39). The formula can be extended to random or noisy domains within which there also exists a GRSF ${\mathscr{J}(x)}$ at all $x\in\mathlarger{\bm{\mathcal{D}}}$.
\begin{thm}
Let $\psi(x)$ be a function defined on an open domain $\bm{\mathcal{D}}$ such that the BW formula holds
\begin{align}
&\frac{1}{2}\Delta|\nabla\psi(x)|^{2}=\sum_{j}\nabla_{j}(\Delta\psi(x))\nabla_{j}\psi(x))
+\|\mathlarger{\bm{\mathrm{H}}}\psi(x)\|^{2}\nonumber\\&
\equiv\sum_{j}\nabla_{j}(\Delta\psi(x))\nabla_{j}\psi(x))
=\sum_{ij}(\nabla_{i}\nabla_{j}\psi(x)\otimes\nabla_{i}\nabla_{j}\psi(x))
\end{align}
Let ${\mathscr{J}}$ be a GRSF existing for all $x\in\bm{\mathcal{D}}$ with the regulated covariances
\begin{align}
&\mathbf{E}\bigg\llbracket\nabla^{j}
\Delta{\mathscr{J}(x)}\otimes\nabla^{i}\Delta{\mathscr{J}(x)}
\bigg\rrbracket=\mathbf{E}\bigg\llbracket\nabla^{j}\nabla_{i}\nabla^{i}
{\mathscr{J}(x)}\otimes\nabla^{i}\nabla_{j}\nabla^{j}
{\mathscr{J}(x)}\bigg\rrbracket
=\delta^{ij}\Theta<\infty=\\&\mathbf{E}\bigg\llbracket\nabla_{i}\nabla_{j}
{\mathscr{J}(x)}\otimes\nabla^{i}\nabla^{j}{\mathscr{J}(x)}\bigg\rrbracket
=\delta^{ij}\delta_{ij}\Xi <\infty
\end{align}
If $\overline{\psi(x)}=\psi(x)+\lambda\mathscr{J}(x)$ then the stochastically averaged BW formula is
\begin{align}
&\frac{1}{2}\mathbf{E}\bigg\llbracket\Delta|\nabla\overline{\psi(x)}|^{2}
\bigg\rrbracket=\frac{1}{2}\Delta|\nabla\psi(x)|^{2}+n\xi^{2}(\Theta+\Xi)\nonumber\\&
\equiv\frac{1}{2}\Delta|\nabla\psi(x)|^{2}+C
\end{align}
\end{thm}
Hence, averaging induces a shift of a constant factor which arises from the nonlinear terms.
\begin{proof}
If $\overline{\psi(x)}=\psi(x)+\lambda{\mathscr{J}(x)}$ then the expectation of the randomly perturbed BW formula is
\begin{align}
&\frac{1}{2}\mathbf{E}\bigg\llbracket\Delta|\nabla\overline{\psi(x)}|^{2}\bigg
\rrbracket=\mathbf{E}\bigg\llbracket\sum_{j}\nabla_{j}(\Delta\overline{\psi(x)})\nabla_{j}
\overline{\psi(x)})+\|{\bm{\mathrm{H}}}\overline{\psi(x)}\|^{2}\bigg\rrbracket\nonumber\\&
\equiv\mathbf{E}\bigg\llbracket\sum_{j}\nabla_{j}(\Delta\overline{\psi(x)})
\nabla_{j}\overline{\psi(x)})\bigg\rrbracket
+\mathbf{E}\bigg\llbracket\sum_{ij}(\nabla_{i}\nabla_{j}
\overline{\psi(x)}\otimes\nabla_{i}\nabla_{j}\overline{\psi(x)})\bigg\rrbracket\nonumber\\&
=\mathbf{E}\bigg\llbracket
\sum_{j}^{n}(\Delta\psi(x)+\lambda\Delta{\mathscr{J}(x)})(\nabla^{j}
\psi(x)+\lambda\nabla^{j}\psi(x))\bigg\rrbracket\nonumber\\&
+\mathbf{E}\bigg\llbracket\sum_{ij}^{n}\nabla_{i}\nabla_{j}\psi(x)
+\lambda\nabla_{i}\nabla_{j}{\mathscr{J}(x)})(\nabla_{i}\nabla_{j}\psi(x)+\lambda\nabla_{i}\nabla_{j}
{\mathscr{J}(x)})\bigg\rrbracket\nonumber\\&
=\sum_{j}\nabla_{j}\Delta\psi(x)\nabla^{j}\psi(x)+\lambda\sum_{j}\nabla_{j}\Delta\psi(x)
\mathbf{E}\bigg\llbracket\nabla^{j}{\mathscr{J}(x)}\bigg\rrbracket
+\lambda\sum_{j}\mathbf{E}\bigg\llbracket\nabla_{j}
\Delta{\mathscr{J}(x)}\nabla^{j}\bigg\rrbracket\psi(x)\nonumber\\& +\lambda^{2}\sum_{j}\mathbf{E}\bigg\llbracket\nabla_{j}\Delta\mathscr{J}(x)\nabla^{j}
{\mathscr{J}(x)}\bigg\rrbracket
+\sum_{ij}(\nabla_{i}\nabla_{j}\psi(x)(\nabla_{i}\nabla_{j}\psi(x))
\nonumber\\&+2\lambda\sum_{ij}\nabla_{i}\nabla_{j}\psi(x)\mathbf{E}
\bigg\llbracket\nabla_{i}\nabla_{j}{\mathscr{J}(x)}\bigg\rrbracket
+\lambda^{2}\sum_{ij}\mathbf{E}\bigg\llbracket\nabla_{i}\nabla_{j}
{\mathscr{J}(x)}\nabla_{i}\nabla_{j}{\mathscr{J}(x)}
\bigg\rrbracket\nonumber\\&=+\sum_{j}\nabla_{j}\Delta\psi(x)\nabla^{j}\psi(x)
+\sum_{ij}(\nabla_{i}\nabla_{j}\psi(x))(\nabla_{i}\nabla_{j}\psi(x))\nonumber\\&
+\lambda^{2}\sum_{ij}\mathbf{E}\bigg\llbracket\nabla_{i}\nabla_{j}
{\mathscr{J}(x)}\otimes\nabla_{i}\nabla_{j}{\mathscr{J}(x)}
\bigg\rrbracket\nonumber
\\&=\frac{1}{2}\Delta|\nabla_{j}\psi(x)|^{2}+
\lambda^{2}\sum_{ij}\mathbf{E}\bigg\llbracket\nabla_{i}\nabla_{j}
{\mathscr{J}(x)}\otimes\nabla_{i}\nabla_{j}
\mathscr{J}(x)\bigg\rrbracket\nonumber\\&
=\frac{1}{2}\Delta|\nabla_{j}\psi(x)|^{2}+
\lambda^{2}\sum_{ij}\delta_{ij}\Xi=\frac{1}{2}\Delta|\nabla_{j}\psi(x)|^{2}+\lambda^{2}n\Xi
\equiv \frac{1}{2}\Delta|\nabla_{j}\psi(x)|^{2}+C
\end{align}
\end{proof}
\subsection{Dirichlet problem for Poisson and Laplace equation with noisy boundary data}
We return to the Dirichlet problem for the PE and LE and now reformulate the problem with noisy boundary data; that is, on a domain which has Gaussian random scalar fields existing on its boundary. The Dirichlet BV problem for the Poisson equation in a random/noisy domain was stated in (1.35) and (1.36). Only the case for noisy or random boundary data and deterministic source is considered.
\begin{prop}
For $\psi\in C^{2}\bm{\mathcal{B}}_{R}(0)$ and $(f,g):\bm{\mathcal{B}}_{R}(0)\rightarrow\bm{\mathrm{R}}$, the usual DBVP is $\Delta\psi(x)=f(x), x\in\bm{\mathcal{B}}_{R}(0)$ and $\psi(x)=g(x)$. The possible DBVPs within the random domain or ball $\bm{\mathfrak{B}}_{R}(0)=\lbrace\bm{\mathcal{B}}_{R}(0),{\mathscr{J}(x)})$ are then:
\begin{enumerate}
\item Noisy or randomly perturbed boundary data
\begin{align}
&\Delta\psi(x)=f(x),~x\in\bm{\mathcal{B}}_{R}(0)\\&
\overline{g(x)}=g(x)+\lambda{\mathscr{J}(x)},~x\in\partial\bm{\mathcal{B}}_{R}(0)
\end{align}
where $\lambda>0$
\item The DBVP for the Laplace equation with noisy boundary data follows from setting $f(x)=0$ so that
\begin{align}
&\Delta\overline{\psi(x)}=0 ,~x\in\bm{\mathcal{B}}_{R}(0)\\&
\overline{\psi(x)}=g(x)+\lambda{\mathscr{J}(x)},~x\in\partial\bm{\mathcal{B}}_{R}(0)
\end{align}
\end{enumerate}
The general solution is then randomly perturbed as
\begin{align}
&\overline{\psi(x)}=\int_{\bm{\mathcal{B}}_{R}(0)}f(y)G(x,y)d^{n}y+\frac{R^{2}-|x|^{2}}
{\|\partial\bm{\mathcal{B}}_{1}(0)\|R}\int_{\partial\bm{\mathcal{B}}_{R}(0)}
\frac{\mathscr{J}(y)d^{n-1}y}{|x-y|^{n}}\nonumber\\&
+\frac{R^{2}-|x|^{2}}{\|\partial\bm{\mathcal{B}}_{1}(0)\|R}\int_{\partial\bm{\mathcal{B}}_{R}(0)}
\frac{{\mathscr{J}(y)}d^{n-1}y}{|x-y|^{n}}
\nonumber\\&=\psi(x)+\frac{R^{2}-|x|^{2}}{\|\partial\bm{\mathcal{B}}_{1}(0)\|R}\int_{\partial\bm{\mathcal{B}}_{R}(0)}
\frac{{\mathscr{J}(y)}d^{n-1}y}{|x-y|^{n}}
\end{align}
so that $\mathbf{E}\llbracket \overline{\psi(x)}\rrbracket =\psi(x)$
\end{prop}
\begin{thm}
Given the scenario with $x=(0,0,a)$ or $\|x\|=a$, and $a\le R$, the $P^{th}$-order moments $\mathbf{E}\llbracket|\overline{\psi(a)}|^{p}
\rbrace$ within the ball $\bm{\mathcal{B}}_{R}(0)\subset\bm{\mathrm{R}}^{3} $ are given by
\begin{align}
&\mathbf{M}_{P}(x)=\mathbf{E}
\big\llbracket|\overline{\psi(a)}|^{P}
\big\rrbracket\nonumber\\&\le C\sum_{Q=0}^{p}\binom{P}{Q}|\psi(a)|^{P-Q}\bigg[\frac{1}{2}
\bigg(\beta^{Q/2}+(-1)^{Q}\beta^{Q/2}\bigg)\bigg]\bigg(\frac{R^{2}-|a|^{2}}{4\pi R}\bigg)^{Q}\bigg(\frac{\log|R-a|-\log|R+a|}{a}\bigg)^{Q}\nonumber\\&
\equiv\le C\sum_{Q=0}^{P}\binom{P}{Q}|\psi(a)|^{P-Q}\bigg[\frac{1}{2}
\bigg(\beta^{Q/2}+(-1)^{Q}\beta^{Q/2}\bigg)\bigg]\bigg(\frac{(R-|a|)(R+|a|)}{4\pi R}\bigg)^{Q}\bigg(\frac{\log(\frac{|R-a|}{|R+a|})}{a}\bigg)^{Q}
\end{align}
and the volatility for $p=2$ is
\begin{align}
&\mathlarger{\bm{\mathsf{M}}}_{2}(x)=\mathlarger{\bm{\mathrm{V}}}(x)
=\mathbf{E}\lbrace|\overline{\psi(a)}|^{2}
\rbrace\le|\psi(x)|^{2}+\bigg(\frac{R^{2}-|a|^{2}}{4\pi R}\bigg)^{2}\bigg(\frac{\log|R-a|-\log|R+a|}{a}\bigg)^{2}\nonumber\\&
=|\psi(a)|^{2}+\bigg(\frac{(R-|a|)(R-|a|)}{4\pi R}\bigg)^{2}\bigg(\frac{\log(\frac{|R-a|}{|R+a|})}{a}\bigg)^{2}
\end{align}
where the (unique) underlying deterministic solution $\bm{\mathcal{B}}_{R}(0)\subset\bm{\mathrm{R}}^{3}$ is
\begin{align}
\psi(x)=\int_{\bm{\mathcal{B}}_{R}(0)}f(y)G(x,y)d^{n}y+\frac{R^{2}-|x|^{2}}
{4\pi R}\int_{\partial\bm{\mathcal{B}}_{R}(0)}\frac{g(y)d^{n-1}y}{|x-y|^{n}}
\end{align}
for the Poisson equation, and
\begin{align}
\psi(x)=\frac{R^{2}-|x|^{2}}
{4\pi R}\int_{\partial\bm{\mathcal{B}}_{R}(0)}\frac{g(y)d^{n-1}y}{|x-y|^{n}}
\end{align}
for the Laplace equation
\end{thm}
\begin{proof}
When $\overline{g(x)}=g(x)+\mathscr{J}(x)$, the randomly perturbed solution within a ball $\bm{\mathcal{B}}_{R}(0)\subset\bm{\mathrm{R}}^{3}$ is
\begin{equation}
\overline{\psi(x)}=\psi(x)+\frac{R^{2}-|x|^{2}}{4\pi R}\int_{\partial\bm{\mathcal{B}}_{R}(0)}
\frac{{\mathscr{J}(y)}d^{2}y}{|x-y|^{3}}
\end{equation}
The $P^{th}$-order moments are then
\begin{align}
&\mathbf{M}_{P}(x)=\mathbf{E}\bigg\llbracket\bigg|\psi(x)+\frac{R^{2}-|x|^{2}}{4\pi R}\int_{\partial\bm{\mathcal{B}}_{R}(0)}\frac{\mathscr{J}(y)d^{2}y}{|x-y|^{3}}\bigg|^{p}
\bigg\rrbracket
\nonumber\\&=\mathbf{E}\bigg\llbracket\sum_{Q=0}^{P}\binom{P}{Q}
|\psi(x)|^{P-Q}\bigg( \frac{R^{2}-|x|^{2}}{4\pi R}\int_{\partial\bm{\mathcal{B}}_{R}(0)}\frac{\mathscr{J}(y)d^{2}y}{|x-y|^{3}}\bigg)^{Q}
\bigg\rrbracket
\nonumber\\&
=\sum_{Q=0}^{P}\binom{P}{Q}
|\psi(x)|^{P-Q}\mathbf{E}\bigg\llbracket\bigg( \frac{R^{2}-|x|^{2}}{4\pi R}\int_{\partial\bm{\mathcal{B}}_{R}(0)}\frac{{\mathscr{J}(y)}d^{2}y}{|x-y|^{3}}
\bigg)^{Q}\bigg\rrbracket
\nonumber\\&
\le C\sum_{Q=0}^{P}\binom{P}{Q}
|\psi(x)|^{P-Q}\bigg( \frac{R^{2}-|x|^{2}}{4\pi R}\bigg)^{Q}\int...
\int_{\partial\bm{\mathcal{B}}_{R}(0)}
\frac{\mathbf{E}\bigg\llbracket|{\mathscr{J}(y)}|^{Q}\bigg
\rrbracket d^{2}y...d^{2}y}{|x-y|^{3}\times...\times|x-y|^{3}}
\nonumber\\&
= C\sum_{Q=0}^{P}\binom{P}{Q}
|\psi(x)|^{P-Q}\mathbf{E}\bigg\llbracket\bigg( \frac{R^{2}-|x|^{2}}{4\pi R}\int...\int_{\partial\bm{\mathcal{B}}_{R}(0)}\frac{|{\mathscr{J}(y)}|^{Q}d^{2}y...d^{2}y}{|x-y|^{3}\times...\times|x-y|^{3}}
\bigg)\bigg\rrbracket
\nonumber\\&
= C\sum_{Q=0}^{P}\binom{P}{Q}
|\psi(x)|^{P-Q}\bigg[\frac{1}{2}\bigg(\beta^{Q/2}+(-1)^{Q}\beta^{Q/2}\bigg)\bigg]
\bigg( \frac{R^{2}-|x|^{2}}{4\pi R}\bigg)^{Q}\int...\int_{\partial\bm{\mathcal{B}}_{R}(0)}
\frac{d^{2}y...d^{2}y}{|x-y|^{3}\times...\times|x-y|^{3}}
\nonumber\\&
= C\sum_{Q=0}^{P}\binom{P}{Q}
|\psi(x)|^{P-Q}\bigg[\frac{1}{2}\bigg(\beta^{Q/2}+(-1)^{Q}\beta^{Q/2}\bigg)\bigg]
\bigg(\frac{R^{2}-|x|^{2}}{4\pi R}\bigg)^{Q}\bigg(\int_{\partial\bm{\mathcal{B}}_{R}(0)}
\frac{d^{2}y}{|x-y|^{3}}\bigg)^{Q}
\end{align}
If $x=(0,0,a)$ or $\|x\|=a$ and $a\le R$ then the boundary integral over the ball can be evaluated. Using spherical coordinates with $\zeta=\cos\theta$.
\begin{align}
&\mathlarger{\bm{\mathsf{M}}}_{P}(x)
\le C\sum_{Q=0}^{P}\binom{P}{Q}
|\psi(a)|^{P-Q}\bigg[\frac{1}{2}\bigg(\beta^{Q/2}+(-1)^{Q}\beta^{Q/2}\bigg)\bigg]
\bigg( \frac{R^{2}-|a|^{2}}{4\pi R}\bigg)^{Q}\nonumber\\&\times\bigg(\int_{0}^{R}\bigg|\int_{0}^{\pi}\int_{-\pi}^{\pi}
\frac{\sin\theta d\theta d \varphi}{(r^{2}-2ra\cos\theta+a^{2})^{3/2}}
\bigg|rdr \bigg)^{Q}\nonumber\\&=C\sum_{Q=0}^{P}\binom{P}{Q}
|\psi(a)|^{P-Q}\bigg[\frac{1}{2}\bigg(\alpha^{Q/2}+(-1)^{Q}\alpha^{Q/2}\bigg)\bigg]
\bigg( \frac{R^{2}-|a|^{2}}{4\pi
R}\bigg)^{Q}\nonumber\\&\times\bigg(2\pi\int_{0}^{R}\bigg|\int_{-1}^{1}
\frac{d\zeta}{(r^{2}-2raQ+a^{2})^{3/2}}\bigg|rdr \bigg)^{Q}\nonumber\\&
=C\sum_{Q=0}^{P}\binom{P}{Q}
|\psi(a)|^{P-Q}\bigg[\frac{1}{2}\bigg(\alpha^{Q/2}+(-1)^{Q}\alpha^{Q/2}\bigg)\bigg]
\bigg( \frac{R^{2}-|a|^{2}}{4\pi
R}\bigg)^{Q}\nonumber\\&\times\bigg(2\pi\int_{0}^{R}\bigg|
\frac{1}{ar\sqrt{-2ar+r^{2}+a^{2}}}-\frac{1}{ar\sqrt{2ar+r^{2}+a^{2}}}\bigg|rdr \bigg)^{Q}\nonumber\\&=C\sum_{Q=0}^{P}\binom{P}{Q}
|\psi(a)|^{P-Q}\bigg[\frac{1}{2}\bigg(\alpha^{Q/2}+(-1)^{Q}\alpha^{Q/2}\bigg)\bigg]
\bigg( \frac{R^{2}-|a|^{2}}{4\pi R}\bigg)^{Q}\nonumber\\&\times\bigg(\frac{2\pi \log|R-a|-2\pi \log|R+a|}{a}\bigg)^{Q}
\nonumber\\&
\equiv C\sum_{Q=0}^{P}\binom{P}{Q}
|\psi(x)|^{P-Q}\bigg[\frac{1}{2}\bigg(\alpha^{Q/2}+(-1)^{Q}\alpha^{Q/2}\bigg)\bigg]
\bigg( \frac{(R+|a|)(R-|a|)}{4\pi
R}\bigg)^{Q}\nonumber\\&\times\bigg({\frac{2\pi}{a} \log
\bigg(\frac{|R-a|}{|R+a|}\bigg)}\bigg)^{Q}
\end{align}
\begin{cor}
The moments $\mathlarger{\bm{\mathsf{M}}}_{p}(x)$ and volatility $\mathlarger{\bm{\mathsf{V}}}(x)$ blowup at the centre of the ball with $a=0$ and for $R\rightarrow\infty$, and decay and vanish at the boundary/surface of the ball $a=R$.
\end{cor}
\end{proof}
\subsection{Laplace equation and Dirichlet problem on a disc with noisy boundary data}
We return to the problem of the Laplace equation on a disc subject to Dirichlet boundary conditions, discussed in (-) but now with noisy data or random perturbations on the boundary. It will be shown that this induces a GRSF in the interior of the disc which is also 'stochastically harmonic'.
\begin{thm}
Let $\bm{\mathcal{D}}\subset\bm{\mathrm{R}}^{2}$ be a disc of radius $R$. The Dirichlet problem on this disc is
\begin{align}
&\Delta_{(r,\theta)}\psi(r,\theta)=0,~(r,\theta)\in\bm{\mathcal{D}}\\&
\psi(R,\theta)=\psi(\theta)=g(\theta),~\theta\in\partial\bm{\mathcal{D}}~or~\theta\in[0,2\pi]
\end{align}
with solution
\begin{equation}
\psi(r,\theta)=\frac{1}{2\pi}\int_{0}^{2\pi} \frac{R^{2}-r^{2}g(\beta)d\beta}{R^{2}-2rR\cos(\theta-\beta)+r^{2}}
\end{equation}
Now consider randomly perturbed or noisy data on the boundary such that
\begin{align}
&\Delta_{(r,\theta)}\psi(r,\theta)=0,~(r,\theta)\in\bm{\mathcal{D}}\\&
\overline{\psi R,\theta)}=\overline{\psi(\theta)}=g(\theta)
+\epsilon\mathscr{J}(\theta),~\theta\in\partial\bm{\mathcal{D}}~or~\theta\in[0,2\pi]
\end{align}
with $\epsilon>0$ small, where $\mathscr{J}(R,\theta)\equiv \mathscr{J}(\theta)$ is the GRSF on the disc boundary with covariance
\begin{equation}
\mathbf{E}\llbracket{\mathscr{J}(\theta)}\otimes {\mathscr{J}(\theta)}^{\prime}\rrbracket
=K(\theta,\theta;\eta)= \exp(-\|x-y\|/\eta) \nonumber
\end{equation}
which vanishes for $|\theta-\theta^{\prime}|>\eta$, and $K(\beta,\beta^{\prime};\xi)=1$ when $\beta=\beta^{\prime}$. From Thm(-) it follows that
\begin{align}
&\mathbf{E}\bigg\llbracket|{\mathscr{J}(\theta)}|^{P}\bigg\rrbracket=\mathbf{E}\bigg\llbracket
{\mathscr{J}(\theta)}\underbrace{\otimes...\otimes}_{P~times}{\mathscr{J}(\theta)}
\bigg\rrbracket=0,~ odd~p\\&
\mathbf{E}\bigg\llbracket|{\mathscr{J}(\theta)}|^{P}\bigg\rrbracket
=\mathbf{E}\bigg\llbracket
{\mathscr{J}(\theta)}\underbrace{\otimes...\otimes}_{P~times}
{\mathscr{J}(\theta)}
\bigg\rrbracket=1,~ even~P
\end{align}
Then:
\begin{enumerate}
\item The solution of $\Delta\overline{\psi(r,\theta)}=0$ is
\begin{align}
&\overline{\psi(r,\theta)}=\frac{1}{2\pi}\int_{0}^{2\pi}\bigg| \frac{R^{2}-r^{2}}{R^{2}-2rR\cos(\theta-\beta)+r^{2}}\bigg|g(\beta)d\beta\nonumber\\&+\frac{1}{2\pi}\int_{0}^{2\pi} \bigg|\frac{R^{2}-r^{2}}{R^{2}-2rR\cos(\theta-\beta)+r^{2}}\bigg|{\mathscr{J}(\beta)}d\beta\nonumber\\&
=\int_{0}^{2\pi}\Pi(R,r,\theta,\beta)g(\beta)d\beta+\int_{0}^{2\pi}\Pi(R,r,\theta,\beta)
{\mathscr{J}(\beta)}d\beta
\end{align}
so that a GRSF is induced within the interior of the disc.
\item The expectation vanishes so that $\overline{\psi(r,\theta)}$ is Gaussian. The GRSF is also stochastically harmonic in that
    \begin{equation}
    \mathbf{E}\bigg\llbracket\Delta\overline{\psi(r,\theta)}\bigg\rrbracket=0\nonumber
\end{equation}
\item  The $P^{th}$-order moments are finite and given by
\begin{align}
\mathbf{E}\big\llbracket|\overline{\psi(r,\theta)}|^{P}\rrbracket&=\frac{1}{2}\sum_{Q=1}^{P}
|\psi(r,\theta)|^{P-Q}\bigg\llbracket \bigg[\bigg(\frac{1}{2\pi}\bigg(
\tan^{-1}\bigg(\frac{|R+r|\tan(-\tfrac{1}{2}\theta)}{|R-r|}\bigg)\nonumber\\&-\frac{1}{2\pi}
\tan^{-1}\bigg(\frac{|R+r|\tan(\pi-\tfrac{1}{2}\theta)}{|R-r|}\bigg)\bigg)\bigg]^{Q}
\bigg\rrbracket\nonumber\\&
+\frac{1}{2}\sum_{Q=1}^{P}
|\psi(r,\theta)|^{P-Q}\bigg\lbrace(-1)^{Q} \bigg[\bigg(\frac{1}{2\pi}\bigg(
\tan^{-1}\bigg(\frac{|R+r|\tan(-\tfrac{1}{2}\theta)}{|R-r|}\bigg)\nonumber\\&-\frac{1}{2\pi}
\tan^{-1}\bigg(\frac{|R+r|\tan(\pi-\tfrac{1}{2}\theta)}{|R-r|}\bigg)\bigg)\bigg]^{Q}
\bigg\rbrace
\end{align}
\item The covariance or 2-point function between the GRSF at points $(r,\theta)$ and $(r^{\prime},\theta^{\prime})$ on the disc is
    \begin{align}
    &\mathbf{E}\bigg\llbracket\overline{ \psi(r,\theta)}\otimes\overline{\psi(r^{\prime},\theta^{\prime})}\bigg\rrbracket\nonumber\\&=\frac{1}{2\pi^{2}}\int_{0}^{2\pi}\int_{0}^{2\pi}\bigg|\frac{|R^{2}-r^{2}|}{R^{2}-2rR\cos(\theta-\beta)+r^{2})}\bigg|
    \nonumber\\&\times \bigg|\frac{|{R^{\prime}}^{2}-{r^{\prime}}^{2}|}{R^{\prime^{2}}-2r^{\prime}R^{\prime}
    \cos(\theta^{\prime}-\beta^{\prime})+{r^{\prime}}^{2})}\bigg|
    g(\beta)g(\beta^{\prime})d\beta d\beta^{\prime}\nonumber\\&
    +\frac{1}{2\pi^{2}}\int_{0}^{2\pi}\int_{0}^{2\pi}\bigg|\frac{|R^{2}-r^{2}|}{R^{2}-2rR\cos(\theta
    -\beta)+r^{2})}
    \bigg|\nonumber\\&\times\bigg|\frac{|{R^{\prime}}^{2}-{r^{\prime}}^{2}|}{R^{\prime^{2}}-
    2r^{\prime}R^{\prime}
    \cos(\theta^{\prime}-\beta^{\prime})+{r^{\prime}}^{2})}\bigg|
    \exp(-(\beta-\beta^{\prime})/\eta)d\beta d\beta^{\prime}
    d\beta d\beta^{\prime}<\infty
    \end{align}
which gives the result. The final result follows from using a regulated colored noise covariance
\end{enumerate}
\end{thm}
The proof is given in Appendix B.
\begin{cor}
The following hold:
\begin{enumerate}
\item The moments are finite for all $(r,\theta)$ so that
\begin{equation}
\mathbf{E}\big\llbracket |\psi(r,\theta)|^{2}\big\rrbracket<\infty
\end{equation}
The volatility is defined for $p=2$ and is finite so that $\mathbf{E}\big\llbracket |\psi(r,\theta)|^{2}\big\rrbracket<\infty$.
\item The limits are
\begin{equation}
\lim_{R\uparrow\infty}\mathbf{E}\big\llbracket |\psi(r,\theta)|^{2}\big\rrbracket=0
\end{equation}
\begin{equation}
\lim_{(r,\theta)\uparrow 0}\mathbf{E}\big\llbracket |\psi(r,\theta)|^{2}\big\rrbracket=0
\end{equation}
\begin{equation}
\lim_{r\uparrow R}\mathbf{E}\big\llbracket |\psi(r,\theta)|^{2}\big\rrbracket=0
\end{equation}
So that the moments and volatility vanish on the boundary and the centre of the disc and in the limit that the disc radius goes to infinity.
\end{enumerate}
\end{cor}
\clearpage
\appendix
\renewcommand{\theequation}{\Alph{section}.\arabic{equation}}
\section{Properties of Gaussian Random Scalar Fields on $\bm{R}^{n}$}
This appendix provides basic definitions of existence, properties, correlations, statistics,
derivatives and integrals of time-independent Gaussian random scalar fields (GRSFs) on
$\mathbf{R}^{n}$. Classical random fields correspond naturally to structures, and
properties of systems. It will be sufficient to briefly establish the following:
\begin{enumerate}
\item The existence of random fields on $\mathbf{R}^{n}$.
\item Statistics, moments, covariances and correlations.
\item Unique properties of Gaussian random fields. (GRVFs.)
\item Sample path continuity, differentiability and the existence of the derivatives of a
    GRSF.
\item Stochastic integration.
\end{enumerate}
\begin{defn}
Let $(\Omega,\mathcal{F},\bm{\mathrm{I\!P}})$ be a probability space. Within the probability
triplet, $(\Omega,\mathcal{F})$ is a\emph{measurable space}, where $\mathcal{F}$ is the
$\sigma$-algebra (or Borel field) that should be interpreted as being comprised of all
reasonable subsets of the state space $\Omega$. Then $\bm{\mathrm{I\!P}}$ is a function such
that $\bm{\mathrm{I\!P}}:\mathcal{F}\rightarrow [0,1]$, so that for all $A\in\mathcal{F}$,
there is an associated probability $\bm{\mathrm{I\!P}}(A)$. The measure is a probability
measure when $\bm{\mathrm{I\!P}}(\Omega)=1$. The probability space obeys the Kolmogorov
axioms such that $\bm{\mathrm{I\!P}}(\Omega)=1$ and $0\le\bm{\mathrm{I\!P}}(A_{i})\le 1$ for
all sets $A_{i}\in \mathcal{F}$. And if $A_{i}\bigcap A_{j}=\empty$, then
$\bm{\mathrm{I\!P}}(\bigcup_{i=1}^{\infty}A_{i})=\sum_{i=1}^{\infty}\bm{\mathrm{I\!P}}(A_{i})$.
Let $x^{i}\subset\bm{\mathcal{D}}\subset\bm{\mathrm{R}}^{n}$ be Euclidean coordinates and
let $(\Omega,{\mathcal{F}},\bm{\mathrm{I\!P}})$ be a probability space. Let
$\widehat{\mathscr{J}}(x;\omega)$ be a random scalar function that depends on the coordinates
$x\subset\bm{\mathcal{D}}\subset\bm{\mathrm{R}}^{n}$ and also $\omega\in\Omega$. Given any
pair $(x,\omega)$ there exists maps
$M:\bm{\mathrm{R}}^{n}\times\Omega\rightarrow\bm{\mathrm{R}}$ such that $
M:(\omega,x)\longrightarrow\mathscr{J}(x);\omega)$, so that
${\mathscr{J}(x,\omega)}$ is a random variable or field on
$\bm{\mathcal{D}}\subset\bm{\mathrm{R}}^{n}$ with respect to the probability space
$(\Omega,\mathcal{F},\bm{\mathrm{I\!P}})$. A random field is then essentially a family of
random variables $\lbrace\mathscr{J}(x;\omega)\rbrace$ defined with respect to the space
$(\Omega,\mathcal{F},\bm{\mathrm{I\!P}})$ and $\bm{\mathrm{R}}^{n}$.
\end{defn}
The fields can also include a time variable $t\in\bm{\mathrm{R}}^{+}$ so that given any
triplet
$(x,t,\omega)$ there is a mapping $f:\bm{\mathrm{R}}\times\bm{\mathrm{R}}^{n}
\times\Omega\rightarrow\bm{\mathrm{R}}$ such that
$f:(x,t,\omega)\longrightarrow{\mathscr{J}}(x,t;\omega)$. However, it will be sufficient to
consider fields that vary randomly in space only. The expected value of the random field with
respect to the space $(\Omega,\bm{\mathscr{J}},\bm{\mathrm{I\!P}})$ is defined as follows
\begin{defn}
Given the random scalar field $\bm{\mathscr{J}(x;\omega)}$, then if $
\int_{\Omega}{\mathscr{J}(x;\omega)}d\bm{\mathrm{I\!P}}(\omega)<\infty$, the stochastic
expectation of $\widehat{\mathscr{J}}(x;\omega)$ is
\begin{equation}
\bm{\mathrm{E}}\big\llbracket\mathscr{J}(x;\omega)\big\rrbracket=
\int_{\Omega}\mathscr{J}(x;\omega)d\bm{\mathrm{I\!P}}(\omega)
\end{equation}
\end{defn}
\begin{defn}
Let $(\Omega,\mathbb{F},\bm{\mathrm{I\!P}})$ be a probability space, then an
$L_{p}(\Omega,\mathcal{F},\bm{\mathrm{I\!P}})$ space or an $L_{p}$-space for $p\ge 1$ is a
linear normed space of random scalar fields that satisfies the conditions
\begin{equation}
\mathbf{E}\bigg\llbracket|\mathscr{J}(x;\omega)|^{p}\bigg\rrbracket
=\int_{\Omega}|{\mathscr{J}(x;\omega)}|^{p}d\bm{\mathrm{I\!P}}(\omega)<\infty
\end{equation}
and the corresponding norm can be defined as $\|\widehat{\mathscr{J}}(x)\|
=(\bm{\mathcal{E}}\llbracket|\widehat{\mathscr{J}}(x;\omega)|^{p}\rrbracket)^{1/p}$ with the
usual Euclidean norm for $p=2$.
\end{defn}
\begin{defn}
Let $(x,y)\in\bm{\mathcal{D}}\subset\bm{\mathrm{R}}^{n}$ and
$(\omega,\eta)\in\bm{\mathcal{D}}$. The covariance of the field at these points is then
formally
\begin{equation}
\mathbf{E}\bigg\llbracket{\mathscr{J}(x,\omega)}\otimes {\mathscr{J}}(y,\xi)
\bigg\rrbracket=\iint_{\Omega}{\mathscr{J}(x,\omega)}\otimes
{\mathscr{J}(y,\xi)}d\bm{\mathrm{I\!P}}(\omega)d\bm{\mathrm{I\!P}}(\xi)
\end{equation}
\end{defn}
The second-order correlations, moments and covariances are now
\begin{defn}
Let $x^{i},y^{i}\in\bm{\mathcal{D}}\subset\bm{R}^{3}$ and let $\omega,\xi\in\Omega$.
The expectations or mean values of the fields ${\mathscr{J}}_{i}(x,\omega)$
and ${\mathscr{J}}(y,\xi)$ are
\begin{align}
&\bm{\mathrm{M}}_{1}(x)=\bm{\mathrm{E}}\left\llbracket{\mathscr{J}(x)}\right\rrbracket
=\int_{\Omega}{\mathscr{J}(x,\omega)}d\bm{\mathrm{I\!P}}(\omega)\\&
\bm{\mathrm{M}}_{1}(y)=\bm{\mathrm{E}}\left\llbracket{\mathscr{J}(y)}\right\rrbracket=
 \int_{\Omega}{\mathscr{J}(y,\xi)}d\bm{\mathrm{I\!P}}(\xi)
\end{align}
The 2nd-order moment or expectation is
\begin{equation}
\bm{\mathrm{M}}_{1}(x)\llbracket{\mathscr{J}(x)}\otimes {\mathscr{J}(y)}\rrbracket=
\iint_{\Omega}{\mathscr{J}}(x,\omega)\otimes{\mathscr{J}}(y,\xi)
d\bm{\mathrm{I\!P}}(\omega)d\bm{\mathrm{I\!P}}(\xi)
\end{equation}
The covariance is then
\begin{align}
&\bm{\mathrm{K}}\llbracket{\mathscr{J}(x)}\otimes
{\mathscr{J}(y)}\rrbracket=\bm{\mathrm{E}}\bigg\llbracket({\mathscr{J}(x)}-
{\mathscr{J}(y)})\otimes ({\mathscr{J}(x)}-
{\mathscr{J}(y)})\bigg\rrbracket\nonumber\\&
=\iint_{\Omega}\bigg(\bm{\mathscr{J}}(x;\omega)-
\bm{\mathscr{J}}(y,\xi)\bigg))\otimes\bigg((\bm{\mathscr{j}}(x;\omega)
-\bm{\mathscr{J}}(y,\xi)\bigg)
d\bm{\mathrm{I\!P}}(\omega)d\bm{\mathrm{I\!P}}(\xi)
\end{align}
so that
\begin{align}
&\bm{\mathrm{K}}\llbracket{\mathscr{J}(x)}\otimes
{\mathscr{J}(y)}\rrbracket=\bm{\mathrm{E}}\bigg\llbracket\widehat{\mathscr{J}(x)}
\otimes\widehat{\mathscr{J}(y)}
\bigg\rrbracket-\bm{\mathrm{E}}\big\llbracket{\mathscr{J}}(x)
\big\rrbracket\bm{\mathrm{E}}\big\llbracket{\mathscr{J}}(y)\big\rrbracket\nonumber\\&
=\bm{\mathrm{E}}\bm{\mathscr{J}}\bigg\llbracket\widehat{\mathscr{j}(x)}
\otimes\widehat{\mathscr{J}(y)}
\bigg\rrbracket-\bm{\mathrm{I\!M}}_{1}(x)\bm{\mathrm{I\!M}}_{1}(y)
\end{align}
\end{defn}
\begin{defn}
Given a set of fields $\widehat{\mathscr{J}}_{i_{1}}(x_{1}),...,\widehat{\mathscr{J}}
_{i_{n}}(x_{n})$ at points $(x_{1}...x_{n})\in\bm{\mathcal{D}}$ then the $n^{th}$-order
moments and cumulants are
\begin{equation}
\bm{\mathrm{E}}\llbracket\widehat{\mathscr{J}}_{i_{1}}(x_{1})\otimes...
\otimes\widehat{\mathscr{J}}_{i_{n}}(x_{n})\rrbracket
\end{equation}
\begin{equation}
\bm{\mathrm{K}}\llbracket\widehat{\mathscr{J}}_{j_{1}}(x_{1})\otimes...
\otimes\widehat{\mathscr{J}}_{j_{n}}(x_{n})\rrbracket
\end{equation}
where at second order
\begin{equation}
\bm{\mathrm{K}}\llbracket\mathscr{J}(x)\otimes\mathscr{J}(y)\rrbracket
=\bm{\mathrm{E}}\llbracket \mathscr{J}(x)\otimes\mathscr{J}(y)\rrbracket-
\bm{\mathrm{M}}_{1}(x)\bm{\mathrm{M}}_{1}(y)
\end{equation}
\end{defn}
A very important class of random fields are the Gaussian random vector fields (GRVFS) which
are characterized only by the first and second moments. The GRVFS can also be isotropic,
homogenous and stationary. The details will be made more precise but the advantages of GRVFs
are briefly enumerated.
\begin{enumerate}
\item GRVFS have convenient mathematical properties which generally simplify
    calculations;indeed, many results can only be evaluated using Gaussian fields.
\item A GRVF can be classified purely by its first and second moments and high-order
    moments and cumulants can be ignored.
\item Gaussian fields accurately describe many natural stochastic processes including
    Brownian motion.
\item A large superposition of non-Gaussian fields can approach a Gaussian field (Feller
    1966.)
\end{enumerate}
For this paper, the following definitions are sufficient for isotropic GRVFS. (More details
can be found in refs...)
\begin{defn}
Any GRVF has normal probability density functions. The following always hold:
\begin{enumerate}
\item The first moment vanishes so that
\begin{equation}
\bm{\mathrm{I\!M}}_{1}(x)=\bm{\mathrm{E}}\big\llbracket\mathscr{J}(x;\omega)\big\rrbracket
=\int_{\Omega}\mathscr{J}(x;\omega)d\bm{\mathrm{I\!P}}(\omega)=0
\end{equation}
\item The covariance then reduces to the '2-point' function
\begin{align}
\bm{\mathcal{K}}\llbracket{\mathscr{J}(x)}\otimes
{\mathscr{J}(y)}\rrbracket&=\bm{\mathrm{E}}\big\llbracket\widehat{\mathscr{J}(x)}
\otimes\widehat{\mathscr{J}(y)}
\big\rrbracket-\bm{\mathrm{E}}\big\llbracket{\mathscr{J}}(x)
\big\rrbracket\bm{\mathrm{E}}\big\llbracket{\mathscr{J}}(y)\big\rrbracket\nonumber\\&
=\bm{\mathrm{E}}\big\llbracket\widehat{\mathscr{J}(x)}
\otimes\widehat{\mathscr{J}(y)}
\big\rrbracket=\zeta J(x,y;\ell)
\end{align}
where $\ell$ is a correlation length such that $J(x,y;\ell)\rightarrow 0$ for $\|x-y\|\gg
\ell.$
\item The GRSF is regulated at all points $x$ if $J(x,x;\ell)<\infty1$. Here,
    $J(x,x;\ell)=1$. For a white-in-space random field or noise
\begin{align}
\bm{\mathrm{E}}\big\llbracket\widehat{\mathscr{J}}(x)
\otimes\widehat{\mathscr{J}}(y)
\big\rrbracket=\zeta\delta^{n}(x-y)
\end{align}
which blows up at $x=y$. An example of a regulated 2-point correlation for a GRSF would be
a colored-in-space noise of the form
\begin{align}
\bm{\mathrm{E}}\big\llbracket\widehat{\mathscr{J}}(x)
\otimes\widehat{\mathscr{J}}(y)
\big\rrbracket=\zeta J(x,y;\ell)=\zeta \exp\left(-\frac{\|x-y\|}{\ell}\right)
\end{align}
\end{enumerate}
\end{defn}
\begin{defn}
The GRVF is isotropic if $J(x,y;\ell)=J(y,x;\ell)$ depends only on the separation $\|x-y\|$
and is stationary if $J(x+\delta x,y+\delta y)=J(x,y;\ell)$. Hence, the 2-point function or
Greens function is translationally and rotationally invariant $\mathrm{R}^{n}$.
\end{defn}
Having established the basic properties of GRSFS on $\bm{\mathrm{R}}^{n}$, and as a
prerequisite to differentiation and integration of random vector fields, the geometric
properties are briefly considered in relation to continuity of the sample paths.
\begin{defn}
Let $\lbrace x_{\alpha}\subset\bm{\mathcal{D}}\subset\bm{\mathds{R}}^{n}$ be a sequence of
points in a domain $\mathds{D}$ such that $x_{\alpha}\rightarrow x$ as
$\alpha\rightarrow\infty$ or $\lim_{\alpha\rightarrow \infty}\|x_{\alpha}-x\|=0$. Let
$\mathscr{J}(x_{\alpha})$ and $\mathscr{J}(x)$ be random vector fields at these points. (Not
necessarily Gaussian.) Then:
\begin{enumerate}
\item The random field $\mathscr{J}(x)$ has \emph{continuous sample paths} with unit
    probability in $\bm{\mathcal{D}}$ if for every sequence $\lbrace
    x_{\alpha}\rbrace\in\bm{\mathcal{D}}$ with $\lim_{\alpha\rightarrow
    \infty}\|x_{\alpha}-x\|=0$ and
    $\bm{\mathrm{I\!P}}[\lim_{\alpha\rightarrow\infty}|\mathscr{J}(x_{\alpha})-
\mathscr{J}(x)|=0;x\in\bm{\mathcal{D}}]=1 $
Continuous sample paths with unit probability, means that with a probability of one, there
are no discontinuities within the entire domain $\bm{\mathcal{D}}$. This condition is also
called sample path continuity.
\item The random vector field $\mathscr{J}(x)$ has almost surely continuous sample paths if
    for every sequence $\lbrace x_{\alpha}\rbrace\in\mathds{D}$ with
    $\lim_{\alpha\rightarrow \infty}\|x_{\alpha}-x\|=0$ so that $
\bm{\mathrm{I\!P}}[\lim_{\alpha\rightarrow\infty}|\mathscr{J}(x_{\alpha})-
\mathscr{J}(x)|=0]=1$. However, there may be $x\in\bm{\mathcal{D}}$ for which the
condition does not hold giving discontinuities within $\bm{\mathcal{D}}$.
\item The random vector field $\mathscr{J}(x)$ is mean-square continuous in
    $\bm{\mathcal{D}}$ if for every sequence $\llbracket
    x_{\alpha}\rrbracket\in\bm{\mathcal{D}}$ with $\lim_{\alpha\rightarrow
    \infty}\|x_{\alpha}-x\|=0$ so that
$\bm{\mathrm{E}}\llbracket\lim_{\alpha\rightarrow\infty}
|\mathscr{J}(x_{\alpha})-\mathscr{J}(x)|^{2}
\rrbracket=0 $
\end{enumerate}
\end{defn}
More details can be found in [REF]
\begin{defn}
A GRSF $\mathscr{J}(x)$ is almost surely continuous at
$x\in\bm{\mathcal{D}}\in\bm{\mathrm{R}}^{n}$ if $
\mathscr{J}(x+\beta)\longrightarrow\mathscr{J}(x) $ as $\beta\rightarrow 0$
\end{defn}
When this holds for all $x\in\bm{\mathcal{D}}$ then this is known as 'sample function
continuity'. The following result due to Adler (1980) gives the sufficient condition for
continuous sample paths within a domain $\bm{\mathcal{D}}$
\begin{lem}
Let $\mathscr{J}(x)$ be a RVF on $\bm{\mathcal{D}}\subset\bm{\mathrm{R}}^{n}$. Then if for
some $C>0$ and $\lambda>0$ with $\eta>\lambda$
\begin{equation}
\bm{\mathrm{E}}\llbracket\big|\mathscr{J}(x+\beta)-\mathscr{J}(x)
\big|^{\lambda}\rrbracket \le\frac{C|\zeta|2n}{|\ln|\beta||^{1+\eta}}
\end{equation}
If $\mathscr{J}(x)$ is a Gaussian random field then given some $C>0$ and some $\epsilon>0$
\begin{equation}
\bm{\mathrm{E}}\llbracket|\mathscr{J}(x+\beta)-\mathscr{J}(x)|^{2}\rrbracket
\le \frac{C}{\ln|\beta|^{1+\epsilon}}
\end{equation}
\end{lem}
\subsection{Differentiability and existence of derivatives}
Let $\mathscr{J}(x)$ be a GRSF, existing for all
$x\in\bm{\mathcal{D}}\subset\bm{\mathrm{R}}^{n}$, with covariance
\begin{equation}
\bm{\mathrm{E}}({\mathscr{J}(x)}\otimes{\mathscr{J}(y)})=
\bm{\mathrm{E}}\llbracket{\mathscr{J}(x)}\otimes{\mathscr{J}(y)}\rrbracket
-\bm{\mathrm{M}}_{1}(x)\bm{\mathrm{M}}_{2}(y)
\end{equation}
where $\bm{\mathcal{M}}_{1}(x)=\bm{\mathrm{E}}\llbracket{\mathscr{J}(x)}\rrbracket=0$.
\begin{defn}
Let $\nabla\mathscr{J}(x)$ denote the gradient of a GRF. Let $\mu_{i}$ be a unit vector along
the $i^{th}$ direction such that $\mu_{1}=(1,0,0,0...), \mu_{2}=(0,1,0,0,...)$ etc. with
$\|\mu_{i}\|=1$. A GRF is differentiable in the mean square sense(MSS) if
\begin{equation}
\nabla \mathscr{J}(x)=\lim_{|\mathcal{R}|\uparrow 0}|
\frac{\mathscr{J}(x+|\mathcal{R}|\mu_{i})-\mathscr{J}(x)}{|\mathcal{R}|}
\end{equation}
which implies that
\begin{equation}
\lim_{\mathcal{R}\uparrow 0}\bm{\mathcal{E}}\bigg\llbracket\bigg|
\frac{\mathscr{J}(x+|\mathcal{R}|\mu_{i})-\mathscr{J}(x)}{|\mathcal{R}|}
-\nabla\mathscr{J}(x)|^{2}
\bigg|\bigg\rrbracket=0
\end{equation}
The Laplacian or second-order partial derivative at $x$ is defined as
\begin{equation}
\nabla_{(x)}\nabla_{(x)}{\mathscr{J}(x)}
=\lim_{\mathcal{R},\mathcal{J}\uparrow
0}\frac{1}{|\mathcal{R}\mathscr{J}|}\bigg[{\mathscr{J}(x+|\mathcal{R}|\mu_{i}
+|\mathscr{J}|\mu_{j})}-
{\mathscr{J}(x+|\mathcal{R}|\mu_{i})}-{\mathscr{J}(x+|\mathcal{S}|\mu_{i})}+{\mathscr{J}(x)}
\bigg]
\end{equation}
\end{defn}
An alternative definition, which is probably more useful, is given as follows
\begin{lem}
A GRF is differentiable in the mean-square sense(MSS)with respect to $\mu_{i}$ if
\begin{equation}
\lim_{\mathcal{R},\mathcal{S}\uparrow
0}\bigg\llbracket\bigg|\frac{\mathscr{J}(x+|\mathcal{R}|\mu_{i})-\mathscr{J}(x)}{|\mathcal{R}|}-
\frac{\mathscr{J}(x+|\mathcal{S}|\mu_{i})-\mathscr{J}(x)}{|\mathcal{S}|}\bigg|^{2}
\bigg\rrbracket=0
\end{equation}
Then a GRF is differentiable if
\begin{enumerate}
\item $\bm{M}_{1}(x)$ is differentiable.(For GRFs $\bm{M}_{1}(x)=0$ and this condition can
    be relaxed.)
\item The following covariance exists and is finite for all points $x=y$.
\begin{align}
&\bm{\mathcal{K}}\llbracket\nabla_{(x)}{\mathscr{J}(x)}
\otimes\nabla_{(y)}{\mathscr{J}(y)})\rrbracket
=\nabla_{(x)}\nabla_{(y)}\bm{\mathcal{K}}\llbracket
({\mathscr{J}(x)}\otimes{\mathscr{J}(y)})\rrbracket\nonumber\\&
=\lim_{\mathcal{R},\mathcal{S}\uparrow
0}\frac{1}{|\mathcal{RS}|}\bigg[\bm{\mathcal{K}}\llbracket
\mathscr{J}(x+\mathcal{R}\mu_{i})\otimes{\mathscr{J}(y+\mathcal{S}\mu_{i})}\rrbracket
-\bm{\mathcal{K}}\llbracket({\mathscr{J}(x)}\otimes
{\mathscr{J}(y+\mathcal{S}\mu_{i}))}\rrbracket\nonumber\\&
-\bm{\mathcal{E}}\llbracket({\mathscr{J}(x+\mathcal{R}\mu_{i})}\otimes{\mathscr{J}(y))}
+\bm{\mathcal{K}}\llbracket{\mathscr{J}(x)}\otimes\widehat{\mathscr{J})}
\rrbracket\nonumber\\&=\lim_{\mathcal{R},\mathcal{R}\uparrow
0}\frac{1}{|\mathcal{RS}|}\bigg[\bm{\mathcal{E}}\mathscr{J}(x+\mathcal{R}\mu_{i})
\otimes
\mathscr{J}(y+\mathcal{S}\mu_{i})\bigg\rrbracket-\bm{\mathcal{E}}({\mathscr{J}(x)}\otimes
{\mathscr{J}(y+\mathcal{S}\mu_{i})})\nonumber\\&-
\bm{\mathcal{E}}\llbracket{\mathscr{J}(x+\mathcal{R}\mu_{i})}\otimes
{\mathscr{J}(y))}\rrbracket
+\bm{\mathcal{E}}\llbracket{(\mathscr{J}(x)}\otimes{\mathscr{J}(y))}\rrbracket<\infty
\end{align}
\end{enumerate}
\end{lem}
\begin{rem}
Mean square differentiability at all $x\in\bm{\mathcal{D}}$ implies mean square continuity
at all $x\in\bm{\mathcal{D}}$ and this is also consistent with Holder and Kolmogorov
continuity conditions.
\begin{align}
\lim_{a\uparrow 0}\mathbf{E}\bigg\llbracket{\mathscr{J}(x+a\xi_{i})}-
{\mathscr{J}(x)|}^{2}\bigg|\bigg\rrbracket=\lim_{a\uparrow 0}(|\mathcal{R}|^{2})
\lim_{a\uparrow
0}\frac{1}{|\mathcal{R}|^{2}}\mathbf{E}\bigg\llbracket{\mathscr{J}(x+\mathcal{R}\xi_{i})}
-{\mathscr{J}(x)|}^{2}\bigg|\bigg\rrbracket=0
\end{align}
Note that this (and all definitions of the derivative) also requires a regulated random field or noise. Again, if the GRF is white-in-space noise then
$\mathbf{E}\llbracket{\mathscr{J}(x)}\otimes{\mathscr{J}(y)}\rrbracket=\delta^{n}(x-y)$
and the derivative cannot be defined.
\end{rem}
\begin{lem}
The condition for the existence of the derivative $\nabla_{i}{\mathscr{J}(x)}$
is also tantamount to the differentiability of the covariance function so that
\begin{align}
\mathbf{E}\llbracket\big[\nabla_{(x)}\mathscr{J}(x)\otimes\nabla_{y}{\mathscr{J}(y)}\rrbracket
=\nabla_{(x)}\nabla_{(y)}\mathbf{K}\llbracket{\mathscr{J}(x)}\otimes
{\mathscr{J}(y)}\rrbracket
\end{align}
exist for all $x=y$. Again, the GRF is differentiable only if it is not a white noise. Note
that for $x=y$
\begin{align}
\mathbf{K}\llbracket\big[\nabla_{(x)}\mathscr{J}(x)\otimes\nabla_{x}{\mathscr{J}(x)}
\rrbracket
=\nabla_{(x)}\nabla_{(x)}\mathbf{K}\llbracket{\mathscr{J}(x)}\otimes
{\mathscr{J}(x)}\rrbracket\nonumber\\&=\nabla_{(x)}\nabla_{(y)}\zeta J(x,y;\ell)=0
\end{align}
\end{lem}
\begin{lem}
The expectation is $\bm{\mathcal{E}}\llbracket\nabla{\mathscr{J}(x)}\rrbracket
=\nabla\bm{\mathcal{E}}\llbracket\mathscr{J}(x)\rrbracket=0$ since
\begin{align}
&\mathbf{E}\bigg\llbracket\lim_{\mathcal{R}\uparrow
0}\frac{{\mathscr{J}(x+|\mathcal{R}|\mu_{i})}-
{\mathscr{J}(x)}}{|\mathcal{R}|}\bigg\rrbracket\nonumber\\&=\lim_{\mathcal{R}\uparrow
0}\frac{1}{\mathcal{R}}\mathbf{E}\bigg\llbracket
\mathscr{J}(x+\mathcal{R}\mu_{i})\bigg\rrbracket-\mathbf{E}\bigg\llbracket{\mathscr{J}(x)}
\bigg\rrbracket=\nabla\mathbf{E}\bigg\llbracket{\mathscr{J}(x)}=0
\bigg\rrbracket
\end{align}
\end{lem}
\subsection{Stochastic Integration of GRSFs}
The integral of a GRF is defined as the limit of a Riemann sum of the field over the partition of a domain.
\begin{prop}
Let $\bm{\mathcal{D}}\subset\bm{\mathrm{R}}^{n}$ be a (closed) domain with boundary $\partial\bm{\mathcal{D}}$ and $x=(x_{1},...,x_{n})\subset\bm{\mathcal{D}}$. Let $\bm{\mathcal{D}}=\bigcup_{q=1}^{M}\bm{\mathcal{D}}_{1}$ be a partition of $\bm{\mathcal{D}}$ with $\bm{\mathcal{D}}_{q}\bigcap\bm{\mathcal{D}}_{q}=\varnothing$ if $p\ne q$. Let $x^{(q)}\in\bm{\mathcal{D}}_{q}$ for all $q=1...M$. Note $x^{(q)}\equiv (x_{1}^{(q)},...x_{n}^{(q)})$. Then $x^{(1)}\in\bm{\mathcal{D}}_{1},x^{(2)}\in\bm{\mathcal{D}}_{2},
...,x^{(M)}\in\bm{\mathcal{D}}$. Let $\mu(\bm{\mathcal{D}}_{q})$ be the volume of the partition $\bm{\mathcal{D}}_{q}$ so that $\mu(\bm{\mathcal{D}})=\sum_{q}^{M}\mu(\bm{\mathcal{D}})$. Similarly, if $\partial\bm{\mathcal{D}}$ is the surface or boundary of $\partial\bm{\mathcal{D}}$ the let $\bm{\mathcal{D}}=\bigcup_{q=1}^{M}\partial\bm{\mathcal{D}}_{q}$ be a partition of $\partial \bm{\mathcal{D}}$ with $\partial\bm{\mathcal{D}}_{q}\bigcap\partial \bm{\mathcal{D}}_{q}=\varnothing$ if $p\ne q$. Let $x^{(q)}\in\partial\bm{\mathcal{D}}_{q}$ for all $q=1...M$. Note $x^{(q)}\equiv (x_{1}^{(q)},...x_{n}^{(q)})$. Then $x^{(1)}\in\partial\bm{\mathcal{D}}_{1}, x^{(2)}\in\partial\bm{\mathcal{D}}_{2},...,x^{(M)} \in\partial\bm{\mathcal{D}}$. Let $\|\bm{\mathcal{D}}_{q}\|\equiv\mu(\partial\bm{\mathcal{D}}_{q})$ be the surface area of the partition $\partial\bm{\mathcal{D}}_{q}$ so that $\mu(\partial\bm{\mathcal{D}})=\bigcap_{q=1}^{M}\mu(\partial\bm{\mathcal{D}}_{q}$. The total volume and area of $\bm{\mathcal{D}}$ is
\begin{align}
\|\bm{\mathcal{D}}\|=\sum_{q=1}^{M}\|\bm{\mathcal{D}}_{q}\|=\sum_{q=1}^{M}\mu(\bm{\mathcal{D}}_{q})
\end{align}
\begin{align}
\|\partial\bm{\mathcal{D}}\|=\sum_{q=1}^{M}\|\partial \bm{\mathcal{D}}_{q}\|=\sum_{q=1}^{M}\mu(\partial\bm{\mathcal{D}}_{q})
\end{align}
Given the probability triplet $(\Omega,\mathcal{F},\bm{\mathsf{P}})$ then a Gaussian random field on $\bm{\mathcal{D}}$ for all $x\in\bm{\mathcal{D}}$ is $\mathscr{J}:\omega\times\bm{\mathcal{D}}\rightarrow\bm{\mathrm{R}}$ and ${\mathscr{J}(x^{q},\omega)}\in\bm{\mathcal{D}}_{q}$ exists for all $x^{(q)}\in\bm{\mathcal{D}}_{q}$ and $\omega\in\Omega$. The stochastic volume
integral and the stochastic surface integral are
\begin{align}
&\int_{\bm{\mathcal{D}}}{\mathscr{J}(x;\omega)}d^{n}x=\lim_{all~\mu(\bm{\mathcal{D}}_{q})\uparrow 0}\sum_{q=1}^{M}{\mathscr{J}(x^{(q)};\omega)}\mu(\bm{\mathcal{D}}_{q})\\&
\int_{\partial\bm{\mathcal{D}}}{\mathscr{J}(x;\omega)}d^{n-1}x
=\lim_{all~\mu(\partial\bm{\mathcal{D}}_{q})\uparrow 0}\sum_{q=1}^{M}{\mathscr{J}(x^{(q)};\omega)}\mu(\partial\bm{\mathcal{D}}_{q})
\end{align}

When a Gaussian random field is integrated, it is the limit of a linear combination of Gaussian random variables/fields so it is again Gaussian.
\end{prop}
Next, the stochastic expectations or averages are defined.
\begin{prop}
Since
\begin{equation}
\mathbf{E}\llbracket\bullet\rrbracket
=\int_{\Omega}(\bullet)d\bm{\mathsf{P}}(\omega)
\end{equation}
the expectations of the volume integral is as follows.
\begin{align}
\mathbf{E}\bigg\llbracket\int_{\bm{\mathcal{D}}}{\mathscr{J}(x;\omega)}d^{n}x\bigg\rrbracket
=\lim_{all~\mu(\bm{\mathcal{D}}_{q})\uparrow 0}{\mathbf{E}}\bigg\llbracket\sum_{q=1}^{M}
{\mathscr{J}(x^{(q)};\omega)}\mu(\bm{\mathcal{D}}_{q})\bigg\rrbracket
\end{align}
or
\begin{align}
&\mathbf{E}\bigg\llbracket\int_{\bm{\mathcal{D}}}{\mathscr{J}(x;\omega)}d^{n}x
\bigg\rrbracket\equiv\iint_{\bm{\mathcal{D}}}{\mathscr{J}(x;\omega)}d^{n}x d\bm{P}(\omega)\\&=\lim_{all~\mu(\bm{\mathcal{D}}_{q})\uparrow 0}\int_{\Omega}\sum_{q=1}^{M}
\mathscr{J}(x^{(q)};\omega))\mu(\bm{\mathcal{D}}_{q})d\bm{P}(\omega)=0
\end{align}
which vanishes for GRFs since $\mathbf{E}\bigg\llbracket
\mathscr{J}(x^{(q)})\bigg\rrbracket=0$. Similarly, for the stochastic surface integrals
\begin{align}
\mathbf{E}\left\llbracket\int_{\bm{\mathcal{D}}}{\mathscr{J}
(x;\omega)}d^{n-1}x\right\rrbracket
=\lim_{all~\mu(\bm{\mathcal{D}}_{q})\uparrow 0}\mathbf{E}\left\llbracket\sum_{q=1}^{M}
{\mathscr{J}(x^{(q)};\omega)}\mu(\bm{\mathcal{D}}_{q})\right\rrbracket
\end{align}
or
\begin{align}
&\mathbf{E}\left\llbracket\int_{\partial\bm{\mathcal{D}}}{\mathscr{J}(x;\omega)}d^{n-1}x
\right\rrbracket\equiv\int_{\Omega}\int_{\partial\bm{\mathcal{D}}}
{\mathscr{J}(x;\omega)}d^{n-1}x d\bm{\mathsf{P}}(\omega)\\&=\lim_{all~\mu(\partial\bm{\mathcal{D}}_{q})\uparrow 0}\int_{\Omega}\sum_{q=1}^{M}{\mathscr{J}(x^{(q)};\omega)}
\mu(\partial\bm{\mathcal{D}}_{q})d\bm{\mathsf{P}}(\omega)=0
\end{align}
\end{prop}
Given an integral (or summation) over a random field or stochastic process, the Fubini theorem states that the expectation of the integral or sum over a random field is equivalent to the integral or sum of the expectation of the field.
\begin{thm}
Let $\mathscr{J}(x)$ be a random field, not necessarily Gaussian, existing for all $x\in\bm{\mathcal{D}}$ with expectation $\bm{\mathbf{E}}\llbracket \mathscr{J}(x)\rrbracket $, not necessarily zero. Then
\begin{equation}
\mathbf{E}\bigg\llbracket\int_{\bm{\mathcal{D}}}\mathscr{J}(x)d\mu(x)
\bigg\rrbracket\equiv \int_{\bm{\mathcal{D}}}\mathbf{E}\big\llbracket
\mathscr{J}(x)\big\rrbracket d\mu(x)
\end{equation}
For a set of N random fields $\mathscr{J}_{q}(x)$
\begin{equation}
\mathbf{E}\bigg\llbracket\sum_{q=1}^{N}\mathscr{J}_{q}(x)\bigg\rrbracket=
\sum_{q=1}^{N}\mathbf{E}\big\llbracket\mathscr{J}_{q}(x)\big\rrbracket
\end{equation}
\end{thm}
It is also possible to define a 'mollifier' or convolution integral.
\begin{prop}
Let $(x,y)\in\bm{\mathcal{D}}$ and let $\Psi(x,y)$ be a smooth function of $(x,y)$ that will typically depend on the separation $\|x-y\|$. Given $\mathscr{J}(y)$ define the volume and surface integral convolutions
\begin{align}
&{\mathscr{F}(x)}=\Psi(x,y)\boxtimes {\mathscr{J}(y)}\equiv\int_{\bm{\mathcal{D}}}\Psi(x,y)
{\mathscr{J}(y)}d^{n}y,~(x,y)\in\bm{\mathcal{D}}\\& {\mathscr{F}(x)}=\Psi(x,y)\boxtimes {\mathscr{J}(y)}\equiv\int_{\bm{\mathcal{D}}}\Psi(x,y)
{\mathscr{J}(y)}d^{n-1}y,(x,y)\in\partial\bm{\mathcal{D}}
\end{align}
\end{prop}
For example, if $\Psi(x,y,L)$ is a Gaussian function of width $L$ then the random field
$\mathscr{J}(x,L)$ can be 'smoothed' at the scale $L$ so that
\begin{align}
&\mathscr{J}(x,L)=\Psi(x,y,L)\boxtimes {\mathscr{J}(y)}\equiv C\int_{\bm{\mathcal{D}}}\exp(-|x-y|^{2}/L^{2}){\mathscr{J}(y)}d^{n}y,~(x,y)
\in\bm{\mathcal{D}}\\& \mathscr{J}(x,L)= \Psi(x,y)\boxtimes {\mathscr{J}(y)}\equiv C\int_{\bm{\mathcal{D}}}\exp(-|x-y|^{2}/L^{2}){\mathscr{J}(y)}d^{n-1}y,(x,y)
\in\partial\bm{\mathcal{D}}
\end{align}
If $K(x,y)=1/|x-y|^{b}$, for $b\ge 1$ then
\begin{align}
&\mathscr{K}(x)=\Psi(x,y)\boxtimes {\mathscr{J}(y)}\equiv C\int_{\bm{\mathcal{D}}}\frac{{\mathscr{J}(y)}}{|x-y|^{b}}d^{n}y,~(x,y)
\in\bm{\mathcal{D}}\\& \mathscr{K}(x)=\Psi(x,y)\boxtimes{\mathscr{J}(y)}\equiv C\int_{\bm{\mathcal{D}}}\frac{{\mathscr{J}(y)}}{|x-y|^{b}}d^{n-1}y,(x,y)\in\partial
\bm{\mathcal{D}}
\end{align}
\subsubsection{2-point functions and correlations}
Integrals of random fields can be correlated.
\begin{prop}
Let $\bm{\mathcal{D}}\subset\bm{\mathrm{R}}^{n}$ and $\bm{\mathcal{D}}^{\prime}\subset\bm{\mathrm{R}}^{n}$ be two domains with boundaries $\partial\bm{\mathcal{D}}$ and $\partial\bm{\mathcal{D}}^{\prime}$. One can have $\bm{\mathcal{D}}\bigcap\bm{\mathcal{D}}^{\prime}=\varnothing$ or $\bm{\mathcal{D}}\bigcap\bm{\mathcal{D}}^{\prime}\ne\varnothing$.
Let $x\in\bm{\mathcal{D}}$ and $x^{\prime}\in\bm{\mathcal{D}}^{\prime}$. If $\bm{\mathcal{D}}\bigcap\bm{\mathcal{D}}^{\prime}\ne\varnothing$ then one can have $(x,x^{\prime})\in \bm{\mathcal{D}}\bigcap\bm{\mathcal{D}}^{\prime}$. As before, partition the domains as
\begin{equation}
\bm{\mathcal{D}}=\bigcup_{q=1}^{M}\bm{\mathcal{D}}_{q},~~\bm{\mathcal{D}}^{\prime}
=\bigcup_{p=1}^{M}\bm{\mathcal{D}}^{\prime}_{p}\nonumber
\end{equation}
with $x^{(q)}\in\bm{\mathcal{D}}_{q}$ and $(x^{(p)})^{\prime}\in\bm{\mathcal{D}}^{\prime}_{p}$. The volumes of the cells are $\mu(\bm{\mathcal{D}}_{q})$ and $\mu(\bm{\mathcal{D}}_{p})$. Define the GRSFs as the stochastic integrals
\begin{align}
&\mathscr{H}(\bm{\mathcal{D}})=\int_{\bm{\mathcal{D}}}{\mathscr{J}(x)}d^{n}x
=\lim_{all~\mu(\bm{\mathcal{D}}_{q})\uparrow 0}\sum_{q=1}^{M}{\mathscr{J}(x_{q}})\mu(\bm{\mathcal{D}}_{q})\\&
\mathscr{H}(x^{\prime})=\int_{\bm{\mathcal{D}}}
{\mathscr{J}(x)}d^{n}x=\lim_{all~\mu(\bm{\mathcal{D}}_{p})\uparrow 0}\sum_{p=1}^{M}{\mathscr{J}(x_{p})}\mu(\bm{\mathcal{D}}_{p})
\end{align}
Then
\begin{align}
&\mathscr{H}(\bm{\mathcal{D}})\otimes {\mathscr{H}(\bm{\mathcal{D}}^{\prime})}=\int_{\bm{\mathcal{D}}}\int_{\bm{\mathcal{D}}^{\prime}}
{\mathscr{J}(x)}\otimes {\mathscr{J}(x^{\prime})}d^{n}x d^{n}x^{\prime}
\\&=\lim_{all~\mu(\bm{\mathcal{D}}_{q},\mu(\bm{\mathcal{D}}^{\prime}_{p})\uparrow 0}\sum_{q=1}^{M}\sum_{p=1}^{M}{\mathscr{J}(x_{p})}\otimes
{\mathscr{J}(x_{q})}\mu(\bm{\mathcal{D}}_{p})\mu(\bm{\mathcal{D}}^{\prime}_{q})
\end{align}
The expectation is then
\begin{align}
&\mathbf{E}\bigg\llbracket\mathscr{H}(\bm{\mathcal{D}})\otimes \mathscr{H}(\bm{\mathcal{D}}^{\prime})\bigg\rrbracket=\left\llbracket\int_{\bm{\mathcal{D}}}\int_{\bm{\mathcal{D}}^{\prime}}
{\mathscr{J}(x)}\otimes\mathscr{J}(x^{\prime}) d^{n}xd^{n}x^{\prime}\right\rrbracket
\\&=\lim_{all~\mu(\bm{\mathcal{D}}_{q},\mu(\bm{\mathcal{D}}^{\prime}_{p})\uparrow 0}\mathbf{E}\left\llbracket\sum_{p=1}^{M}\sum_{q=1}^{M}{\mathscr{J}(x_{p})}\otimes
{\mathscr{J}(x_{q})}\mu(\bm{\mathcal{D}}_{p})\mu(\bm{\mathcal{D}}^{\prime}_{q})\right
\rrbracket\\&\equiv=\lim_{all~\mu(\bm{\mathcal{D}}_{q},\mu(\bm{\mathcal{D}}^{\prime}_{p})\uparrow 0}\sum_{p=1}^{M}\sum_{q=1}^{M}\mathbf{E}\bigg\llbracket{\mathscr{J}(x_{p})}\otimes{\mathscr{J}(x_{q})}\bigg\rrbracket\mu(\bm{\mathcal{D}}_{p})\mu(\bm{\mathcal{D}}^{\prime}_{q})\\&
\lim_{all~\mu(\bm{\mathcal{D}}_{q},\mu(\bm{\mathcal{D}}^{\prime}_{p})\uparrow 0}\sum_{p=1}^{M}\sum_{q=1}^{M}\alpha(x_{p},x_{q};\xi)\mu(\bm{\mathcal{D}}_{p})\mu(\bm{\mathcal{D}}^{\prime}_{q})\\&
=\int_{\bm{\mathcal{D}}}\int_{\bm{\mathcal{D}}^{\prime}}\alpha(x,x^{\prime};\xi)d^{n}x d^{n}x^{\prime}
\end{align}
\end{prop}
Similarly, for the surface integrals
\begin{align}
&\mathbf{E}\bigg\llbracket\mathscr{H}(\partial\bm{\mathcal{D}})\otimes \mathscr{H}(\partial\bm{\mathcal{D}}^{\prime})\bigg\rrbracket\left\llbracket\int_{\partial\bm{\mathcal{D}}}\int_{\partial\bm{\mathcal{D}}^{\prime}}{\mathscr{J}}\otimes{\mathscr{J}(x^{\prime})}d^{n-1}xd^{n-1}x^{\prime}
\right\rrbracket\\&=\lim_{all~\mu(\partial\bm{\mathcal{D}}_{q},\mu(\partial\bm{\mathcal{D}}^{\prime}_{p})\uparrow 0}\mathbf{E}\left\llbracket\sum_{p=1}^{M}\sum_{q=1}^{M}{\mathscr{J}(x_{p})}\otimes
\mathscr{J}(x_{q})\mu(\partial\bm{\mathcal{D}}_{p})\mu(\bm{\mathcal{D}}^{\prime}_{q})\right\rrbracket\\&
\equiv=\lim_{all~\mu(\partial\bm{\mathcal{D}}_{q},\mu(\partial\bm{\mathcal{D}}^{\prime}_{p})\uparrow 0}\sum_{p=1}^{M}\sum_{q=1}^{M}\mathbf{E}\bigg\llbracket{\mathscr{J}(x_{p})}\otimes
{\mathscr{J}(x_{q})}\bigg\rrbracket\mu(\partial\bm{\mathcal{D}}_{p})\mu(\partial\bm{\mathcal{D}}^{\prime}_{q})\\&
\lim_{all~\mu(\partial\bm{\mathcal{D}}_{q},\mu(\partial\bm{\mathcal{D}}^{\prime}_{p})\uparrow 0}\sum_{p=1}^{M}\sum_{q=1}^{M}\alpha(x_{p},x_{q};\xi)\mu(\partial\bm{\mathcal{D}}_{p})\mu(\partial\bm{\mathcal{D}}^{\prime}_{q})\\&
=\int_{\partial\bm{\mathcal{D}}}\int_{\partial\bm{\mathcal{D}}^{\prime}}\alpha(x,x^{\prime};\xi)d^{n-1}x d^{n-1}x^{\prime},~~(x,x^{\prime}\in\partial\bm{\mathcal{D}},\partial\bm{\mathcal{D}}^{\prime}
\end{align}
\begin{cor}
The volatility is
\begin{align}
&\lim_{x\rightarrow x^{\prime}}\mathbf{E}\bigg\llbracket\mathscr{H}(\bm{\mathcal{D}})\otimes \mathscr{H}(\bm{\mathcal{D}})\bigg\rrbracket=\mathbf{E}\bigg\llbracket\mathscr{Y}(x)|^{2}\bigg\rrbracket
\\&\int_{\bm{\mathcal{D}}}\int_{\bm{\mathcal{D}}^{\prime}}U d^{n}xd^{n}x^{\prime}=U\|\bm{\mathcal{D}}\|\|\bm{\mathcal{D}}^{\prime}\|
\end{align}
\end{cor}
A stochastically averaged Dirichlet energy for the GRSF can be defined as follows
\begin{prop}
Let ${\mathscr{J}(x)}$ be a GRSF existing on a domain $\bm{\mathcal{D}}
\subset\bm{\mathrm{R}}^{n}$. The derivative of the GRSF exists and is the GRVF $\overline{u_{i}(x)}=\nabla_{i}{\mathscr{J}(x)}$. The covariance is
\begin{equation}
\mathbf{E}\big\llbracket \overline{u_{i}(x)}\otimes\overline{u_{j}(y)}\bigg\rrbracket=
\big\llbracket\nabla_{i}{\mathscr{J}(x)}\otimes\nabla_{j}
{\mathscr{J}(x(y)}\rrbracket
=\delta_{ij}\beta(x,y;\zeta)
\end{equation}
and $\mathbf{E}\big\llbracket\nabla_{i}{\mathscr{J}}(x)\otimes\nabla_{j}
{\mathscr{F}(y)}\rrbracket=\delta_{ij}$. The expectation of the Dirichlet energy integral is then
\begin{align}
&\mathbf{E}\left\llbracket\int_{\bm{\mathcal{D}}}\nabla_{i}
{\mathscr{J}(x)}\otimes\nabla_{j}
{\mathscr{J}(x)}d^{n}x\right\rrbracket=
\sum_{ij}\int_{\bm{\mathcal{D}}}\mathbf{E}\bigg\llbracket\nabla_{i}
{\mathscr{J}(x)}\otimes\nabla_{j}
{\mathscr{J}(x)}\bigg\rrbracket d^{n}x\nonumber\\&=\sum_{ij}\delta_{ij}\int_{\bm{\mathcal{D}}}d^{n}x=n\|\bm{\mathcal{D}}\|
<\infty
\end{align}
where $\|\bm{\mathcal{D}}\|$ is the volume of $\bm{\mathcal{D}}$.
\end{prop}
\begin{lem}
Let $f:\bm{\mathrm{R}}\rightarrow\bm{\mathrm{R}}$ be a smooth function then $\exists~C>0 $ such that
\begin{equation}
\bigg(\int_{0}^{b}f(x)dx\bigg)\bigg(\int_{0}^{b}f(x)dx
\bigg)\le C\int_{0}^{b}\bigg|\int_{0}^{b}|f(x)|^{2}dx\bigg|dx
\end{equation}
if $|\int_{0}^{b}f(x)dx|<\infty$. This generalises to
\begin{equation}
\left(\int_{0}^{b}f(x)dx\right)\underbrace{\times...\times}_{Q~ times}\left(\int_{0}^{b}f(x)dx\right) \le C \underbrace{\int_{o}^{b}...\int_{o}^{b}}_{Q~times}|f(x)|^{Q}\underbrace{dx...dx}_{Q~times}
\end{equation}
For $\bm{\mathcal{D}} \subset\bm{\mathrm{R}}^{n}$,
$f:\bm{\mathcal{D}}\rightarrow\bm{\mathrm{R}}$ and a measure $d\mu(x)$,
\begin{align}
&\left|\left(\int_{\bm{\mathcal{D}}}f(x)d\mu(x)\right)\right|^{Q}=
\left(\int_{\bm{\mathcal{D}}}f(x)d\mu(x)\right)\underbrace{\times...\times}_{Q~ times}\left(\int_{\bm{\mathcal{D}}}f(x)d\mu(x)\right)\nonumber\\&\le C\underbrace{\int_{\bm{\mathcal{D}}}...\int_{\bm{\mathcal{D}}}}_{Q-1~times}\bigg|\int_{\bm{\mathcal{D}}}|f(x)|^{Q}
d\mu(x)\bigg|\underbrace{d\mu(x)...d\mu(x)}_{Q-1~times}\nonumber\\&
\le C\underbrace{\int_{\bm{\mathcal{D}}}...\int_{\bm{\mathcal{D}}}}_{Q-1~times}\bigg\|f(x)
\bigg\|_{L_{Q}(\bm{\mathcal{D}})}^{Q}\underbrace{d\mu(x)...d\mu(x)}_{Q-1~times}
\end{align}
\end{lem}
The result is extended to include GRSFs ${\mathscr{J}(x)}$.
\begin{prop}
Let $\mathscr{J}$ be a regulated GRSF defined for all $x\in\bm{\mathcal{D}}$ such that $\mathbf{E}\llbracket{\mathscr{J}(x)}\rrbracket=0$ and
\begin{equation}
\mathbf{E}\bigg\llbracket\big|{\mathscr{J}(x(x)}\big|^{Q}\bigg\rrbracket=
\frac{1}{2}[\alpha^{Q/2}+(-1)^{Q}\alpha^{Q/2}]
\end{equation}
so that $\mathbf{E}\big\llbracket\big|{\mathscr{J}(x(x)}\big|^{Q}\big\rrbracket=0$ for all odd p. Then
\begin{align}
\mathbf{E}\bigg\llbracket
\bigg|\int_{\bm{\mathcal{D}}}{\mathscr{J}(x)}d\mu(x)\bigg|^{Q}\bigg\rrbracket&~\le
\frac{1}{2}C[\alpha^{Q/2}+(-1)^{Q}\alpha^{Q/2}]\bigg|\int_{\bm{\mathcal{D}}}d\mu(x)\bigg|^{Q}
\nonumber\\&=\frac{1}{2}[\alpha^{Q/2}+(-1)^{Q}\alpha^{Q/2}]\|\bm{\mathcal{D}}\|^{Q}
\end{align}
where $\|\bm{\mathcal{D}}\|$ is the volume of $\bm{\mathcal{D}}$.
\end{prop}
\begin{proof}
Expand the lhs so that
\begin{align}
&\mathbf{E}\bigg\llbracket\bigg|\int_{\bm{\mathcal{D}}}
{\mathscr{J}(x(x)}d\mu(x)
\bigg|^{Q}\bigg\rrbracket=
\mathbf{E}\bigg\llbracket
\bigg|\int_{\bm{\mathcal{D}}}{\mathscr{J}(x(x)}d\mu(x)\times...
\times\int_{\bm{\mathcal{D}}}{\mathscr{J}(x(x)}d\mu(x)\bigg\rrbracket
\nonumber\\&\mathbf{E}\bigg\llbracket
\int_{\bm{\mathcal{D}}}\times...\times\int_{\bm{\mathcal{D}}}
|{\mathscr{J}(x(x)}|^{p}d\mu(x)\times...\times d\mu(x)\bigg\rrbracket\nonumber\\&
\le \int_{\bm{\mathcal{D}}}\times...\times\int_{\bm{\mathcal{D}}}
\mathbf{E}\bigg\llbracket|{\mathscr{J}(x(x)}|^{p}\bigg\rrbracket
d\mu(x)\times...\times d\mu(x)\nonumber\\&=\frac{1}{2}[\alpha^{Q/2}+(-1)^{Q}\alpha^{Q/2}]
\int_{\bm{\mathcal{D}}}d\mu{x}\times...\times\int_{\bm{\mathcal{D}}}d\mu{x}\nonumber\\&
=\frac{1}{2}[\alpha^{Q/2}+(-1)^{Q}\alpha^{Q/2}]\bigg|\int_{\bm{\mathcal{D}}}d\mu(x)\bigg|^{Q}
\nonumber\\&=\frac{1}{2}[\alpha^{Q/2}+(-1)^{Q}\alpha^{Q/2}]\|\bm{\mathcal{D}}\|^{Q}
\end{align}
\end{proof}
\begin{prop}
Let ${\mathscr{J}(x(x)}$ be a regulated GRSF defined for all $x\in\bm{\mathcal{D}}$ such that $\mathbf{E}\llbracket{\mathscr{J}(x(x)}\rrbracket=0$ and $\mathbf{E}\big\llbracket
\big|{\mathscr{J}(x(x)}\big|^{Q}\big\rrbracket=\frac{1}{2}[\alpha^{Q/2}+(-1)^{Q}\alpha^{Q/2}]$ as before. Let $\mathcal{O}(x,y)=\mathcal{O}(x-y)$ be a kernel or 'mollifier' that smooths the random field ${\mathscr{J}(x(x)}$ such that
\begin{equation}
\overline{g(x)}=\int_{\bm{\mathcal{D}}}\mathcal{O}(x,y){\mathscr{J}(x(y)}d\mu(y)
\end{equation}
For example if $\mathcal{Q}(x,y)=c\exp(-\|x-y\|^{2}/b)$ then the field is 'Gaussian smoothed'. To order $b$
\begin{align}
\mathbf{E}\bigg\llbracket
\bigg|\int_{\bm{\mathcal{D}}}\mathcal{O}(x,y){\mathscr{J}(x(y)}d\mu(x)\bigg|^{Q}
\bigg\rrbracket&~\le\frac{1}{2}[\alpha^{Q/2}+(-1)^{Q}\alpha^{Q/2}]\bigg|
\int_{\bm{\mathcal{D}}}\mathrm{Q}(x,y)d\mu(x)\bigg|^{Q}
\end{align}
\end{prop}
\subsection{Real and complex superpositions of GRFs}
Let $\mathscr{U}(x),\mathscr{V}(x),\mathscr{W}(x),\mathscr{X}(x)$ be independent GRFs existing for all $x\in\bm{\mathcal{D}}\subset\bm{\mathrm{R}}^{n}$. The fields are Gaussian
and regulated so that
\begin{align}
&\mathbf{E}\big\llbracket\mathscr{U}(x)\otimes\mathscr{U}(y)\rrbracket=\alpha A(x,y;\xi)\\&
\mathbf{E}\big\llbracket\mathscr{V}(x)\otimes\mathscr{V}(y)\rrbracket=\beta B(x,y;\xi)\\&
\mathbf{E}\big\llbracket\mathscr{W}(x)\otimes\mathscr{W}(y)\rrbracket=\gamma C(x,y;\xi)\\&
\mathbf{E}\big\llbracket\mathscr{X}(x)\otimes\mathscr{X}(y)\rrbracket=\epsilon D(x,y;\xi)
\end{align}
with regulated and bounded volatilities
\begin{align}
&\mathbf{E}\big\llbracket\mathscr{U}(x)\otimes\mathscr{U}(x)\rrbracket=\alpha\\&
\mathbf{E}\big\llbracket\mathscr{V}(x)\otimes\mathscr{V}(x)\rrbracket=\beta\\&
\mathbf{E}\big\llbracket\mathscr{W}(x)\otimes\mathscr{W}(x)\rrbracket=\gamma\\&
\mathbf{E}\big\llbracket\mathscr{X}(x)\otimes\mathscr{X}(x)\rrbracket=\epsilon
\end{align}
Since the fields are considered independent
\begin{align}
&\mathbf{E}\big\llbracket\mathscr{U}(x)\otimes\mathscr{V}(y)\rrbracket=0\\&
\mathbf{E}\big\llbracket\mathscr{U}(x)\otimes\mathscr{W}(y)\rrbracket=0\\&
\mathbf{E}\big\llbracket\mathscr{U}(x)\otimes\mathscr{X}(y)\rrbracket=0\\&
\mathbf{E}\big\llbracket\mathscr{U}(x)\otimes\mathscr{V}(x)\rrbracket=0\\&
\mathbf{E}\big\llbracket\mathscr{U}(x)\otimes\mathscr{W}(x)\rrbracket=0\\&
\mathbf{E}\big\llbracket\mathscr{U}(x)\otimes\mathscr{X}(x)\rrbracket=0
\end{align}
and so on.
\begin{prop}
Let $(a,b,c,d)\in\bm{\mathrm{R}}^{+}$ then a linear superposition of GRFs is also a GRF so that
\begin{align}
&\mathscr{F}(x=a\mathds{U}(x)+b\mathds{V}(x)\\&
\mathscr{F}(x=a\mathds{U}(y)+b\mathds{V}(y)
\end{align}
Then $\mathbf{E}\llbracket\mathscr{F}(x(x)\rrbracket=0$ and
\begin{align}
&\mathbf{E}\big\llbracket\mathscr{F}(x)\otimes \mathscr{F}(y)\big\rrbracket
=a^{2}A(x,y;\xi)+b^{2}B(x,y;\xi)\\&
\mathbf{E}\big\llbracket\mathscr{F}(x)\otimes
\mathscr{F}(x)\big\rrbracket=a^{2}+b^{2}
\end{align}
Similarly, if
\begin{align}
&\mathscr{F}(x)=a\mathscr{U}(x)+b\mathscr{V}(x)+c\mathscr{W}(x)\\&
\mathscr{F}(y)=a\mathscr{U}(y)+b\mathscr{V}(y)+C\mathscr{W}(y)
\end{align}
then
\begin{align}
&\mathbf{E}\big\llbracket \mathscr{F}(x)\otimes\mathscr{F}(y)\big\rrbracket
=a^{2}A(x,y;\xi)+b^{2}B(x,y;\xi)+c^{2}C(x,y;\xi)\\&
\mathbf{E}\big\llbracket\mathscr{F}(x)\otimes \mathscr{F}(x)\big\rrbracket=a^{2}+b^{2}+c^{2}
\end{align}
and
\begin{align}
&\mathscr{F}(x)=a\mathscr{U}(x)+b\mathscr{V}(x)+c\mathscr{W}(x)+d\mathscr{X}(x)\\&
\mathscr{F}(y)=a\mathscr{U}(y)+b\mathscr{V}(y)+C\mathscr{W}(y)+d\mathscr{X}(y)
\end{align}
then
\begin{align}
&\mathbf{E}\big\llbracket\mathscr{F}(x)\otimes \mathscr{F}(y)\big\rrbracket=a^{2}A(x,y;\xi)+b^{2}B(x,y;\xi)+c^{2}C(x,y;\xi)
+d^{2}D(x,y;\xi)\\&
\mathbf{E}\big\llbracket\mathscr{F}(x)\otimes
\mathscr{F}(y)\big\rrbracket=a^{2}+b^{2}
+c^{2}+d^{2}
\end{align}
\end{prop}
\begin{proof}
\begin{align}
&\mathbf{E}\llbracket\mathscr{F}(x)\otimes
\mathscr{F}(y)\rrbracket
=a^{2}\mathbf{E}\llbracket\mathscr{U}(x)\mathscr{U}(y)\rrbracket
+ab\mathbf{E}\llbracket\mathscr{U}(x)\mathscr{V}(y)\rrbracket\nonumber\\&
+ab\mathbf{E}\llbracket\mathscr{V}(x)\mathscr{U}(y)\rrbracket+b^{2}\mathbf{E}
\llbracket\mathscr{V}(x)\mathscr{V}(y)\rrbracket\nonumber\\&
=a^{2}A(x,y;\xi)+b^{2}B(x,y;\xi)
\end{align}
so that the volatility is
\begin{align}
&\mathbf{E}\llbracket\mathscr{F}(x)\otimes
\mathscr{F}(y)\rrbracket=a^{2}+b^{2}
\end{align}
\begin{align}
&\mathbf{E}\llbracket\mathscr{F}(x)\otimes\mathscr{F}(y)\rrbracket
=a^{2}\mathbf{E}\llbracket\mathscr{U}(x)\mathscr{U}(y)\rrbracket+ab\mathbf{E}
\llbracket\mathscr{U}(x)\mathscr{V}(y)\rrbracket+ac\mathbf{E}
\llbracket\mathscr{U}(x)\mathscr{W}(y)\rrbracket\nonumber\\&
+ab\mathbf{E}\llbracket\mathscr{V}(x)\mathscr{U}(y)\rrbracket+b^{2}\mathbf{E}
\llbracket\mathscr{V}(x)\mathscr{V}(y)\rrbracket
+cb\mathbf{E}\llbracket\mathscr{V}(x)\mathscr{W}(y)\rrbracket\nonumber\\&
+ac\mathbf{E}\llbracket\mathscr{W}(x)\mathscr{U}(y)\rrbracket+cb\mathbf{E}\llbracket
\mathscr{W}(x)\mathscr{V}(y)\rrbracket+c^{2}\mathbf{E}\llbracket\mathscr{W}(x)
\mathscr{W}(y)\rrbracket\nonumber\\&
=a^{2}A(x,y;\xi)+b^{2}B(x,y;\xi)+c^{2}C(x,y;\xi)
\end{align}
and so the volatility is bounded as
\begin{equation}
\mathbf{E}\llbracket\mathscr{F}(x)\otimes\mathscr{F}(x)\rrbracket
=|a^{2}+b^{2}+c^{2}|<\infty
\end{equation}
and similarly for (-).
\end{proof}
\subsection{Complex Gaussian random fields}
Given real random Gaussian fields it is possible to construct a complex Gaussian
random field.
\begin{prop}
Let $\mathscr{X}(x),\mathscr{Y}(x)$ be independent GRFs for all $x\in\bm{\mathcal{D}}\subset\bm{\mathrm{R}}^{n}$ with regulated
covariances
\begin{align}
&\mathbf{E}\llbracket\mathscr{X}(x)\otimes\mathscr{X}(y)\rrbracket=\alpha A(x,y;\xi)\\&
\mathbf{E}\llbracket\mathscr{Y}(x)\otimes\mathscr{Y}(y)\rrbracket=\beta B(x,y;\xi)
\end{align}
and $\mathbf{E}\llbracket\mathscr{X}(x)\otimes\mathscr{X}(x)\rrbracket=\alpha$ and
$\mathbf{E}\llbracket\mathscr{Y}(x)\otimes\mathscr{Y}(x)\rrbracket=\beta$. For all $(x,y)\in \bm{\mathcal{D}}$ define the following complex fields and their derivatives
\begin{align}
&\mathscr{Z}(x)=\mathscr{X}(x)+i\mathscr{Y}(x)\\&
\mathscr{Z}^{*}(x)=\mathscr{X}(x)-i\mathscr{Y}(x)\\&
\mathscr{Z}(y)=\mathscr{X}(y)+i\mathscr{Y}(y)\\&
\mathscr{Z}^{*}(y)=\mathscr{X}(y)-i\mathscr{Y}(y)
\end{align}
\begin{align}
&\nabla_{j}\mathscr{Z}(x)=\nabla_{j}\mathscr{X}(x)+i\nabla_{j}\mathscr{Y}(x)\\&
\nabla_{j}\mathscr{Z}^{*}(x)=\nabla_{j}\mathscr{X}(x)-i\nabla_{j}\mathscr{Y}(x)\\&
\nabla_{j}\mathscr{Z}(y)=\nabla_{j}\mathscr{X}(y)+i\mathscr{Y}(y)\\&
\nabla_{j}\mathscr{Z}^{*}(y)=\nabla_{j}\mathscr{X}(y)-i\nabla_{j}\mathscr{Y}(y)
\end{align}
Then the covariances are real
\begin{align}
&\mathbf{E}\llbracket\mathscr{Z}(x)\otimes\mathscr{Z}(y)\rrbracket=\alpha A(x,y;\xi)-\beta B(x,y;\xi)\\&
\mathbf{E}\llbracket\mathscr{Z}(x)\otimes\mathscr{Z}^{*}(y)\rrbracket=\alpha A(x,y;\xi)+\beta B(x,y;\xi)\\&
\mathbf{E}\llbracket \mathscr{Z}^{*}(x)\otimes\mathscr{Z}(y)\rrbracket=
\alpha A(x,y;\xi)+\beta B(x,y;\xi)\\&
\mathbf{E}\llbracket \mathscr{Z}^{*}(x)\otimes\mathscr{Z}^{*}(y)\rrbracket=\alpha A(x,y;\xi)-\beta B(x,y;\xi)
\end{align}
with the finite bounded volatilities
\begin{align}
&\mathbf{E}\llbracket \mathscr{Z}(x)\otimes \mathscr{Z}(x)\rrbracket=\alpha-\beta\\&
\mathbf{E}\llbracket \mathscr{Z}(x)\otimes \mathscr{Z}(x)\rrbracket=\alpha-\beta
\end{align}
\end{prop}
\begin{proof}
\begin{align}
&\mathbf{E}\llbracket\mathscr{Z}(x)\otimes\mathscr{Z}(x)\rrbracket=
\mathbf{E}\llbracket(\mathscr{X}(x)+i\mathscr{Y}(x))(\mathscr{X}(y)+i\mathscr{Y}(y))
\rrbracket\nonumber\\&
=\mathbf{E}\llbracket \mathscr{X}(x)\mathscr{X}(y)\rrbracket+i
\mathbf{E}\llbracket\mathscr{X}(x)\mathscr{Y}(y)\rrbracket +i\mathbf{E}\llbracket
\mathscr{Y}(x)\mathscr{X}(y)\rrbracket+
\mathbf{E}\llbracket\mathscr{Y}(x)\mathscr{Y}(y)\rrbracket\nonumber\\&
=\mathbf{E}\llbracket\mathscr{X}(x)\mathscr{X}(y)\rrbracket\mathbf{E}+
\mathbf{E}\llbracket\mathscr{Y}(x)\mathscr{Y}(y)\rrbracket
=\alpha A(x,y;\xi)+\beta B(x,y;\xi)
\end{align}
and similarly to prove (1.19-1.21).
\end{proof}
\section{2-Point Correlations and Covariances of Gaussian Random Scalar Fields and Their Derivatives on a Ball, Cylinder and Disc}
\raggedbottom
In this appendix, the covariances or 2-point functions are explicitly computed for GRSFs on various geometries, namely the 3-sphere, cylinder and a disc.
\begin{prop}
Let $\bm{\mathcal{B}}_{R}(0)\subset\bm{\mathrm{R}}^{3}$ be a ball of radius $R$ and let ${\mathscr{J}(x)}$ be a GRSF existing
within $\bm{\mathcal{B}}_{R}(0)$ and also on the boundary $\partial\bm{\mathcal{B}}_{R}(0)$. In spherical coordinates
${\mathscr{J}(x)}={\mathscr{J}(r,\theta,\varphi)}$. Since the field is Gaussian it is characterised by its 2-point correlation function, which is taken to be regulated.(Colored noise.) Then for any points $(r,\theta,\varphi)$ and
$(r^{\prime},\theta^{\prime},\varphi^{\prime})$
\begin{align}
&\mathbf{E}\big\llbracket
{\mathscr{J}(r,\theta,\varphi)}\big\rrbracket=0\\&
\mathbf{E}\big\llbracket
{\mathscr{J}(r,\theta,\varphi)}\otimes {\mathscr{J}(r^{\prime},\theta^{\prime},\varphi^{\prime})}\big\rrbracket=
\lambda K(r,r^{\prime})H(\theta,\theta^{\prime};\eta)Q(\varphi,\varphi^{\prime};\chi)
\end{align}
such that $F(r,r)=H(\theta,\theta)=Q(\varphi,\varphi)=1$. The derivatives
\begin{equation}
\nabla_{r}F(r,r^{\prime};\epsilon),\nabla_{r^{\prime}}
F(r,r^{\prime};\epsilon),\nabla_{\theta}H(\theta,\theta^{\prime};\eta)
\nabla_{\theta^{\prime}}H(\theta,\theta^{\prime};\eta),
\nabla_{\varphi}Q(\varphi,\varphi^{\prime};\chi),\nabla_{\varphi^{\prime}}
Q(\varphi,\varphi^{\prime};\chi)\nonumber
\end{equation}
exist as do the derivatives
\begin{align}
&\nabla_{r}{\mathscr{J}(r,\theta,\varphi)},~~\nabla_{r^{\prime}}{\mathscr{J}(r^{\prime},
\theta^{\prime},\varphi^{\prime})}
,\nabla_{\theta}{\mathscr{J}(r,\theta,\varphi)}\nonumber\\&~~\nabla_{\theta^{\prime}}
{\mathscr{J}(r^{\prime},\theta^{\prime},\varphi^{\prime})}
\nabla_{\varphi}{\mathscr{J}(r,\theta,\varphi)},~~\nabla_{\varphi^{\prime}}
{\mathscr{J}(r^{\prime},\theta^{\prime},\varphi^{\prime})}\nonumber
\end{align}
with expectations
\begin{align}
&\mathbf{E}\big\llbracket\nabla_{r}{\mathscr{J}(r,\theta,\varphi)}\big\rrbracket=
\mathbf{E}\big\llbracket\nabla_{r^{\prime}}
{\mathscr{J}(r^{\prime},\theta^{\prime},\varphi^{\prime})}\bigg\rrbracket
=0\nonumber\\&
\mathbf{E}\big\llbracket\nabla_{\theta}{\mathscr{J}(r,\theta,\varphi)}\big\rrbracket=
\mathbf{E}\big\llbracket\nabla_{\theta^{\prime}}{\mathscr{J}(r^{\prime},
\theta^{\prime},\varphi^{\prime})}\big\rrbracket=0
\nonumber\\&
\mathbf{E}\big\llbracket\nabla_{\varphi}{\mathscr{J}(r,\theta,\varphi)},
=\nabla_{\varphi^{\prime}}{\mathscr{J}(r^{\prime},\theta^{\prime},\varphi^{\prime})}
\big\rrbracket=0
\end{align}
However
\begin{align}
&\nabla_{r^{\prime}}{\mathscr{J}(r,\theta,\varphi)}=\nabla_{r}
{\mathscr{J}(r^{\prime},\theta^{\prime},\varphi^{\prime})}=0\nonumber\\&
\nabla_{\theta^{\prime}}{\mathscr{J}(r,\theta,\varphi)}=
\nabla_{\theta}{\mathscr{J}(r^{\prime},\theta^{\prime},\varphi^{\prime})}=0
\nonumber\\&
\nabla_{\varphi^{\prime}}{\mathscr{J}(r,\theta,\varphi)}=
\nabla_{\varphi}{\mathscr{J}(r^{\prime},\theta^{\prime},\varphi^{\prime})}=0
\end{align}
In spherical coordinates $\nabla_{i}=(\nabla_{r},\nabla_{\theta},\nabla_{\varphi})$
with $\nabla_{r}=\partial_{r},\nabla_{\theta}=\tfrac{1}{r}\partial_{\theta},
\nabla_{\varphi}=\tfrac{1}{r\sin\theta}\partial\varphi$. The possible correlations between the fields and their derivatives are then
\begin{align}
&\mathbf{E}\big\llbracket
\nabla_{r}\mathscr{J}(r,\theta,\varphi)\otimes
\nabla_{r^{\prime}}{\mathscr{J}(r^{\prime},\theta^{\prime},\varphi^{\prime})}
\big\rrbracket=\lambda\nabla_{r}\nabla_{r^{\prime}}K(r,r^{\prime})H(\theta,\theta^{\prime};\eta)
Q(\varphi,\varphi^{\prime};\chi)\nonumber\\&\equiv \lambda\partial_{r}\partial_{r^{\prime}}K(r,r^{\prime})H(\theta,\theta^{\prime};\eta)
Q(\varphi,\varphi^{\prime};\chi)\\&\mathbf{E}\big\llbracket
\nabla_{\theta}{\mathscr{J}(r,\theta,\varphi)}\otimes
\nabla_{\theta^{\prime}}{\mathscr{J}(r^{\prime},\theta^{\prime},\varphi^{\prime})}
\big\rrbracket=\lambda K(r,r^{\prime})\nabla_{\theta}\nabla_{\theta^{\prime}}H(\theta,\theta^{\prime};\eta)
Q(\varphi,\varphi^{\prime};\chi)\nonumber\\&\equiv
\lambda\frac{1}{rr^{\prime}}K(r,r^{\prime})\partial_{\theta}\partial_{\theta^{\prime}}H(\theta,\theta^{\prime};\eta)
Q(\varphi,\varphi^{\prime};\chi)
\\&\mathbf{E}\big\llbracket
\nabla_{\varphi}{\mathscr{J}(r,\theta,\varphi)}\otimes
\nabla_{\varphi^{\prime}}{\mathscr{J}(r^{\prime},\theta^{\prime},\varphi^{\prime})}
\big\rbrace=\lambda K(r,r^{\prime})H(\theta,\theta^{\prime};\eta)
\nabla_{\varphi}\nabla_{\varphi^{\prime}}Q(\varphi,\varphi^{\prime};\chi)\nonumber\\&
\lambda\frac{1}{r\sin\theta}\frac{1}{r^{\prime}\sin\theta^{\prime}}K(r,r^{\prime})H(\theta,\theta^{\prime};\eta)
\partial_{\varphi}\partial_{\varphi^{\prime}}Q(\varphi,\varphi^{\prime};\chi)
\\&\mathbf{E}\big\lbrace\partial_{r}{\mathscr{J}(r,\theta,\varphi)}\circ
\partial_{\theta^{\prime}}{\mathscr{J}(r^{\prime},\theta^{\prime},\varphi^{\prime})}
\big\rbrace=\lambda\nabla_{r}K(r,r^{\prime})\nabla_{\theta^{\prime}}H(\theta,\theta^{\prime};\eta)
Q(\varphi,\varphi^{\prime};\chi)\nonumber\\&
\lambda\partial_{r}K(r,r^{\prime})\tfrac{1}{r}\partial_{\theta^{\prime}}H(\theta,\theta^{\prime};\eta)
Q(\varphi,\varphi^{\prime};\chi)
\\&\mathbf{E}\big\llbracket\nabla_{r}{\mathscr{J}(r,\theta,\varphi)}\otimes
\nabla_{\varphi^{\prime}}{\mathscr{J}(r^{\prime},\theta^{\prime},\varphi^{\prime})}\big\rrbracket=\lambda \nabla_{r}K(r,r^{\prime})H(\theta,\theta^{\prime};\eta)
\nabla_{\varphi^{\prime}}Q(\varphi,\varphi^{\prime};\chi)\nonumber\\&
=\lambda \partial_{r}K(r,r^{\prime})H(\theta,\theta^{\prime};\eta)
\tfrac{1}{r^{\prime}\sin\theta^{\prime}}\partial_{\varphi^{\prime}}Q(\varphi,\varphi^{\prime};\chi)
\\&\mathbf{E}\big\lbrace
\nabla_{\theta}{\mathscr{J}(r,\theta,\varphi)}
\circ\nabla_{r^{\prime}}{\mathscr{J}(r^{\prime},\theta^{\prime},\varphi^{\prime})}
\big\rbrace=\lambda\nabla_{r^{\prime}}K(r,r^{\prime})\nabla_{\theta}H(\theta,\theta^{\prime};\eta)Q(\varphi,\varphi^{\prime};\chi)
\nonumber\\& \equiv\lambda\partial_{r^{\prime}}K(r,r^{\prime})\tfrac{1}{r}\partial_{\theta}
H(\theta,\theta^{\prime};\eta)Q(\varphi,\varphi^{\prime};\chi)
\\&\mathbf{E}\big\lbrace\nabla_{\theta}{\mathscr{J}(r,\theta,\varphi)}
\circ\nabla_{\varphi^{\prime}}{\mathscr{J}(r^{\prime},\theta^{\prime},\varphi^{\prime})}
\big\rbrace=\lambda\nabla_{r^{\prime}}K(r,r^{\prime})H(\theta,\theta^{\prime};\eta)
\nabla_{\varphi^{\prime}}Q(\varphi,\varphi^{\prime};\chi)\nonumber\\&
\lambda\partial_{r^{\prime}}K(r,r^{\prime})H(\theta,\theta^{\prime};\eta)
\tfrac{1}{r^{\prime}\sin\theta^{\prime}}\partial_{\varphi^{\prime}}Q(\varphi,\varphi^{\prime};\chi)...
\\&\mathbf{E}\big\llbracket\partial_{\varphi}{\mathscr{J}(r,\theta,\varphi)}
\otimes\nabla_{r^{\prime}}{\mathscr{J}(r^{\prime},\theta^{\prime},\varphi^{\prime})}
\big\rrbracket=\lambda\nabla_{r^{\prime}}K(r,r^{\prime})\nabla_{\theta}H(\theta,\theta^{\prime};\eta)
\nabla_{\varphi} Q(\varphi,\varphi^{\prime};\chi)\nonumber\\&
=\lambda\partial_{r^{\prime}}K(r,r^{\prime})\tfrac{1}{r}\partial_{\theta}H(\theta,\theta^{\prime};\eta)\partial_{\varphi} Q(\varphi,\varphi^{\prime};\chi)
\\&\mathbf{E}\big\llbracket\nabla_{\varphi}{\mathscr{J}(r,\theta,\varphi)}
\otimes\nabla_{\theta^{\prime}}{\mathscr{J}(r^{\prime},\theta^{\prime},\varphi^{\prime})}
\big\rrbracket=\lambda K(r,r^{\prime};\epsilon)\nabla_{\theta^{\prime}}H(\theta,\theta^{\prime};\eta)\nabla_{\varphi} Q(\varphi,\varphi^{\prime};\chi)\nonumber\\&=\lambda K(r,r^{\prime};\epsilon)\tfrac{1}{r^{\prime}}\nabla_{\theta^{\prime}}
H(\theta,\theta^{\prime};\eta)\tfrac{1}{r\sin\theta}\partial_{\varphi} Q(\varphi,\varphi^{\prime};\chi)
\\&\mathbf{E}\big\llbracket{\mathscr{J}(r,\theta,\varphi)}
\otimes\nabla_{r^{\prime}}{\mathscr{J}(r^{\prime},\theta^{\prime},\varphi^{\prime})}
\big\rrbracket=\lambda \nabla_{r^{\prime}}K(r,r^{\prime};\epsilon)
H(\theta,\theta^{\prime};\eta)Q(\varphi,\varphi,\chi)\nonumber\\&
\equiv\lambda \partial_{r^{\prime}}K(r,r^{\prime};\epsilon)
H(\theta,\theta^{\prime};\eta)Q(\varphi,\varphi,\chi)
\\&\mathbf{E}\big\llbracket{\mathscr{J}(r,\theta,\varphi)}
\circ\nabla_{\theta^{\prime}}{\mathscr{J}(r^{\prime},\theta^{\prime},\varphi^{\prime})}
\big\rrbracket=\lambda K(r,r^{\prime};\epsilon)\nabla_{\theta^{\prime}}H(\theta,\theta^{\prime};\eta)Q(\varphi,\varphi,\chi)
\nonumber\\&=\lambda K(r,r^{\prime};\epsilon)
\tfrac{1}{r^{\prime}}\partial_{\theta^{\prime}}H(\theta,\theta^{\prime};\eta)
Q(\varphi,\varphi,\chi)
\\&\mathbf{E}\big\llbracket{\mathscr{J}(r,\theta,\varphi)}
\circ\nabla_{\theta^{\prime}}{\mathscr{J}(r^{\prime},\theta^{\prime},\varphi^{\prime})}
\big\rrbracket=\lambda K(r,r^{\prime};\epsilon)H(\theta,\theta^{\prime};\eta)\nabla_{\varphi^{\prime}}Q(\varphi,\varphi,\chi)\nonumber\\&
\equiv\lambda K(r,r^{\prime};\epsilon)H(\theta,\theta^{\prime};\eta)\tfrac{1}{r^{\prime}\sin\theta^{\prime}}
\partial_{\varphi^{\prime}}Q(\varphi,\varphi^{\prime},\chi)
\\&\mathbf{E}\big\llbracket\nabla_{r}{\mathscr{J}(r,\theta,\varphi)}
\otimes {\mathscr{J}(r^{\prime},\theta^{\prime},\varphi^{\prime})}
\big\rrbracket=\lambda \nabla_{r}K(r,r^{\prime};\epsilon)
H(\theta,\theta^{\prime};\eta)Q(\varphi,\varphi,\chi)\nonumber\\&
=\lambda\partial_{r}K(r,r^{\prime};\epsilon)
H(\theta,\theta^{\prime};\eta)Q(\varphi,\varphi,\chi)
\\&\mathbf{E}\big\llbracket\nabla_{\theta}
\mathscr{J}(r,\theta,\varphi)
\otimes\mathscr{J}(r^{\prime},\theta^{\prime},\varphi^{\prime})
\big\rrbracket=\lambda K(r,r^{\prime};\epsilon)
\nabla_{\theta}H(\theta,\theta^{\prime};\eta)Q(\varphi,\varphi,\chi)
\nonumber\\&=\lambda K(r,r^{\prime};\epsilon)\nabla_{\theta^{\prime}}\tfrac{1}{r}\partial_{\theta}H(\theta,\theta^{\prime};\eta)Q(\varphi,\varphi,\chi)
\\&\mathbf{E}\big\lbrace\nabla_{\varphi}
\mathscr{J}(r,\theta,\varphi)
\circ\mathscr{J}(r^{\prime},\theta^{\prime},\varphi^{\prime})
\big\rrbracket=\lambda K(r,r^{\prime};\epsilon)H(\theta,\theta^{\prime};\eta)\nabla_{\varphi}Q(\varphi,\varphi,\chi)
\nonumber\\&=\lambda K(r,r^{\prime};\epsilon)H(\theta,\theta^{\prime};\eta)\tfrac{1}{r\sin\theta}\partial_{\varphi}Q(\varphi,\varphi,\chi)
\end{align}
\end{prop}
\begin{cor}
The regulated covariances or volatilities are finite and bounded as
\begin{align}
&\lim_{r\rightarrow r^{\prime},\theta\rightarrow\theta^{\prime},\varphi\rightarrow\varphi^{\prime}}
\mathbf{E}\big\llbracket
\nabla_{r}\mathscr{J}(r,\theta,\varphi)\otimes
\nabla_{r^{\prime}}\mathscr{J}(r^{\prime},\theta^{\prime},\varphi^{\prime})
\big\rrbracket=\lambda\\&\lim_{r\rightarrow r^{\prime},\theta\rightarrow\theta^{\prime},\varphi\rightarrow\varphi^{\prime}}
\mathbf{E}\big\llbracket
\nabla_{\theta}\mathscr{J}(r,\theta,\varphi)\otimes
\nabla_{\theta^{\prime}}\mathscr{J}(r^{\prime},\theta^{\prime},\varphi^{\prime})
\big\rrbracket=\lambda\\&\lim_{r\rightarrow r^{\prime},\theta\rightarrow\theta^{\prime},\varphi\rightarrow\varphi^{\prime}}
\mathbf{E}\big\llbracket
\nabla_{\varphi}\mathscr{J}(r,\theta,\varphi)\otimes
\nabla_{\varphi^{\prime}}\mathscr{J}(r^{\prime},\theta^{\prime},\varphi^{\prime})
\big\rbrace=\lambda\\&\lim_{r\rightarrow r^{\prime},\theta\rightarrow\theta^{\prime},\varphi\rightarrow\varphi^{\prime}}
\mathbf{E}\big\llbracket\partial_{r}
\mathscr{J}(r,\theta,\varphi)
\otimes\nabla_{\theta^{\prime}}
\mathscr{J}(r^{\prime},\theta^{\prime},\varphi^{\prime})\big\rrbracket=\lambda
\\&\lim_{r\rightarrow r^{\prime},\theta\rightarrow\theta^{\prime},\varphi\rightarrow\varphi^{\prime}}
\mathbf{E}\big\llbracket\nabla_{r}\mathscr{J}(r,\theta,\varphi)
\otimes\nabla_{\varphi^{\prime}}\mathscr{J}(r^{\prime},\theta^{\prime},\varphi^{\prime})
\big\rrbracket=\lambda
\\&\lim_{r\rightarrow r^{\prime},\theta\rightarrow\theta^{\prime},\varphi\rightarrow\varphi^{\prime}}
\mathbf{E}\big\llbracket\nabla_{\theta}
\mathscr{J}(r,\theta,\varphi)
\otimes\nabla_{r^{\prime}}\mathscr{J}(r^{\prime},\theta^{\prime},\varphi^{\prime})
\big\rrbracket=\lambda
\\&\lim_{r\rightarrow r^{\prime},\theta\rightarrow\theta^{\prime},\varphi\rightarrow\varphi^{\prime}}
\mathbf{E}\big\llbracket\nabla_{\theta}\mathscr{J}(r,\theta,\varphi)
\otimes\nabla_{\varphi^{\prime}}\mathscr{J}(r^{\prime},\theta^{\prime},\varphi^{\prime})
\big\rrbracket=\lambda
\\&\lim_{r\rightarrow r^{\prime},\theta\rightarrow\theta^{\prime},\varphi\rightarrow\varphi^{\prime}}
\mathbf{E}\big\llbracket\nabla_{\varphi}
\mathscr{J}(r,\theta,\varphi)
\circ\nabla_{r^{\prime}}\mathscr{J}(r^{\prime},\theta^{\prime},\varphi^{\prime})
\big\rrbracket=\lambda
\\&\lim_{r\rightarrow r^{\prime},\theta\rightarrow\theta^{\prime},\varphi\rightarrow\varphi^{\prime}}
\mathbf{E}\big\llbracket\nabla_{\varphi}
\mathscr{J}(r,\theta,\varphi)
\otimes\nabla_{\theta^{\prime}}\mathscr{J}(r^{\prime},\theta^{\prime},
\varphi^{\prime})
\big\rrbracket=\lambda
\\&\lim_{r\rightarrow r^{\prime},\theta\rightarrow\theta^{\prime},\varphi\rightarrow\varphi^{\prime}}
\mathbf{E}\big\llbracket\mathscr{J}(r,\theta,\varphi)
\otimes\nabla_{r^{\prime}}\mathscr{J}(r^{\prime},\theta^{\prime},
\varphi^{\prime})
\big\rrbracket=\lambda
\\&\lim_{r\rightarrow r^{\prime},\theta\rightarrow\theta^{\prime},\varphi\rightarrow\varphi^{\prime}}\mathbf{E}
\big\llbracket\mathscr{J}(r,\theta,\varphi)
\otimes\nabla_{\theta^{\prime}}\mathscr{J}(r^{\prime},\theta^{\prime},\varphi^{\prime})
\big\rrbracket=\lambda
\\&\lim_{r\rightarrow r^{\prime},\theta\rightarrow\theta^{\prime},\varphi\rightarrow\varphi^{\prime}}
\mathbf{E}\big\llbracket\mathscr{J}(r,\theta,\varphi)
\otimes\nabla_{\theta^{\prime}}\mathscr{J}(r^{\prime},\theta^{\prime},\varphi^{\prime})
\big\rrbracket=\lambda
\\&\lim_{r\rightarrow r^{\prime},\theta\rightarrow\theta^{\prime},\varphi\rightarrow\varphi^{\prime}}\mathbf{E}\big\lbrace\nabla_{\theta}
\mathscr{J}(r,\theta,\varphi)
\otimes\mathscr{J}(r^{\prime},\theta^{\prime},\varphi^{\prime})
\big\rrbracket=\lambda
\\&\lim_{r\rightarrow r^{\prime},\theta\rightarrow\theta^{\prime},\varphi\rightarrow\varphi^{\prime}}
\mathbf{E}\big\llbracket\nabla_{\varphi}\mathscr{J}(r,\theta,\varphi)
\otimes\mathscr{J}(r^{\prime},\theta^{\prime},\varphi^{\prime})
\big\rrbracket=\lambda
\end{align}
\end{cor}
\begin{cor}
At the surface of the ball $r=r^{\prime}=R$ and $K(R,R)=1$ so the correlations on the surface/boundary are
\begin{align}
&\mathbf{E}\llbracket\nabla_{\theta}
\mathscr{J}(\theta,\varphi)\otimes
\nabla_{\theta}
\mathscr{J}(\theta^{\prime},\varphi^{\prime})=\lambda\tfrac{1}{rr^{\prime}}
\partial_{\theta}\partial_{\theta^{\prime}}H(\theta,\theta^{\prime};\eta)
Q(\varphi,\varphi^{\prime};\chi)
\\&\mathbf{E}\llbracket\nabla_{\theta}\mathscr{J}(\theta,\varphi)
\otimes\nabla_{\varphi^{\prime}}\mathscr{J}(\theta^{\prime},\varphi^{\prime})
=\lambda\tfrac{1}{r}\partial_{\theta}H(\theta,\theta^{\prime};\eta)\tfrac{1}{r^{\prime}
\sin\theta}\partial_{\varphi^{\prime}}
Q(\varphi,\varphi^{\prime};\chi)\\&
\mathbf{E}\llbracket\nabla_{\varphi}\mathscr{J}(\theta,\varphi)\otimes
\nabla_{\theta^{\prime}}\mathscr{J}(\theta^{\prime},\varphi^{\prime})=\lambda
\partial_{\theta^{\prime}}H(\theta,\theta^{\prime};\eta)\tfrac{1}{r\sin\theta}
\partial_{\varphi}Q(\varphi,\varphi^{\prime};\chi)\\&
\mathbf{E}\llbracket\nabla_{\varphi}\mathscr{J}(\theta,\varphi)\otimes
\nabla_{\theta^{\prime}}\mathscr{J}(\theta^{\prime},\varphi^{\prime})=\lambda
H(\theta,\theta^{\prime};\eta)\tfrac{1}{r\sin\theta}
\tfrac{1}{r^{\prime}\sin\theta^{\prime}}\partial_{\varphi}\partial_{\varphi^{\prime}}Q(\varphi,\varphi^{\prime};\chi)
\\&
\mathbf{E}\llbracket\mathscr{J}(\theta,\varphi)
\otimes\nabla_{\theta^{\prime}}\mathscr{J}(\theta^{\prime},\varphi^{\prime})
\rbrace=\lambda\tfrac{1}{r^{\prime}}\partial_{\theta^{\prime}}H(\theta,\theta^{\prime};\eta)
Q(\varphi,\varphi^{\prime};\chi)\\&
\mathbf{E}\big\llbracket\mathscr{J}(\theta,\varphi)
\otimes\nabla_{\varphi^{\prime}}\mathscr{J}(\theta^{\prime},\varphi^{\prime})
\rbrace=\lambda H(\theta,\theta^{\prime};\eta)
\tfrac{1}{r^{\prime}\sin\theta^{\prime}}Q(\varphi,\varphi^{\prime};\chi)
\\&\mathbf{E}\big\llbracket\nabla_{\theta}
\mathscr{J}(\theta,\varphi)\otimes\mathscr{J}(\theta^{\prime},\varphi^{\prime})
=\frac{\lambda}{r}\partial_{\theta}
H(\theta,\theta^{\prime};\chi)Q(\varphi,\varphi^{\prime};\chi)
\\&
\mathbf{E}\big\llbracket\nabla_{\varphi}
\mathscr{J}(\theta,\varphi)\otimes\mathscr{J}(\theta^{\prime},\varphi^{\prime})
=\lambda H(\theta,\theta^{\prime};\chi)\tfrac{1}{r\sin\theta}\partial_{\varphi}
Q(\varphi,\varphi^{\prime};\chi)
\end{align}
\end{cor}
\begin{cor}
The volatilities at any point on the surface $r=r^{\prime}=R$ are finite and bounded so that
\begin{align}
&\lim_{r\rightarrow r^{\prime},\theta\rightarrow\theta^{\prime},\varphi\rightarrow\varphi^{\prime}}
\mathbf{E}\llbracket
\nabla_{\theta}\mathscr{J}(\theta,\varphi)
\otimes\nabla_{\theta}
\mathscr{J}(\theta^{\prime},\varphi^{\prime})=\lambda
\\&\lim_{r\rightarrow r^{\prime},\theta\rightarrow\theta^{\prime},\varphi\rightarrow\varphi^{\prime}}
\mathbf{E}\lbrace\nabla_{\theta}\mathscr{J}(\theta,\varphi)
\otimes\nabla_{\varphi^{\prime}}\mathscr{J}(\theta^{\prime},\varphi^{\prime})
=\lambda\\&
\lim_{r\rightarrow r^{\prime},\theta\rightarrow\theta^{\prime},\varphi\rightarrow\varphi^{\prime}}
\mathbf{E}\lbrace\nabla_{\varphi}\mathscr{J}(\theta,\varphi)\otimes
\nabla_{\theta^{\prime}}\mathscr{J}(\theta^{\prime},\varphi^{\prime})=\lambda
\\&
\lim_{r\rightarrow r^{\prime},\theta\rightarrow\theta^{\prime},\varphi\rightarrow\varphi^{\prime}}
\mathbf{E}\llbracket\nabla_{\varphi}\mathscr{J}(\theta,\varphi)\otimes
\nabla_{\theta^{\prime}}\mathscr{J}(\theta^{\prime},\varphi^{\prime})
=\lambda
\\&\lim_{r\rightarrow r^{\prime},\theta\rightarrow\theta^{\prime},\varphi\rightarrow\varphi^{\prime}}
\mathbf{E}\llbracket\mathscr{J}(\theta,\varphi)
\otimes\nabla_{\theta^{\prime}}\mathscr{J}(\theta^{\prime},\varphi^{\prime})
\rrbracket=\lambda\\&
\lim_{r\rightarrow r^{\prime},\theta\rightarrow\theta^{\prime},\varphi\rightarrow\varphi^{\prime}}
\mathbf{E}\big\llbracket\mathscr{J}(\theta,\varphi)
\otimes\nabla_{\varphi^{\prime}}\mathscr{J}(\theta^{\prime},\varphi^{\prime})
\rrbracket=\lambda
\\&
\lim_{r\rightarrow r^{\prime},\theta\rightarrow\theta^{\prime},\varphi\rightarrow\varphi^{\prime}}
\mathbf{E}\big\llbracket\nabla_{\theta}
\mathscr{J}(\theta,\varphi)\otimes \mathscr{J}(\theta^{\prime},\varphi^{\prime})\rrbracket
=\lambda\\& \lim_{r\rightarrow r^{\prime},\theta\rightarrow\theta^{\prime},\varphi\rightarrow\varphi^{\prime}}
\mathbf{E}\big\llbracket\nabla_{\varphi}
\mathscr{J}(\theta,\varphi)\otimes \mathscr{J}(\theta^{\prime},\varphi^{\prime})\rrbracket
=\lambda
\end{align}
\end{cor}
\subsubsection{GRSFs on a finite cylinder}
Let $\bm{C}_{R,L}\subset\bm{\mathrm{R}}^{3}$ be a finite cylinder of radius $R$ and length $L$. The GRSF within $\bm{C}_{R,L}$ is ${\mathscr{J}(r,z,\varphi)}$. With rotational symmetry about the cylinder axis, the random field is ${\mathscr{J}(r,z)}$. The expectation and 2-point functions are
\begin{align}
&\mathbf{E}\big\llbracket{\mathscr{J}(r,z)}\big\rrbracket=0\\&
\mathbf{E}\big\llbracket{\mathscr{J}(r,z)}\otimes
{\mathscr{J}(r^{\prime},z^{\prime})}\big\rrbracket=\lambda K(r,r^{\prime};\epsilon)Z(z,z^{\prime};\xi)
\end{align}
and the 2-point function is regulated so that the volatility of the field at any point in the cylinder is finite and bounded
\begin{equation}
\lim_{r\rightarrow r^{\prime},z\rightarrow z^{\prime}}
\mathbf{E}\big\llbracket
\mathscr{J}(r,z)\otimes
\mathscr{J}(r^{\prime},z^{\prime})\big\rrbracket=\lambda
\end{equation}
The gradient is $\nabla_{i}=(\partial_{r},\partial_{z})$ so that the possible 2-point covariances for any two points $(r,z)\in\bm{C}_{R,L}$ and
$(r^{\prime},z^{\prime})\in\bm{C}_{R,L}$ are
\begin{align}
&\mathbf{E}\llbracket\nabla_{r}\mathscr{J}(r,z)\otimes \nabla_{r^{\prime}}\mathscr{J}(r^{\prime},z^{\prime})\rrbracket=\lambda \partial_{r}\partial_{r^{\prime}}
K(r,r;\epsilon)Z(z,z^{\prime})\\&
\mathbf{E}\llbracket\nabla_{z}\mathscr{J}(r,z)\otimes \nabla_{z^{\prime}}\mathscr{J}(r^{\prime},z^{\prime})\rrbracket=\lambda K(r,r;\epsilon)\partial_{z}\partial_{z^{\prime}}Z(z,z^{\prime}\\&
\mathbf{E}\llbracket\nabla_{r}\mathscr{J}(r,z)\otimes \nabla_{z^{\prime}}\mathscr{J}(r^{\prime},z^{\prime})\rrbracket=\lambda \partial_{r}K(r,r;\epsilon)\partial_{z^{\prime}}Z(z,z^{\prime}\\&
\mathbf{E}\llbracket\nabla_{z}\mathscr{J}(r,z)\otimes \nabla_{z^{\prime}}\mathscr{J}(r^{\prime},z^{\prime})\rrbracket=\lambda
\partial_{r^{\prime}}K(r,r;\epsilon)\partial_{z}Z(z,z^{\prime}
\end{align}
The volatilities are finite and bounded since
\begin{align}
&\lim_{r\rightarrow^{\prime},z\rightarrow z^{\prime}}\mathbf{E}\llbracket \nabla_{r}\mathscr{J}(r,z)\otimes\nabla_{r^{\prime}}\mathscr{J}(r^{\prime},z^{\prime})\rrbracket=\lambda
\\&\lim_{r\rightarrow^{\prime},z\rightarrow z^{\prime}}\mathbf{E}\llbracket\nabla_{z}\mathscr{J}(r,z)\otimes \nabla_{z^{\prime}}\mathbf{J}(r^{\prime},z^{\prime})\rrbracket=\lambda\\&
\lim_{r\rightarrow^{\prime},z\rightarrow z^{\prime}}\mathbf{E}\llbracket\nabla_{r}\mathscr{J}(r,z)\otimes \nabla_{z^{\prime}}\mathscr{J}(r^{\prime},z^{\prime})\rrbracket=\lambda\\&
\lim_{r\rightarrow^{\prime},z\rightarrow z^{\prime}}\mathbf{E}\llbracket\nabla_{z}\mathscr{J}(r,z)\otimes \nabla_{r^{\prime}}\mathscr{J}(r^{\prime},z^{\prime})\rrbracket=\lambda
\end{align}
At the cylinder boundary $r=r^{\prime}=R$
\begin{align}
\mathbf{E}\llbracket\nabla_{z}\mathscr{J}(R,z)\otimes \nabla_{z^{\prime}}\mathscr{J}(R,z^{\prime})\rrbracket=\lambda K(R,R;\epsilon)\partial_{z}\partial_{z^{\prime}}Z(z,z^{\prime})
=\partial_{z}\partial_{z^{\prime}}Z(z,z^{\prime})
\end{align}
\begin{lem}
Finally, the correlations for a SGRF $\mathscr{J}(r,\theta)$ on a disc $\bm{\mathcal{D}}\subset\bm{\mathrm{R}}^{2}$ are
\begin{align}
&\mathbf{E}\llbracket \nabla_{r}{\mathscr{J}(r,\theta)}\otimes \nabla_{r^{\prime}}{\mathscr{J}(r^{\prime},\theta^{\prime})}\rrbracket=\lambda \partial_{r}\partial_{r^{\prime}}K(r,r;\epsilon)H(\theta,\theta^{\prime})\\&
\mathbf{E}\llbracket\nabla_{\theta}{\mathscr{J}(r,z)}\otimes \nabla_{\theta^{\prime}}{\mathscr{J}(r^{\prime},\theta^{\prime})}\rrbracket=\lambda K(r,r;\epsilon)\partial_{\theta}\partial_{\theta^{\prime}}H(\theta,\theta^{\prime})\\&
\mathbf{E}\llbracket\nabla_{r}{\mathscr{J}(x)(r,\theta)}\otimes \nabla_{\theta^{\prime}}{\mathscr{J}(r^{\prime},\theta^{\prime})}\rrbracket=\lambda \partial_{r}K(r,r;\epsilon)\partial_{z^{\prime}}H(\theta,\theta^{\prime})\\&
\mathbf{E}\llbracket\nabla_{\theta}{\mathscr{J}(r,\theta)}\otimes \nabla_{\theta^{\prime}}{\mathscr{J}(r^{\prime},\theta^{\prime})}\rbrace=\lambda
\partial_{r^{\prime}}K(r,r;\epsilon)\partial_{\theta}H(\theta,\theta^{\prime})
\end{align}
\end{lem}
\begin{cor}
The volatilities are finite and bounded since
\begin{align}
&\lim_{r\rightarrow^{\prime},\theta\rightarrow \theta^{\prime}}\mathbf{E}\llbracket \nabla_{r}\mathscr{J}(r,\theta)\otimes \nabla_{r^{\prime}}\mathscr{J}(x)(r^{\prime},\theta^{\prime})\rrbracket=\lambda
\\& \lim_{r\rightarrow^{\prime},z\rightarrow
\theta^{\prime}}\mathbf{E}\llbracket\nabla_{\theta}\mathscr{J}
(r,\theta)\otimes\nabla_{\theta^{\prime}}\mathscr{J}(r^{\prime},\theta^{\prime})\rrbracket
=\lambda\\&\lim_{r\rightarrow^{\prime},\theta\rightarrow \theta^{\prime}}\mathbf{E}\llbracket\nabla_{r}
\mathscr{J}(r,\theta)\otimes\nabla_{\theta^{\prime}}\mathscr{J}(r^{\prime},\theta^{\prime})
\rrbracket=\lambda\\&
\lim_{r\rightarrow^{\prime},\theta\rightarrow \theta^{\prime}}\mathbf{E}\llbracket\nabla_{\theta}
\mathscr{J}(r,\theta)\otimes\nabla_{r^{\prime}}\mathscr{J}(r^{\prime},\theta^{\prime})\rrbracket=\lambda
\end{align}
\end{cor}
\section{Proof of Theorem 4.28.}
\begin{proof}
To prove $\Delta\overline{\psi(r,\theta)}=0$, apply the Laplace operator in polar coordinates to giving
\begin{align}
&\Delta \overline{\psi(r,\theta)}=\frac{1}{2\pi}\int_{0}^{2\pi}\Delta\bigg| \frac{R^{2}-r^{2}}{R^{2}-2rR\cos(\theta-\beta)+r^{2}}\bigg|g(\beta)d\beta\nonumber\\&+\frac{1}{2\pi}\int_{0}^{2\pi}\Delta \bigg|\frac{R^{2}-r^{2}}{R^{2}-2rR\cos(\theta-\beta)+r^{2}}\bigg|{\mathscr{J}(\beta)}d\beta=0
\end{align}
since $\Delta$ acts only on the terms
\begin{align}
&\mathbf{E}\bigg\llbracket\Delta \overline{\psi(r,\theta)}\bigg\rrbracket=\frac{1}{2\pi}\int_{0}^{2\pi}\Delta\bigg| \frac{R^{2}-r^{2}}{R^{2}-2rR\cos(\theta-\beta)+r^{2}}\bigg|g(\beta)d\beta\nonumber\\&+\frac{1}{2\pi}\int_{0}^{2\pi}\Delta \bigg|\frac{R^{2}-r^{2}}{R^{2}-2rR\cos(\theta-\beta)+r^{2}}\bigg|
\mathbf{E}\bigg\llbracket{\mathscr{J}(\beta)}\bigg\rrbracket d\beta=0
\end{align}
so that $\overline{\psi(r,\theta)}$ is both harmonic and stochastically harmonic. To compute the moments, the GRF $\overline{\psi(r\theta)}$ induced within the disc is
\begin{align}
&\overline{\psi(r,\theta)}=\frac{1}{2\pi}\int_{0}^{2\pi} \frac{(R^{2}-r^{2})g(\beta)d\beta}{R^{2}-2rR\cos(\theta-\beta)+r^{2}}+\frac{1}{2\pi}\int_{0}^{2\pi} \frac{(R^{2}-r^{2}){\mathscr{J}(\beta)}d\beta}{R^{2}-2rR\cos(\theta-\beta)+r^{2}}\nonumber\\&
=\psi(r,\theta)+\frac{1}{2\pi}\int_{0}^{2\pi} \frac{(R^{2}-r^{2}){\mathscr{J}(\beta)}d\beta}{R^{2}-2rR\cos(\theta-\beta)+r^{2}}
\end{align}
Then
\begin{align}
&\mathbf{E}\bigg\llbracket|\overline{\psi(r,\theta)}|^{p}\bigg\rrbracket=
\mathbf{E}\bigg\llbracket\bigg||\psi(r,\theta)|+\frac{1}{2\pi}\int_{0}^{2\pi} \frac{(R^{2}-r^{2}){\mathscr{J}(\beta)}d\beta}{R^{2}-2rR\cos(\theta-\beta)+r^{2}}\bigg|^{p}
\bigg\rrbracket
\end{align}
Now using the binomial expansion for vanishing odd terms gives
\begin{align}
\mathbf{E}\bigg\llbracket|\overline{\psi(r,\theta)}|^{p}\bigg\rrbracket&=
\frac{1}{2}\sum_{Q=0}^{P}\binom{P}{Q}|\psi(r,\theta)|^{P-Q}\mathbf{E}\bigg\llbracket\bigg(\frac{1}{2\pi}\int_{0}^{2\pi} \frac{(R^{2}-r^{2}){\mathscr{J}(\beta)}d\beta}{R^{2}-2rR\cos(\theta-\beta)+r^{2}}
\bigg)^{Q}\nonumber\\&+\bigg(-\frac{1}{2\pi}\int_{0}^{2\pi} \frac{(R^{2}-r^{2}){\mathscr{J}(\beta)}d\beta}{R^{2}-2rR\cos(\theta-\beta)+r^{2}}
\bigg)^{Q}\bigg]\bigg\rrbracket\nonumber\\&=\frac{1}{2}\sum_{Q=0}^{p}\binom{P}{Q}
|\psi(r,\theta)|^{P-Q}\mathbf{E}\bigg\llbracket\bigg(\frac{1}{2\pi}\int_{0}^{2\pi} \frac{(R^{2}-r^{2}){\mathscr{J}(\beta)}d\beta}{R^{2}-2rR\cos(\theta-\beta)+r^{2}}\nonumber\\&
\underbrace{\times...\times}_{Q~times}\frac{1}{2\pi}\int_{0}^{2\pi} \frac{(R^{2}-r^{2}){\mathscr{J}(\beta)}d\beta}{R^{2}-2rR\cos(\theta-\beta)+r^{2}}
\bigg)\bigg\rrbracket\nonumber\\&+\frac{1}{2}\sum_{Q=0}^{P}\binom{P}{Q}
|\psi(r,\theta)|^{P-Q}\mathbf{E}\bigg\llbracket\bigg(-\frac{1}{2\pi}\int_{0}^{2\pi} \frac{\mathscr{J}(\beta)(R^{2}-r^{2}){\mathscr{J}(\beta)}d\beta}{R^{2}-2rR\cos(\theta-\beta)+r^{2}}\bigg)\nonumber\\&\underbrace{
\times...\times}_{Q~times}\bigg(-\frac{1}{2\pi}\int_{0}^{2\pi} \frac{(R^{2}-r^{2})}{R^{2}-2rR\cos(\theta-\beta)+r^{2}}
\bigg)\bigg)\bigg\rrbracket\nonumber\\&\le C\frac{1}{2}\sum_{Q=0}^{P}\binom{P}{Q}|\psi(r,\theta)|^{P-Q}\mathbf{E}\bigg\llbracket\bigg[\frac{1}{(2\pi)^{Q}}\int_{0}^{2\pi}...
\int_{0}^{2\pi}\bigg|\frac{(R^{2}-r^{2})}{R^{2}-2rR\cos(\theta-\beta)+r^{2}}\bigg|^{Q}
\nonumber\\&\underbrace{{\mathscr{J}(\beta)}\otimes...\otimes{\mathscr{J}(\beta)}d\beta\times...\times d\beta}_{Q~times}\bigg]\bigg\rrbracket\nonumber\\&+\frac{1}{2}\sum_{Q=0}^{P}\binom{P}{Q}
|\psi(r,\theta)|^{P-Q}\mathbf{E}\bigg\llbracket\bigg[\frac{(-1)^{Q}}{(2\pi)^{Q}}\int_{0}^{2\pi}...
\int_{0}^{2\pi}\bigg|\frac{(R^{2}-r^{2})\otimes{\mathscr{J}(\beta)}d\beta}{R^{2}
-2rR\cos(\theta-\beta)+r^{2}}\bigg|^{Q}\nonumber\\&\underbrace{{\mathscr{J}(\beta)}\otimes...
\otimes{\mathscr{J}(\beta)}d\beta\times...\times d\beta}_{Q~times}\bigg]\bigg\rrbracket\nonumber\\&=C\frac{1}{2}\sum_{Q=0}^{P}\binom{P}{Q}|\psi(r,\theta)|^{P-Q}\bigg[\frac{1}{(2\pi)^{Q}}\int_{0}^{2\pi}...
\int_{0}^{2\pi}\bigg|\frac{(R^{2}-r^{2})}{R^{2}-2rR\cos(\theta-\beta)+r^{2}}\bigg|^{Q}\nonumber\\&
\underbrace{\mathbf{E}\bigg\llbracket{\mathscr{J}(\beta)}\otimes...\otimes {\mathscr{J}(\beta)}\bigg\rrbracket d\beta\times...\times d\beta}_{Q~times}\bigg]\bigg\rrbracket\nonumber\\&+C\frac{1}{2}\sum_{Q=0}^{P}\binom{P}{Q}
|\psi(r,\theta)|^{P-Q}\bigg[\frac{(-1)^{Q}}{(2\pi)^{Q}}\int_{0}^{2\pi}...
\int_{0}^{2\pi}\bigg|\frac{(R^{2}-r^{2})}{R^{2}-2rR\cos(\theta-\beta)+r^{2}}\bigg|^{Q}\nonumber\\&
\times\nonumber\\&\underbrace{\mathbf{E}\bigg\llbracket\mathscr{J}(\beta)\otimes...\otimes {\mathscr{J}(\beta)}\bigg\rrbracket d\beta\times...\times d\beta}_{Q~times}\bigg]\nonumber\\&=C\frac{1}{2}\sum_{Q=0}^{P}\binom{P}{Q}
|\psi(r,\theta)|^{P-Q}\bigg[\frac{1}{(2\pi)^{Q}}\int_{0}^{2\pi}...
\int_{0}^{2\pi}\bigg|\frac{(R^{2}-r^{2})}{R^{2}-2rR\cos(\theta-\beta)+r^{2}}\bigg|^{Q}
d\beta\times...\times d\beta\bigg]\nonumber\\&+\frac{1}{2}\sum_{Q=0}^{P}\binom{P}{Q}
|\psi(r,\theta)|^{P-Q}\bigg[\frac{(-1)^{Q}}{(2\pi)^{Q}}\int_{0}^{2\pi}...
\int_{0}^{2\pi}\bigg|\frac{(R^{2}-r^{2})}{R^{2}-2rR\cos(\theta-\beta)+r^{2}}\bigg|^{Q}
d\beta\times...\times d\beta\bigg]\nonumber\\&+\frac{1}{2}C\sum_{Q=0}^{P}\binom{P}{Q}
|\psi(r,\theta)|^{P-Q}\bigg[\bigg(\frac{1}{(2\pi)}\int_{0}^{2\pi}
\frac{(R^{2}-r^{2})}{R^{2}-2rR\cos(\theta-\beta)+r^{2}}
d\beta\bigg)^{Q}\bigg]\nonumber\\&+C\frac{1}{2}\sum_{Q=0}^{P}\binom{P}{Q}
|\psi(r,\theta)|^{P-Q}\bigg[(-1)^{Q}\bigg(\frac{1}{(2\pi)}\int_{0}^{2\pi}
\frac{(R^{2}-r^{2})}{R^{2}-2rR\cos(\theta-\beta)+r^{2}}
d\beta\bigg)^{Q}\bigg]
\end{align}
Evaluating the angular integral then gives the moments as
\begin{align}
&\mathbf{E}\big\llbracket|\overline{\psi(r,\theta)}|^{p}\rrbracket\le\frac{1}{2}C
\sum_{Q=1}^{P}
|\psi(r,\theta)|^{P-Q}\nonumber\\&\bigg\llbracket\bigg[\bigg(\frac{1}{2\pi}\bigg(
\tan^{-1}\bigg(\frac{|R+r|\tan(-\tfrac{1}{2}\theta)}{|R-r|}\bigg)-\frac{1}{2\pi}
\tan^{-1}\bigg(\frac{|R+r|\tan(\pi-\tfrac{1}{2}\theta)}{|R-r|}\bigg)\bigg)\bigg]^{Q}
\bigg\rrbracket
\end{align}
and only even moments are nonzero.
\end{proof}
}
\end{document}